\documentclass{article}
\usepackage{amsmath, amssymb, graphicx, url, amsthm}
\usepackage{caption, subcaption}
\usepackage{authblk}
\usepackage{xcolor}
\usepackage{float}
\usepackage{tikz}
\usepackage{color-edits}
\addauthor{vs}{red}
\addauthor{jc}{blue}
\usepackage{bbm}
\usepackage{hyperref}
\usepackage{cleveref}
\usepackage{booktabs}
\usepackage{pgfplots}
\usepgfplotslibrary{fillbetween}
\usetikzlibrary{patterns, intersections}

\usepackage[margin=1in]{geometry}
\newcommand{\E}{\mathbb{E}}
\newcommand{\R}{\mathbb{R}}
\newcommand{\mcG}{{\mathcal G}}
\newcommand{\mcT}{{\mathcal T}}
\def\trans{\intercal}
\newcommand{\ci}{\perp\!\!\!\perp}
\DeclareMathOperator*{\softmax}{\ensuremath{\mathrm{smax}^\beta}}
\newtheorem{assumption}{Assumption}
\newtheorem{theorem}{Theorem}
\newtheorem{remark}{Remark}
\newtheorem{lemma}[theorem]{Lemma}
\newtheorem{definition}{Definition}
\newtheorem{corollary}[theorem]{Corollary}

\title{Inference on Optimal Dynamic Policies via Softmax Approximation}
\usepackage{times}
\title{Inference on Optimal Dynamic Policies via Softmax Approximation}
\author[1]{Qizhao Chen}
\author[2]{Morgane Austern}
\author[3]{Vasilis Syrgkanis}
\affil[1,2]{Department of Statistics, Harvard University}
\affil[3]{Department of Management Science and Engineering, Stanford University}
\date{}

\pgfplotsset{compat=1.18} 
\begin{document}

\maketitle

\begin{abstract}
Estimating optimal dynamic policies from offline data is a fundamental problem in dynamic decision making. In the context of causal inference, the problem is known as estimating the optimal dynamic treatment regime. Even though there exists a plethora of methods for estimation, constructing confidence intervals for the value of the optimal regime and structural parameters associated with it is inherently harder, as it involves non-linear and non-differentiable functionals of unknown quantities that need to be estimated. Prior work resorted to sub-sample approaches that can deteriorate the quality of the estimate. We show that a simple soft-max approximation to the optimal treatment regime, for an appropriately fast growing temperature parameter, can achieve valid inference on the truly optimal regime. We illustrate our result for a two-period optimal dynamic regime, though our approach should directly extend to the finite horizon case. Our work combines techniques from semi-parametric inference and $g$-estimation, together with an appropriate triangular array central limit theorem, as well as a novel analysis of the asymptotic influence and asymptotic bias of softmax approximations.
\end{abstract}

\section{Introduction}

In most real world decision making settings, subjects undergo multiple exposures to some treatment or multiple treatments over time; patients receive multiple therapies, plots of lands are treated with multiple seeds, digital economy users are marketed by multiple campaigns. A typical problem that arises frequently: having access to large amounts of data where many subjects underwent multiple treatments sequentially, based on some naturally occurring and unknown treatment policy, can we identify what would have been the optimal dynamic treatment policy. Moreover, can we estimate with confidence the value of that optimal policy and understand whether we can reap large benefits (with statistical significance) if we intervene and change the status quo. The problem is typically termed as off-policy optimization or off-policy reinforcement learning and has also has garnered attention in the field of precision medicine. However, typical approaches in off-policy reinforcement learning primarily focus on getting fast statistical learning rates (see e.g. \cite{shi2022statistically, gunn2022adaptive, song2017semiparametric, kallus2020} for recent work), but not uncertainty quantification and construction of confidence intervals, which have been mostly studied in off-policy evaluation but not optimization (see e.g. \cite{pmlr-v139-karampatziakis21a} for recent work). Our aim is to make a step in filling this gap and address uncertainty quantification and confidence interval construction for off-policy optimization in reinforcement learning settings.

We consider a setting where we have collected data from multiple realizations of a two-period Markovian dynamic decision making process. We focus on two periods for simplicity of exposition, even though our techniques can easily be extended to multiple periods. The data stem from application of an unknown observational dynamic treatment policy on a sampled set of treated units. We assume we have access to n i.i.d. samples, one for each treated unit, of the form $Z=(S,T_1, X, T_2, Y)$, where $S$ is an initial state, $T_1$ a first period treatment, $X$ a second period state, $T_2$ a second period treatment, and $Y$ a final observed outcome of interest. We add subscripts $i$ to denote the realizations of $i$th data point $Z_i= (S_i,T_{1,i}, X_i, T_{2,i}, Y_i), i\in [n].$ The data is assumed to adhere to the causal graph depicted in Figure~\ref{fig:cg}. Moreover, the treatments $T_1, T_2\in \mathcal{T}$ are discrete and take $K$ possible values, for a constant $K$. We note that our assumption that the decision process is Markovian is without loss of generality, since the state $X$ is of arbitrary dimension and is allowed to capture all past history (for instance, we can set $X=(T_1, S, S')$ for some raw second period state $S'$). This formalism helps simplify notation.
\begin{figure}[H]
\centering 
\begin{subfigure}[t]{0.45\textwidth}
    \hspace*{-3cm}
    \centering
    \begin{tikzpicture} \large
    \node (S) at (0,0) {$S$};
    \node (T1) at (1.5, 1.5) {$T_1$};
    \node (X) at (3,0) {$X$};
    \node (T2) at (4.5,1.5) {$T_2$};
    \node  (Y) at (6,0) {$Y$};
    \draw[thick, ->] (S) -- (T1);
    \draw[thick, ->] (S) -- (X);
    \draw[thick, ->] (T1) -- (X);
    \draw[thick, ->] (X) -- (Y);
    \draw[thick, ->] (X) -- (T2);
    \draw[thick, ->] (T2) -- (Y);
    \hspace{10mm}
    \end{tikzpicture}
    \hspace{-3cm}
    \caption{Causal graph for observed data}
    \label{fig:cg}
\end{subfigure}
~~~~
\begin{subfigure}[t]{0.45\textwidth}
    \centering
    \begin{tikzpicture} \large
    \node (S) at (0,0) {$S$};
    \node (T1) at (1, 1.5) {$T_1$};
    \node (pi1) at (2.5, 1.5) {$\pi_1(S)$};
    \node (X) at (3.5,0) {$X^{(\pi)}$};
    \node (T2) at (4.5,1.5) {$T_2$};
    \node (pi2) at (6,1.5) {$\pi_2(X^{(\pi)})$};
    \node  (Y) at (6.5,0) {$Y^{(\pi)}$};
    \draw[thick, ->] (S) -- (T1);
    \draw[thick, ->] (S) -- (X);
    \draw[thick, ->, dashed] (S) to[bend right] (pi1);
    \draw[thick, ->] (pi1) -- (X);
    \draw[thick, ->] (X) -- (Y);
    \draw[thick, ->] (X) -- (T2);
    \draw[thick, ->, dashed] (X) to[bend right] (pi2);
    \draw[thick, ->] (pi2) -- (Y);
    \end{tikzpicture}
    \caption{Intervention graph (SWIG) of counterfactual outcome under alternative adaptive policy $\pi$}
    \label{fig:swig}
\end{subfigure}
\caption{Causal graphs describing observed and counterfactual data under alternative policy $\pi$}
\end{figure}
Our goal is to estimate and construct a confidence interval for the best dynamic policy in this setting. In particular, if we intervene and instead of the observed dynamic policy we deploy an alternative dynamic policy $\pi$ (depicted in the single world intervention graph (SWIG) in Figure~\ref{fig:swig}), then we will observe counterfactual or potential outcomes $Y^{(\pi)}$ and our goal is to find the policy $\pi^*$ that optimizes the mean counterfactual reward and construct a confidence interval for its value $V^*$ defined as:
\begin{align*}
    \pi^* :=~& \arg\max_{\pi} \E\{Y^{(\pi)}\}, & 
    V^* :=~& \E\{Y^{(\pi^*)}\}.
\end{align*}

The majority of prior work in off-policy evaluation and optimization in reinforcement learning has focused on the estimation of good policies with small regret. However, the focus of our work is in the ability to construct confidence intervals for the value of the optimal policy. Constructing confidence intervals is important in many high-stakes domains where we want to understand from offline data whether some candidate optimal policy will produce a positive improvement with statistical significance. If we are not confident on the magnitude of the improvement that the optimal regime will bring, then most decision makers would go with the status quo policy. Hence, construction of valid confidence intervals is important in high-stakes decision making.

Most prior works that provide confidence interval construction is focused on policy evaluation for a particular fixed policy as opposed to the optimal dynamic policy. Constructing confidence intervals for the optimal dynamic policy is an inherently harder problem as the target quantity tends not to be a smooth function of the distribution. The closest prior work to ours is that of \cite{chakraborty2010inference} which also addresses inference on structural parameters associated with the dynamic optimal policy, but resorts to sub-sampling techniques that can potentially deteriorate the estimation quality of the point estimate. Instead, we take the route of using smooth approximations to the optimal policy and show that the level of smoothness can be tuned appropriately to obtain correct confidence intervals for the truly optimal (non-smooth) policy. Our smooth approximation can be thought as an analogue to soft-Q-learning \cite{schulman2017equivalence,nachum2017bridging,haarnoja2017reinforcement,Haarnoja2018}, but in the context of the G-estimation framework proposed by \cite{robins2004optimal}, which we elaborate in Section~\ref{sec: identiG}.

\section{Related Literature}The problem of learning the optimal treatment rule
has been well studied, notably in the field of precision medicine (\cite{zhao2012estimating, zhang2012robust, fan2017concordance, song2017semiparametric, gunn2022adaptive}). This work has focused on finding the optimal treatments when there is only one time period by optimizing estimators of variants of the value functions or contrasts of value functions of treatment decision rules. For the dynamic treatment regime with multiple treatment stages, there have also been efforts to develop Q-learning (\cite{watkins1989learning, watkins1992q, zhao2009reinforcement, zhao2011reinforcement, qian2011performance}) and A-learning ((\cite{murphy2003optimal, robins2004optimal, moodie2007demystifying, lu2013variable})) type algorithms for the estimation of the optimal dynamic treatment regime. A-learning methods are semi-parametrically efficient and tend to have better performance than Q-learning methods when Q-functions are misspecified (\cite{chakraborty2010inference, qian2011performance, shi2018high}). For example, \cite{shi2022statistically} proposes an advantage learning process for offline infinite horizon settings to boost statistical efficiency in the estimation of the optimal treatment regime. They provide a finite sample bias guarantee for their process and show that the estimated contrast of Q-functions from their process converges faster than existing Q-learning estimation methods as the number of stages approaches infinity. However, these works focus on improving statistical efficiency rather than 
constructing confidence intervals for quantities related to the optimal policy. 

In the literature of optimal treatment estimation under dynamic treatment regimes, most estimators are in general not asymptotically normal (\cite{robins2004optimal,moodie2010estimating,chakraborty2010inference,song2015penalized}). This is particularly due to the fact that the first-stage pseudo-outcomes are not differentiable with respect to the second-stage structural parameters. 
This phenomenon also appears in other subfields such as in supervised classification when classes are not well separated (e.g \cite{laber2011adaptive}), where the lack of differentiability of the $0$-$1$ loss can lead to non-asymptotic normality. To address this problem, different solutions have been proposed. A first solution consists of adapting the estimation procedure by using a hard threshold estimator or a soft threshold estimator (e.g see \cite{chakraborty2010inference,moodie2012q}). These estimators shrink down the value of the second-stage structural parameters when close to the regions of non-regularity. 
Although it has been empirically shown to perform well (e.g see \cite{moodie2012q}), this approach suffers from a few drawbacks. Firstly, these threshold estimators often depend on hyperparameters that need to be tuned.
 Moreover, as noted in \cite{chakraborty2010inference}, theoretical understanding of these estimators is still lacking. 
Another solution proposes the construction of confidence intervals not through asymptotic normality but instead by relying on adaptive confidence intervals or adaptive bootstrap procedures (\cite{berger1994p, robins2004optimal, laber2011adaptive, laber2014dynamic, huang2015characterizing}). 
For example, \cite{laber2014dynamic} proposes to build different confidence intervals in, respectively, the regions of regularity and non-regularity. In the latter case, they exploit a union-bound argument over all possible values of the second-stage parameters. However, this approach offers conservative intervals that overcover and is difficult to generalize to multiple stages without having the confidence region size become excessively large. 
Another proposed solution uses the ``m out of n bootstrap" procedure (\cite{bretagnolle1983lois, swanepoel1986note, shao1989general, bickel2012resampling}). 
\cite{chakraborty2013inference} proposes to choose the size $m$ adaptively and derives theoretical guarantees for the obtained estimators. However, note that as is chosen to be much smaller than the sample size, i.e. $m=o(n)$, the proposed estimator will not converge at the optimal rate but at a slower rate. In contrast, our estimation algorithm uses a simple softmax approximation technique and we obtain asymptotic coverage guarantees as long as the log-ratio of the temperature parameter in the softmax operator to the sample size lies in a simple range. Even though an alternative softmax operator has also been proposed (\cite{asadi2017alternative}) with theoretically better convergence properties that have been employed in other statistical problems, we adhere to the widely adopted Boltzmann softmax operator, demonstrating that it still allows for valid inference for our problem of study.

 Semi-parametric approaches are frequently employed for estimating optimal dynamic treatment regimes (e.g \cite{robins2004optimal,murphy2003optimal}). These estimators are built using the generalized method of moments (\cite{levit1976efficiency, hasminskii1979nonparametric, ibragimov1981statistical, pfanzagl1982lecture, klaassen1987consistent, robinson1988root, bickel1988estimating, vaar:1991, bickel1993efficient, newey1994asymptotic, robins1995semiparametric, van2000asymptotic, newey1998undersmoothing, ai2003efficient, newey2004twicing, ai2007estimation, ai2012semiparametric, kosorok2007introduction,tsiatis2007semiparametric}). The nuisance parameter is often high-dimensional or infinite-dimensional and in general the difficulty in estimating nuisance parameters can negatively affect the estimation of the target parameter. To solve this problem, the efficient semi-parametric inference literature imposes conditions on the moment function under which errors made in the estimation of the nuisance parameters have a reduced impact on the estimation of the target parameter (\cite{hasminskii1979nonparametric,bickel1988estimating,zhang2014confidence,belloni2011inference,belloni2014inference,belloni2014uniform,belloni2014pivotal,javanmard2014confidence,javanmard2014hypothesis,javanmard2018debiasing,van2014asymptotically,ning2017general,chernozhukov2015valid,neykov2018unified,ren2015asymptotic,jankova2015confidence,jankova2016confidence,jankova2018semiparametric,bradic2017uniform,zhu2017breaking,zhu2018linear}). This approach dates back to the classical work on doubly
robust estimation and targeted maximum likelihood (\cite{robins1995analysis, robins1995semiparametric, van2006targeted, van2011targeted, luedtke2016statistical, toth2016tmle}) as well as to more recent
work on locally robust or Neyman orthogonal conditions on the moment function (\cite{chernozhukov2022locally,belloni2017program,chernozhukov2018double}).
 Debiased machine learning uses Neyman orthogonal moments, which satisfy that the first partial derivatives with respect to the nuisance parameters are zero. When the nuisance parameters live in some low complexity class (e.g. 
 Donsker class) or when cross-fitting is used (\cite{newey2018cross}), the target parameter can then be estimated at parametric rates. Neyman orthogonal moments can be constructed using the concept of a nonparametric influence function (\cite{chernozhukov2022locally, newey1998undersmoothing, newey2004twicing, bravo2020two}) or the related notion of the Riesz representer (\cite{chernozhukov2020adversarial, chernozhukov2022automatic, chernozhukov2022debiased, robins2007comment, newey2018cross, athey2018approximate, hirshberg2018debiased, hirshberg2017augmented, rothenhausler2019incremental, singh2019biased}). 
 Building on the work of \cite{chernozhukov2018double}, \cite{lewis2020double} showcased the use of the double/debiased machine learning framework to dynamic treatment effect models, specifically Structural Nested Mean Models, by leveraging a sequential residualization approach. 
 In our paper, we extend the debiased machine learning approach to estimate the structural parameters of the \emph{optimal} dynamic treatment regime. 
 We establish asymptotic normality of the value of the optimal treatment policy and of the structural parameters for both the first and second periods. As is standard in the debiased machine learning literature, our results require that the nuisance function spaces exhibit low complexity (with a small critical radius) or that cross-fitting is employed.






\section{Identification of Optimal Dynamic Policies via G-Estimation}\label{sec: identiG}

First observe that by invoking the Markovianity of the policy and the conditional independencies implied by the intervention graphs, we can characterize the optimal dynamic policy in a backwards induction manner (see Appendix~\ref{app:charac} for a formal argument):
\begin{align}\label{eqn:opt-pi-char}
\pi_1^*(S) =~& \arg \max_{\tau_1} \E\left\{Y^{(\tau_1, \pi_2^*)}\mid S\right\},&
\pi_2^*(X) =~& \arg\max_{\tau_2} \E\left\{Y^{(T_1, \tau_2)}\mid X\right\}.
\end{align}
In essence, we go to the last period and we optimize our second period action given the second period state $X$, without any change to our first period policy. This provides the optimal second period policy $\pi_2^*$. Then we go to the first period and we fantasize that we are continuing with the second period policy that we have just constructed and under this fantasy, we optimize our first period action given the first period state.

A well-known method for estimating an optimal dynamic regime is $g$-estimation \cite{robins2004optimal}, which is a form of backwards dynamic programming or $Q$-learning and closely related to advantage or $A$-learning \cite{murphy2003optimal,blatt2004learning,Moodie2007,schulte2014q}. This process generalizes to arbitrary number of periods, but we describe it here for simplicity in the two period case. The key argument in $g$-estimation is that the improvement that any policy $\pi$ brings, as compared to the observed policy, can be decomposed as the sum of a sequence of improvements, one for each decision-making period. Each of these improvement corresponds to removing the effect of the treatment that was given at that period and adding the effect of the treatment that would have been assigned by policy $\pi$ at that period. For such a decomposition, we need an appropriate notion of effect of each treatment. As a side note, appropriate notions of such an effect are also potentially useful in \emph{credit assignment}, i.e. attributing parts of the outcome to an assigned treatment; hence estimating these effect functions is interesting in its own right.

More formally, we can write (see Appendix~\ref{app:blip} for details on the derivation):
$$
    \E\left\{Y^{(\pi)} - Y\right\} = \E\left\{\gamma_1^{(\pi_2)}(\pi_1(S), S) - \gamma_1^{(\pi_2)}(T_1, S)\right\} + \E\left\{\gamma_2(\pi_2(X), X) - \gamma_2(T_2, X)\right\}.
$$

The functions $\gamma_t$ are known in the literature as the blip effects and correspond to the following quantity: what is the increase in reward if at the current moment we switch from the baseline treatment to some other treatment level $\tau$, and then continue with the target policy $\pi$. More formally we define:
\begin{align*}
    \gamma_2(\tau_2, x) =~& \E\left\{Y^{(T_1, T_2)} - Y^{(T_1, 0)}\mid T_2=\tau_2, X=x\right\},\\
    \gamma_1^{(\pi_2)}(\tau_1, s) =~& \E\left\{Y^{(T_1, \pi_2)} - Y^{(0, \pi_2)}\mid T_1=\tau_1, S=s\right\}.
\end{align*}
We will denote with $\gamma_t^*$ the blip effects that correspond to the optimal target policy (note that this only matters for the definition of $\gamma_1$).

If we had some way of estimating the blip effect functions for any target policy, then this decomposition allows us to construct the optimal dynamic regime in a recursive manner. We go to the second period and we note that under the conditional independencies implied by the intervention graphs, one can easily argue that we can equivalently define the optimal policy $\pi_2^*$ from Equation~\eqref{eqn:opt-pi-char} as the policy that maximizes the blip effect $\gamma_2=\gamma_2^*$, i.e.
\begin{align*}
    \pi_{2}^*(X) = \arg\max_{\tau_2} \gamma_2^*(\tau_2, X).
\end{align*}

We can then use this policy as the continuation policy when we define the blip effect for the first period, which we denote with $\gamma_1^*$ (as it is the blip effect that is associated with the optimal target policy). Again, under the conditional independencies implied by the intervention graphs, we can easily argue that the optimal first period policy $\pi_1^*$, as defined in Equation~\eqref{eqn:opt-pi-char} is the one that optimizes the first period blip effect, i.e.
\begin{align*}
    \pi_{1}^*(S) = \arg\max_{\tau_1} \gamma_1^*(\tau_1, S)
\end{align*}

Note that the conditional expectations that define the blip effects involve counterfactual quantities and therefore they don't correspond to simple regression problems on the observed data. The ingenuity of the $g$ estimation approach is that it provides a way to estimate the blip functions by noticing that they must satisfy a set of conditional moment restrictions. In particular, one can roughly make the following arguments \cite{robins2004optimal,lewis2020double}: if we subtract from $Y$ the blip effect of the observed last treatment, then the resulting random variable can roughly be thought as the counterfactual outcome $Y^{(T_1, 0)}$, which should be independent of $T_2$ conditional on $X$. This then implies that for all functions $f$:
\begin{align*}
    \E\left\{\left[Y - \gamma_2^*(T_2, X)\right]\, \left[f(T_2, X) - \E\{f(T_2, X)\mid X\}\right]\right) = 0. 
\end{align*}

This continuum of moment conditions can be used to identify the function $\gamma_2^*$. Similarly, if we subtract from the observed outcome the blip effect of the last period, add the blip effect of the optimal second period treatment and subtract the blip effect of the first period treatment, then the resulting variable can roughly be thought as the counterfactual outcome $Y^{(0, \pi_2^*)}$, which is independent of $T_1$ given $S$. This then implies that for all functions $f$:
\begin{align*}
    \E\left\{\left[Y - \gamma_2^*(T_2, X) + \gamma_2^*(\pi_2^*(X), X) - \gamma_1^*(T_1, S)\right]\,\left[f(T_1, S) - \E\{f(T_1, S)\mid S\}\right]\right\} = 0.
\end{align*}
This continuum of moment conditions can be used to identify the function $\gamma_1^*$.

This approach becomes very practical if one imposes semi-parametric restrictions on the data generating process that imply that the blip functions admit a parametric form. In particular, we will make a typical assumption that the blip effects that correspond to the optimal target policy are linear in some sufficiently expressive feature map (such a modeling assumption is referred to in the literature as a Linear Structural Nested Mean Model; SNMM \cite{hernan2010causal}), i.e. for some known $d$-dimensional feature map vectors $\phi, \mu$: 
\begin{align*}
    \gamma_2^*(\tau_2, x) =~& \psi_0 ^{\trans}\phi(\tau_2, x),& \gamma_1^*(\tau_1, s) =~& \theta_0 ^{\trans}\mu(\tau_1, s)
\end{align*}
with the convention that $\phi(0, x)=\mu(0,s)=0$. Then we can identify the structural parameters $\psi_0$ by simply solving for $\psi$ the following  vector of moment restrictions:
\begin{align*}
    \E(\epsilon_2(\psi)\, [\phi(T_2, X) - \E\{\phi(T_2, X)\mid X\}]) =~& 0, &
    \epsilon_2(\psi) =~& Y - \psi ^{\trans}\phi(T_2, X)
\end{align*}
then calculating the optimal second period policy:
\begin{align*}
    \pi_{2}^*(X) = \arg\max_{\tau_2} \psi_0 ^{\trans}\phi(\tau_2, X).
\end{align*}

Similarly, we can define the residual outcome that emulates $Y^{(0,\pi_2^*)}$ as:
\begin{equation}\label{eqn:period1-res}
    \epsilon_1(\theta, \psi) = Y - \psi ^{\trans}\phi(T_2, X) + \max_{\tau_2} \{\psi ^{\trans}\phi(\tau_2, X)\} - \theta ^{\trans}\mu(T_1, S)
\end{equation}
and calculate the first period structural parameters $\theta_0$ by solving for $\theta$ the following vector of moment restrictions:
\begin{equation}\label{eqn:non-ortho-non-smooth}
    \E\left\{\epsilon_1(\theta, \psi_0)\, [\mu(T_1, S) - \E\{\mu(T_1, S) \mid S\}]\right\} = 0,
\end{equation}
and calculate the optimal first period policy:
\begin{align*}
    \pi_1^*(S) = \arg\max_{\tau_1} \theta_0 ^{\trans}\mu(\tau_1, S)
\end{align*}
If we denote with:
\begin{align}\label{eqn:adjusted-y}
    Y_{\text{adjusted}}^* := Y - \psi_0 ^{\trans}\phi(T_2, X) - \theta_0 ^{\trans}\mu(T_1, S) + \max_{\tau_2} \psi_0 ^{\trans}\phi(\tau_2, X) + \max_{\tau_1} \theta_0 ^{\trans}\mu(\tau_1, S)
\end{align}
then value of the optimal policy is then given as:
\begin{align*}
    V^* = \E\{Y^{(\pi^*)}\} = \E\{Y_{\text{adjusted}}^*\}.
\end{align*}

\section{Challenges in Inference on Optimal Value and Structural Parameters}

We can easily translate the identification recipe in the previous section to an estimation strategy. In particular, we can replace the moment equations with their empirical analogues, estimate regression models for the conditional expectations of the feature maps and construct estimates $\hat{\psi}$ and $\hat{\theta}$ for the structural parameters using $n$ samples. There are two caveats with this approach. 

The first is that the estimation error of the regression functions $\E\{\phi(T_2,X)\mid X\}$ and $\E\{\mu(T_1,S)\mid S\}$, will have a first order impact on the accuracy of our structural parameters and hence the quality of the resulting optimal policies. Moreover, if we want to use regularization based methods to flexibly estimate these regressions with modern regression approaches, then the resulting estimates $\hat{\theta}, \hat{\psi}$ will be heavily biased, due to regularization bias. However, we can easily address this issue by resorting to existing techniques in the literature. In particular, instead of considering the raw moments in the previous section, we will consider Neyman orthogonal variants of the moments (see e.g. \cite{robins2004optimal,lewis2020double}). For instance, for the second period parameter $\psi_0$, we solve for $\psi$ an empirical analogue of the moment equation:
\begin{align*}
    \E([\epsilon_2(\psi) - \E\{\epsilon_2(\psi)\mid X\}] [\phi(T_2, X) - \E\{\phi(T_2, X)\mid X\}]) =~& 0.
\end{align*}
Note that we can write:
\begin{align*}
    \E\{\epsilon_2(\psi)\mid X\} = \E\{Y\mid X\} - \psi^{\trans}\E\{\phi(T_2, X)\mid X\}.
\end{align*}
Thus, we can estimate this term by estimating the following two regression problems:
\begin{equation}\label{eqn:period2-nuisance}
h^*(X) =  \E(Y\mid X), \quad
r^*(X) =  \E\{\phi(T_2, X)\mid X\}
\end{equation}

This moment is the standard residual-on-residual regression moment appearing in the double/debiased machine learning literature. The extra centering terms we added in the first parenthesis, do not alter the solution, but decrease the variance of the moment. This leads to a more efficient estimation method and more importantly, as we will see formally, robust to small influence from estimation errors in the various regression models that appear in the moment. We can perform similar correction to the first stage moment. Before we do that, let us elaborate on the second caveat of the estimation approach, which is more crucial and the main topic of this work.

One route to constructing asymptotically valid confidence intervals is showing that the structural parameter estimates are asymptotically linear (i.e. asymptotically equivalent to a simple average of some fixed function, known as the influence function). Asymptotic linearity implies asymptotic normality which can be used to provide confidence intervals for the parameter and, coupled with the delta method, for any smooth functional of the structural parameters. However, neither the moment that estimates the first period structural parameter $\theta$ is smooth in $\psi$ nor the policy value is smooth in $\theta, \psi$. Both of these functionals contain the maximum operator, which introduces non-smoothness and leads to non-Gaussian asymptotic stochastic behavior and inability to construct confidence intervals.

The main idea of our work is that by replacing the maximum operator with the soft-maximum operator with an appropriate temperature parameter, growing at the right rate, we can obtain valid inference on the truly optimal policy and its corresponding structural parameters.
The key bottleneck of the problem is analyzing the asymptotic linearity of the first period parameter estimator $\hat\theta$, which will be our main theorem. Inference on the structural parameters is also interesting in its own right, for instance when one is interested in understanding dimensions of heterogeneity of the effect of the treatment in the first period, under an optimal continuation policy. Building on this step, in Section~\ref{sec: aysmlinval} we present asymptotic normality and inference on the policy value, which, in essence, is a repetition of the soft-max analysis we present for inference on the first stage structural parameter.

\section{Softmax Approximation} 

We will define a softmax approximation to the moment conditions in Equation~\eqref{eqn:non-ortho-non-smooth} that define $\theta_0$. To do so, it will be convenient to introduce a soft-max operator:
\begin{align*}
    \softmax_{\tau \in \mcT} f(\tau) = \sum_{\tau\in \mcT} \frac{\exp\{\beta f(\tau)\}}{\sum_{t\in\mcT} \exp\{\beta f(t)\}} f(\tau).
\end{align*}
We can then define the softmax analogue of the residual $\epsilon_1(\theta, \psi)$ defined in Equation~\eqref{eqn:period1-res} as:
\begin{align*}
    \epsilon_1^\beta(\theta,\psi) = Y - \psi ^{\trans}\phi(T_2, X) + \softmax_{\tau_2\in \mcT} \psi ^{\trans}\phi(\tau_2, X) - \theta ^{\trans}\mu(T_1, S) 
\end{align*}
Then the softmax approximation to the Neyman orthogonal moment condition that defines the parameter $\theta_0$ can be written as:
\begin{align*}
    \E([\epsilon_1^\beta(\theta, \psi_0) - \E\{\epsilon_1^\beta(\theta, \psi_0)\mid S\}][ \mu(T_1, S) - \E\{\mu(T_1, S)\mid S\}]) = 0.
\end{align*}

We will denote with $\theta_0^\beta$ the solution to this softmax moment condition. 
Note that the term $\E\{\epsilon_1^\beta(\theta, \psi_0)\mid S\}$ can be decomposed as:
\begin{align*}
    \E\{Y\mid S\} - \E\{\psi_0 ^{\trans}\phi(T_2, X) - \softmax_{\tau_2\in \mcT}\psi_0 ^{\trans}\phi(\tau_2, X)\mid S\}- \theta ^{\trans}\E\{\mu(T_1, S)\mid S\}.
\end{align*}
Therefore we can estimate this term by estimating the following three regression problems:
\begin{equation}\label{eqn:period1-nuisance}
\begin{aligned}
q^*(S) =\E\{Y\mid S\},\quad
p_1^*(S) = \E\{\mu(T_1, S)\mid S\}, \\ 
p_{2,\beta}^*(S)= \E\{\psi_0 ^{\trans}\phi(T_2, X) - \softmax_{\tau_2\in \mcT} \psi_0 ^{\trans}\phi(\tau_2, X)\mid S\}.
\end{aligned}
\end{equation}
It will be convenient to denote with $m_\beta(Z;\theta, \psi, q, p_1, p_2)$ the softmax approximation to the Neyman orthogonal moment, parameterized by the aforementioned regression functions, i.e.
\begin{equation}\label{eqn:apx-moment-def}
     m_\beta(Z;\theta, \psi, q, p_1, p_2) = \{\epsilon_1^\beta(\theta, \psi) - q(S) + p_2(S) + \theta ^{\trans}p_1(S) \}\{\mu(T_1, S) - p_1(S)\}.
\end{equation}
and its expected value as:
\begin{align*}
    M(\theta, \psi, q, p_1, p_2;\beta) =~& \E\{m_\beta(Z;\theta, \psi, q, p_1, p_2)\}.
\end{align*}
For any $\beta\in (0, \infty]$, let $\theta_0^\beta$ be the solution to the following moment condition, with respect to $\theta$: 
$$M(\theta, \psi_0, q^*, p_1^*, p_{2,\beta}^*;\beta) = 0.$$
Note that we have
$\theta_0 = \theta_0^\infty.$

\section{Estimation Procedure and Main Theorem}\label{sec:mainthm}
First, we define our estimation process, which consists of two steps of backwards induction.

\paragraph{Step 1: Procedure for estimate $\hat{\psi}$ of second period structural parameter $\psi_0$.} 
\begin{tabbing}
   \qquad \enspace Construct estimates $\hat{h}, \hat{r}$ of the nuisance functions $h^*$, $r^*$ defined in Equation~\eqref{eqn:period2-nuisance}.\\
   \qquad \enspace For $i=1$ to $i=n$ \\
   \qquad \qquad Define residuals $\check{Y}_i=Y_i-\hat{h}(X_i)$ and $\check{\Phi}_i=\phi(T_{2,i},X_i) - \hat{r}(X_i)$ .\\
\qquad \enspace Obtain $\hat{\psi}$ by solving the empirical moment equation with respect to $\psi$:
\end{tabbing}
\vspace{-1.8em}
\begin{align*}
    n^{-1} \sum_{i=1}^n (\check{Y}_i - \psi ^{\trans}\check{\Phi}_i) \check{\Phi}_i = 0.
\end{align*}

\paragraph{Step 2: Procedure for estimate $\hat{\theta}^\beta$ of first period structural parameter $\theta_0$}
\begin{tabbing}
   \qquad \enspace Construct estimates $\hat{q}, \hat{p}_1$ of the nuisance functions
$q^*, p_1^*$ defined in Equation~\eqref{eqn:period1-nuisance}.\\
\qquad \enspace Construct estimate $\hat{p}_{2,\beta} \equiv \hat{p}_{2,\beta, \hat{\psi}}$ by regressing $\hat{\psi} ^{\trans}\phi(T_2, X) - \softmax_{\tau_2} \hat{\psi} ^{\trans}\phi(\tau_2, X)$ on $S$. \\
   \qquad \enspace For $i=1$ to $i=n$ \\
   \qquad \qquad Define residuals: $\hat{Y}_i=Y_i - \hat{q}(S_i)$ and $\hat{M}_i=\mu(T_{1,i}, S_i) - \hat{p}_1(S_i)$.\\
   \qquad \qquad Define residual: $\hat{\Phi}_i = \hat{\psi} ^{\trans}\phi(T_{2,i}, X_i) - \softmax_{\tau_2} \hat{\psi} ^{\trans}\phi(\tau_{2}, X_i) - \hat{p}_{2,\beta}(S_i)$\\
\qquad \enspace Obtain $\hat{\theta}^\beta$ by solving the empirical moment equation with respect to $\theta$:
\end{tabbing}
\vspace{-1.8em}
\begin{align*}
    n^{-1}\sum_{i=1}^n (\hat{Y}_i - \hat{\Phi}_i - \theta ^{\trans}\hat{M}_i) \hat{M}_i = 0.
\end{align*}

  Having defined our estimation process, we are now ready to prove our main result, which shows that we can construct asymptotically valid confidence intervals for the structural parameters $\theta_0, \psi_0$, assuming the temperature parameter $\beta$ of the soft-max grows fast enough ($\omega(n^{1/\{2(1+\delta)\}})$  for some appropriately chosen $\delta\in(0,1]$), but not too fast ($o(n^{1/2})$). We further need additional assumptions on the estimation quality of all the nuisance functions involved in the estimation process and require that the nuisance estimates lie in function classes of small statistical complexity, e.g. with critical radius of $o_p(n^{-1/4})$, where the classical definition of critical radius of a function class can be found in e.g. Section 3 of \cite{chernozhukov2020adversarial}.

To present the theorem, we will also need to introduce the following norm notation. For any fixed or random vector $v$, we denote with $\|v\|_{p}=\{\sum |v_i|^p\}^{1/p}$  for any $p\ge 1$, and with $\|v\|_{\infty}=\max |v_i|$, where $v_i$ are the individual components of vector $v.$ Moreover, for any vector-valued function $g$ that takes as input a random variable $X$, we define $\|g\|_2 = [\E_X\{\|g(X)\|_2^2\}]^{1/2}$.


\begin{theorem}[Main Theorem I]\label{thm:main}
Assume that the random variables $\max_{t} \psi_0 ^{\trans}\phi(t,X) - \psi_0 ^{\trans}\phi(\tau, X)\ (\tau\in \mcT)$ are almost surely bounded and each admits a density $f_\tau$ on $(0,c)$ for some constant $c>0$, that satisfies $f_\tau(x)\leq H/x^{1-\delta}$, for some $0\leq H<\infty$ and $0<\delta\leq 1$. Suppose that $\beta=\omega(n^{1/\{2(1+\delta)\}})$ and $\beta=o(n^{1/2})$. Moreover, suppose that the nuisance estimates satisfy the rate conditions:
$$\|h^*-\hat{h}\|_{2}, \|r^* - \hat{r}\|_2, \|q^* - \hat q\|_2, \|p_1^* - \hat p_1\|_2=o_p(n^{-1/4})$$ 
Define $p_{2,\beta,\psi}^*(S)=\E[\psi^{\trans}\phi(T_2, X)-\softmax_{\tau_2} \psi^{\trans}\phi(\tau_2, X)\mid S]$ and $\hat{p}_{2,\beta,\psi}$ the outcome of the estimation algorithm when regressing $\psi^{\trans}\phi(T_2, X)-\softmax_{\tau_2} \psi^{\trans}\phi(\tau_2, X)$ on $S$. Assume that there exists a sufficiently small neighborhood $\mathcal{N}$ of $\psi_0$ such that for all $\epsilon>0$, as $n\to\infty$\footnote{We note that such a guarantee would follow easily, for instance, from results on oracle inequalities via localized complexities and a uniform cover argument over the low dimensional parametric space that the parameter $\psi$ lies in.}
\begin{align*}
\sup_{\beta>0}pr\left(\sup_{\psi\in\mathcal{N}}\|\hat p_{2,\beta,\psi} - p_{2,\beta,\psi}^*\|_2 \ge \epsilon n^{-1/4}\right)\to 0.
\end{align*}
Further, suppose that the nuisance estimates take values from a function space of
critical radius $\delta_n$ with $\delta_n=o_p(n^{-1/4})$. Assume the nuisance estimates $\hat q(S), \hat p_1(S), \hat p_{2,\beta}(S)$ are all almost surely bounded. Moreover, assume the following boundedness conditions:
$$\sup_{\beta>0}\|\theta_0^\beta\|_2<\infty$$
and assume the random variables
$$\sum_\tau \|\phi(\tau,X)\|_2, \|\tilde{M}\|_2,\tilde{Y},\sup_{\beta>0}\|p_{2,\beta}^*(S)\|_2,  \|\mu(T_1,S)\|_2, \sup_{\tau\in \mathcal{T}, \psi \in \mathcal{N}} |\psi^{\trans} \phi(\tau, X)|$$
are all almost surely bounded, where $\tilde{M}= \mu(T_1,S) - p_1^*(S)$, $\tilde{Y}= Y - q^*(S)$. Assume that the matrix $\E(\tilde{M} \tilde{M} ^{\trans})$ is bounded and strictly positive definite so that its inverse matrix $\E(\tilde{M} \tilde{M} ^{\trans})^{-1}$ exists. Then the estimates $\hat{\psi}, \hat{\theta}^\beta$ are asymptotically linear, i.e.
\begin{align*}
    n^{1/2}(\hat{\psi} - \psi_0) = n^{-1/2} \sum_{i=1}^n \rho_{\psi}(Z_i) + o_p(1),\quad
    n^{1/2}(\hat{\theta}^\beta - \theta_0) = n^{-1/2} \sum_{i=1}^n \rho_{\theta}(Z_i) + o_p(1)
\end{align*}
for some functions $\rho_\psi, \rho_\theta$ such that $E\{\rho_\psi(Z_i)\} = 0$ and $E\{\rho_\theta(Z_i)\} = 0$.
Moreover, asymptotically valid confidence intervals, for $\psi_0$ and $\theta_0$ respectively, with target coverage level $\alpha$, can be constructed via any consistent estimates $\hat{\sigma}_\psi^2, \hat{\sigma}_\theta^2$ of the variances $\sigma_\psi^2 = \E\{\rho_\psi(Z)^2\}$ and $\sigma_\theta^2 = \E\{\rho_\theta(Z)^2\}$ as:
\begin{align*}
    CI_{\psi}(\alpha) = [\hat{\psi} \pm n^{-1/2}z_{1-\alpha/2} \hat{\sigma}_{\psi}], \quad
    CI_{\theta}(\alpha) = [\hat{\theta}^\beta \pm n^{-1/2}z_{1-\alpha/2} \hat{\sigma}_{\theta}], 
\end{align*}
where $z_{q}$ is the $q$-th quantile of the standard normal distribution.
\end{theorem}

  The more expansive theorems that we define in subsequent sections also provide an exact form of the asymptotic influence functions and the asymptotic variances, which we omit in the main theorem for succinctness of exposition. Furthermore, we note that a simpler sufficient condition that implies the condition that the random variables $\{\max_{t} \psi_0 ^{\trans}\phi(t,X) - \psi_0 ^{\trans}\phi(\tau, X)\}_{\tau\in \mcT}$ each admits a density $f_\tau$, that satisfies $f_\tau(x)\leq H/x^{1-\delta}$, for some $0\leq H<\infty$ and $0<\delta\leq 1$, is that the variables $\{\psi_0 ^{\trans}\phi(\tau, X)\}_{\tau\in \mcT}$ admit a joint density, which is a very benign regularity assumption. This follows from the following lemma:

\begin{lemma}\label{lem:simpler}
   Let $(U_t)_{t\in \mathcal{T}}$ be a collection of real-valued random variables and write $U_{\rm{max}}=\max_{t\in \mathcal{T}}U_t$. Suppose that the random variables $(U_t)_{t\in \mathcal{T}}$ admit a joint density that is continuous at zero; then for all $t\in \mathcal{T}$, the probability measure of random variable $U_{\rm{max}}-U_t$ on $(0,\infty)$ is absolutely continuous with respect to the Lebesgue measure on $(0,\infty)$. Hence, in particular, this implies that each $U_{\rm{max}}-U_t$ admits a density that $f_\tau$ that satisfies $f_\tau(x)\leq H/x^{1-\delta}$ for some $0\leq H<\infty$ and $\delta=1.$
\end{lemma}

Furthermore (as we show in Appendix~\ref{app:cfit}), the constraint on the critical radius of the nuisance function spaces can be lifted if one performs nested cross-fitting at each step of the recursion, at the cost of reducing sample size when training each regression.

\section{Asymptotic linearity Theorem for the Optimal Policy Value}\label{sec: aysmlinval}
In previous sections, we have established the asymptotic linearity of the first-period and second-period structural parameter estimators $\hat\theta,\hat\psi$ and thus constructed confidence intervals for the corresponding parameters. In this section, we will further establish asymptotic linearity for an estimator of the optimal policy value and thus construct a confidence interval for the value.

Recall that the value of the optimal policy is identified as $V^* = \E( Y_{\text{adjusted}}^* )$, where $Y_{\text{adjusted}}^*$ is defined in Equation~\eqref{eqn:adjusted-y}. It will be helpful to define the softmax counterpart:
\begin{align*}
    V_\beta^* = \E\{ Y - \psi_0 ^{\trans}\phi(T_2, X) - \theta_0 ^{\trans}\mu(T_1, S) + \softmax_{\tau_2} \psi_0^{\trans}\phi(\tau_2, X) + \softmax_{\tau_1} \theta_0^{\trans}\mu(\tau_1, S)\}.
\end{align*}
Then consider the estimator $\hat{V}$, defined as:
$$\frac{1}{n}\sum_{i=1}^n \{Y_i - \hat\psi^{\trans}\phi(T_{2,i}, X_i) - (\hat\theta^\beta)^{\trans}\mu(T_{1,i}, S_i) + \softmax_{\tau_2}\hat\psi^{\trans}\phi(\tau_2, X_i) + \softmax_{\tau_1}(\hat\theta^\beta)^{\trans}\mu(\tau_1,S_i)\}.$$
\begin{theorem}
    \label{thm: value_var}
Assume that the following almost sure finiteness conditions hold: 
$$\sup_S\sum_{\tau\in\mcT}\|\mu(\tau,S)\|_2,\ \sup_{\beta>0}|V_{\beta}^*|,\ |Y|<\infty \mbox{ a.s.}$$
Then under the conditions of Theorem~\ref{thm:main}, we have that
    $$n^{1/2}(\hat V - V_\beta^*)=n^{-1/2}\sum_{i=1}^n\rho_{V,*}(Z_i) + o_p(1)$$
with influence function
\begin{equation}\label{eq: inf}
\begin{aligned}
    \rho_{V,*}(Z)=~Y_{\text{adjusted}}^* -V^* 
&~+ \E\{\phi_{\infty}(X) - \phi(T_2,X)\}^{\trans}\rho_\psi(Z)\\
&~+ \E\{\mu_{\infty}(S)-\mu(T_1,S)\}^{\trans}\rho_\theta(Z),
\end{aligned}
\end{equation}
where 
\begin{align*}
 \rho_\psi(Z)=~&\E(\tilde{P}\tilde{P}^{\trans})^{-1}\{Y-\E(Y\mid X) - \psi_0^{\trans}\tilde{P}\}\tilde{P},\\
    \rho_{\theta}(Z) =~&   \E(\tilde{M}\tilde{M}^{\trans})^{-1} \left\{m_*(Z; \theta_0, \psi_0,g_0) + J_*^{\trans}\rho_{\psi}(Z)\right\},\\
m_*(Z;\theta, \psi, g) =~& \{\epsilon_1(\theta, \psi) - q(S) + p_2(S) + \theta^\trans p_1(S) \}\, \{\mu(T_1, S) - p_1(S)\}, \\
J_* =~& \E\left\{\left(\phi_\infty(X) - \phi(T_2, X)\right)\tilde{M}^{\trans}\right\}.
\end{align*}
and $\tilde{M} = \mu(T_1, S) - \E\{\mu(T_1, S)\mid S\}$ and $\tilde{P}=\phi(T_2, X) - \E\{\phi(T_2, X)\mid X\}$ and, if we denote $\mathcal{M}(X)=\{\tau: \psi_0^{\trans}\phi(\tau, X) = \max_{t}\psi_0^{\trans}\phi(t, X)\}$ and ${\cal L}(S)=\{\tau: \theta_0^{\trans}\mu(\tau, S) = \max_{t}\theta_0^{\trans}\mu(t, S)\}$, the set of optimal actions at each state, then
\begin{align*}
    \phi_\infty(X)=~&|\mathcal{M}(X)|^{-1}\sum_{\tau\in\mathcal{M}(X)}\phi(\tau, X), &
    \mu_\infty(S)=|\mathcal{L}(S)|^{-1}\sum_{\tau\in\mathcal{L}(S)}\mu(\tau, S).
\end{align*}
\end{theorem}

Theorem~\ref{thm: value_var} implies that we can build confidence intervals for the value $V_\beta^*$. Moreover, the bias lemma that follows implies that we can build confidence intervals for the true optimal policy value $V^*$. For each treatment $\tau\in\mathcal{T}$, define the random variables
\begin{align*}
    U_\tau^{\psi_0}= \psi_0^{\trans}\phi(\tau,X),\quad U_\tau^{\theta_0}=\theta_0^{\trans}\mu(\tau,S).
\end{align*}
Taking the maxima of $U_\tau^{\psi_0}$ and $U_\tau^{\theta_0}$ over all treatments $\tau\in\mathcal{T}$, we define
\begin{align*}
    U_{\max}^{\psi_0}= \max_{t\in \mathcal{T}} \left\{\psi_0^{\trans}\phi(t,X)\right\},\quad {U}_{\max}^{\theta_0}= \max_{t\in \mathcal{T}} \left\{\theta_0^{\trans}\mu(t,S)\right\}.
\end{align*}
We will also need to make the following assumption:
\begin{assumption}\label{ass: density}
    Assume that there is a constant $c>0$ such that  for each treatment $\tau\in\mathcal{T}$, we have that $$\|U_\tau^{\psi_0}\|_{\infty},\|U_\tau^{\theta_0}\|_{\infty}<\infty$$ and that on $(0,c)$ the random variables $U_{\max}^{\psi_0} - U_{\tau}^{\psi_0}$, ${U}_{\max}^{\theta_0} - {U}_{\tau}^{\theta_0}$ admit densities $f_\tau$ and $\tilde{f}_\tau$ that satisfy $f_\tau(x)\leq H/x^{1-\delta}$ and $\tilde f_\tau(x)\leq \tilde H/x^{1-\tilde\delta}$, for some $0\leq H, \tilde H<\infty$ and $0<\delta,\tilde\delta\leq 1$. 
\end{assumption}

\begin{lemma}[Softmax Bias Control for Optimal Policy Value]\label{lem: value_bias}
Suppose that Assumption~\ref{ass: density} holds. Assume that $\|\tilde{M}\|_{2}$ is uniformly bounded and that $\E(\tilde{M}\tilde{M}^{\trans})$ is strictly positive definite.
If $\beta=\omega(n^{1/\{2(1+\min\{\delta,\tilde \delta\})\}})$ then we get that as $n\to \infty$:
    $$n^{1/2}(V_\beta^* - V^*) = o(1).$$
\end{lemma}

Combining Theorem~\ref{thm: value_var} and Lemma~\ref{lem: value_bias} gives us the following final result:
\begin{corollary}[Main Corollary]
    Under the conditions of Theorem~\ref{thm: value_var} and Lemma~\ref{lem: value_bias}, the estimator $\hat V$ is asymptotically linear around $V^*$:
    $$n^{1/2}(\hat V - V^*)=n^{-1/2}\sum_{i=1}^n\rho_{V,*}(Z_i) + o_p(1)$$
with influence function $\rho_{V,*}$ as defined in \eqref{eq: inf}. Letting $z_{q}$ be the $q$-th quantile of the standard normal distribution, asymptotically valid confidence intervals for $V^*$, with target coverage level $\alpha$, can be constructed via any consistent estimate $\hat{\sigma}_V^2$ of the variance $\sigma_V^2 = \E\{\rho_{V,*}(Z)^2\}$ as:
\begin{align*}
    CI_{V}(\alpha) =~& [\hat{V} \pm z_{1-\alpha/2} \hat{\sigma}_{V}/\sqrt{n}], 
\end{align*}
where $z_q$ is the $q$-th quantile of the standard normal distribution.
\end{corollary}

\section{Monte Carlo Experiments}
We investigate performance of our inference method on a simple data generating process and demonstrate the efficacy of our method and the importance of choosing the softmax parameter $\beta$, based on our theorem. We consider observational data from the structural equation model:
\begin{align*}
    S =~& N(0, I_5), &
    T_1 =~& \text{Bernoulli}(\text{Logistic}(S_1)),\\
    X =~& \alpha_1 T_1 + S + N(0, I_5), &
    T_2 =~& \text{Bernoulli}(\text{Logistic}(X_1)), \\
    Y =~& \alpha_2\, (X_1 + 1)\, T_2 + X_1 + N(0, 1).
\end{align*}
where $S_1, X_1$ denote the first components of the random vectors $S$ and $X$ respectively. As we show in Appendix~\ref{app:experiments}, the second period blip effect $\gamma_2$ is of the form:
\begin{align*}
    \gamma_2(T_2, X) =~& \alpha_2 (X_1 + 1) T_2,  &
    \phi(T_2, X) =~& (T_2, X_1 T_2), &
    \psi_0 = (\alpha_2, \alpha_2),
\end{align*}
while the blip effect in the first period, under the optimal continuation policy can be very well approximated as a linear function of the feature map $\mu(T_1, S) = (T_1, S_1 T_1)$. Moreover, $\alpha_1$ controls the magnitude of the blip effect of the first period treatment.

For each $n, \alpha_1, \alpha_2$ and softmax parameter $\beta$ specification, we run $100$ monte-carlo experiments and calculate coverage of the estimated 95\% and 90\% confidence intervals (and the standard error of the coverage) for the optimal policy value. We used Random Forest regression for estimating all the nuisance functions in the estimation process and no sample splitting.\footnote{Code for replicating this experiment can be found in this \href{https://colab.research.google.com/github/syrgkanislab/optimal_dynamic_regime/blob/main/FinalForPaperSoftmaxInference.ipynb}{Jupyter Notebook} from this \href{https://github.com/syrgkanislab/optimal_dynamic_regime}{Github Repository}.} We report results in Figures~\ref{fig:cov} and \ref{fig:cov2}. We chose pairs of $\alpha_1, \alpha_2$ that identify starkly different estimation regimes, thereby capturing different failure modes of poorly constructed confidence intervals due to irregularity of the optimal policy estimand. We see that $\beta\approx n^{1/4+\epsilon}$ (as indicated by Theorem~\ref{thm:main} with $\delta=1$) produces consistently good coverage above 90\% across all specifications, while $\beta$ that is either below $n^{1/4}$ or above $n^{1/2}$, undercovers in at least one specification (highlighted). For $n=10k$ the recommended specification always achieves almost nominal coverage (up to standard error), while the other two specifications fail blatantly in at least one specification. 

\begin{figure}[H]
\centering
    \scriptsize
\bgroup
\def\arraystretch{1}%
\setlength\tabcolsep{4pt}
\begin{tabular}{l|ll|ll|ll|ll}
\toprule
$\beta$ & \multicolumn{2}{l}{$(\alpha_1,\alpha_2)=(0, 0)$} & \multicolumn{2}{l}{$(\alpha_1,\alpha_2)=(0, 1)$} & \multicolumn{2}{l}{$(\alpha_1,\alpha_2)=(1, 0)$} & \multicolumn{2}{l}{$(\alpha_1,\alpha_2)=(1, 1)$} \\
{} &                     $n=1000$ &        $n=10000$ &                     $n=1000$ &        $n=10000$ &                     $n=1000$ &        $n=10000$ &                     $n=1000$ &        $n=10000$ \\
\midrule
$n^{0.15}$ &              $0.94 \pm 0.02$ &  $0.95 \pm 0.02$ &              $0.89 \pm 0.03$ &  $0.92 \pm 0.03$ &              $\mathbf{0.79 \pm 0.04}$ &  $\mathbf{0.82 \pm 0.04}$ &              $0.88 \pm 0.03$ &  $0.92 \pm 0.03$ \\
$n^{0.26}$ &              $0.91 \pm 0.03$ &  $0.95 \pm 0.02$ &              $0.90 \pm 0.03$ &  $0.94 \pm 0.02$ &              $0.90 \pm 0.03$ &  $0.97 \pm 0.02$ &              $0.90 \pm 0.03$ &  $0.92 \pm 0.03$ \\
$n^{0.75}$ &              $\mathbf{0.79 \pm 0.04}$ &  $0.92 \pm 0.03$ &              $0.87 \pm 0.03$ &  $0.95 \pm 0.02$ &              $0.95 \pm 0.02$ &  $0.96 \pm 0.02$ &              $0.91 \pm 0.03$ &  $0.92 \pm 0.03$ \\
\bottomrule
\end{tabular}
\egroup
\caption{Coverage of 95\% confidence intervals for optimal policy value.}
\label{fig:cov}
\end{figure}
\begin{figure}[H]
\centering
    \scriptsize
\bgroup
\def\arraystretch{1}%
\setlength\tabcolsep{3pt}
\begin{tabular}{l|ll|ll|ll|ll}
\toprule
$\beta$ & \multicolumn{2}{l}{$(\alpha_1,\alpha_2)=(0, 0)$} & \multicolumn{2}{l}{$(\alpha_1,\alpha_2)=(0, 1)$} & \multicolumn{2}{l}{$(\alpha_1,\alpha_2)=(1, 0)$} & \multicolumn{2}{l}{$(\alpha_1,\alpha_2)=(1, 1)$} \\
{} &                     $n=1000$ &        $n=10000$ &                     $n=1000$ &        $n=10000$ &                     $n=1000$ &        $n=10000$ &                     $n=1000$ &        $n=10000$ \\
\midrule
$\beta=n^{0.15}$ &              $0.85 \pm 0.04$ &  $0.89 \pm 0.03$ &              $0.87 \pm 0.03$ &  $0.82 \pm 0.04$ &              $\mathbf{0.68 \pm 0.05}$ &  $\mathbf{0.72 \pm 0.04}$ &              $0.82 \pm 0.04$ &  $0.86 \pm 0.03$ \\
$\beta=n^{0.26}$ &              $0.83 \pm 0.04$ &  $0.88 \pm 0.03$ &              $0.82 \pm 0.04$ &  $0.89 \pm 0.03$ &              $0.88 \pm 0.03$ &  $0.88 \pm 0.03$ &              $0.81 \pm 0.04$ &  $0.86 \pm 0.03$ \\
$\beta=n^{0.75}$ &              $\mathbf{0.69 \pm 0.05}$ &  $\mathbf{0.77 \pm 0.04}$ &              $\mathbf{0.74 \pm 0.04}$ &  $0.90 \pm 0.03$ &              $0.83 \pm 0.04$ &  $0.89 \pm 0.03$ &              $0.79 \pm 0.04$ &  $0.87 \pm 0.03$ \\
\bottomrule
\end{tabular}
\egroup
\caption{Coverage of 90\% confidence intervals for optimal policy value.}
\label{fig:cov2}
\end{figure}

\section*{Acknowledgement}
This work was partially supported by a 2023 Amazon Research Award.

\vspace*{-10pt}
\bibliography{refs}

\begin{thebibliography}{10}

\bibitem{ai2003efficient}
Chunrong Ai and Xiaohong Chen.
\newblock Efficient estimation of models with conditional moment restrictions containing unknown functions.
\newblock {\em Econometrica}, 71(6):1795--1843, 2003.

\bibitem{ai2007estimation}
Chunrong Ai and Xiaohong Chen.
\newblock Estimation of possibly misspecified semiparametric conditional moment restriction models with different conditioning variables.
\newblock {\em Journal of Econometrics}, 141(1):5--43, 2007.

\bibitem{ai2012semiparametric}
Chunrong Ai and Xiaohong Chen.
\newblock The semiparametric efficiency bound for models of sequential moment restrictions containing unknown functions.
\newblock {\em Journal of Econometrics}, 170(2):442--457, 2012.

\bibitem{asadi2017alternative}
Kavosh Asadi and Michael~L Littman.
\newblock An alternative softmax operator for reinforcement learning.
\newblock In {\em International Conference on Machine Learning}, pages 243--252. PMLR, 2017.

\bibitem{athey2018approximate}
Susan Athey, Guido~W Imbens, and Stefan Wager.
\newblock Approximate residual balancing: debiased inference of average treatment effects in high dimensions.
\newblock {\em Journal of the Royal Statistical Society Series B: Statistical Methodology}, 80(4):597--623, 2018.

\bibitem{belloni2017program}
Alexandre Belloni, Victor Chernozhukov, Ivan Fernandez-Val, and Christian Hansen.
\newblock Program evaluation and causal inference with high-dimensional data.
\newblock {\em Econometrica}, 85(1):233--298, 2017.

\bibitem{belloni2011inference}
Alexandre Belloni, Victor Chernozhukov, and Christian Hansen.
\newblock Inference for high-dimensional sparse econometric models.
\newblock {\em arXiv:1201.0220}, 2011.

\bibitem{belloni2014inference}
Alexandre Belloni, Victor Chernozhukov, and Christian Hansen.
\newblock Inference on treatment effects after selection among high-dimensional controls.
\newblock {\em The Review of Economic Studies}, 81(2):608--650, 2014.

\bibitem{belloni2014uniform}
Alexandre Belloni, Victor Chernozhukov, and Kengo Kato.
\newblock Uniform post-selection inference for least absolute deviation regression and other {Z}-estimation problems.
\newblock {\em Biometrika}, 102(1):77--94, 2014.

\bibitem{belloni2014pivotal}
Alexandre Belloni, Victor Chernozhukov, and Lie Wang.
\newblock Pivotal estimation via square-root lasso in nonparametric regression.
\newblock {\em The Annals of Statistics}, 42(2):757--788, 2014.

\bibitem{berger1994p}
Roger~L Berger and Dennis~D Boos.
\newblock P values maximized over a confidence set for the nuisance parameter.
\newblock {\em Journal of the American Statistical Association}, 89(427):1012--1016, 1994.

\bibitem{bickel2012resampling}
Peter~J Bickel, Friedrich G{\"o}tze, and Willem~R van Zwet.
\newblock {\em Resampling fewer than n observations: gains, losses, and remedies for losses}.
\newblock Springer, 2012.

\bibitem{bickel1993efficient}
Peter~J Bickel, Chris~AJ Klaassen, Ya'acov Ritov, and Jon~A Wellner.
\newblock {\em Efficient and adaptive estimation for semiparametric models}, volume~4.
\newblock Johns Hopkins University Press, 1993.

\bibitem{bickel1988estimating}
Peter~J Bickel and Yaacov Ritov.
\newblock Estimating integrated squared density derivatives: Sharp best order of convergence estimates.
\newblock {\em Sankhy{\=a}: The Indian Journal of Statistics, Series A}, pages 381--393, 1988.

\bibitem{blatt2004learning}
Doron Blatt, Susan~A Murphy, and Ji~Zhu.
\newblock A-learning for approximate planning.
\newblock {\em Ann Arbor}, 1001:48109--2122, 2004.

\bibitem{bradic2017uniform}
Jelena Bradic and Mladen Kolar.
\newblock Uniform inference for high-dimensional quantile regression: Linear functionals and regression rank scores.
\newblock {\em arXiv:1702.06209}, 2017.

\bibitem{bravo2020two}
Francesco Bravo, Juan~Carlos Escanciano, and Ingrid Van~Keilegom.
\newblock Two-step semiparametric empirical likelihood inference.
\newblock 2020.

\bibitem{bretagnolle1983lois}
J~Bretagnolle.
\newblock Lois limites du bootstrap de certaines fonctionnelles.
\newblock In {\em Annales de l'IHP Probabilit{\'e}s et statistiques}, volume~19, pages 281--296, 1983.

\bibitem{chakraborty2013inference}
Bibhas Chakraborty, Eric~B Laber, and Yingqi Zhao.
\newblock Inference for optimal dynamic treatment regimes using an adaptive m-out-of-n bootstrap scheme.
\newblock {\em Biometrics}, 69(3):714--723, 2013.

\bibitem{chakraborty2010inference}
Bibhas Chakraborty, Susan Murphy, and Victor Strecher.
\newblock Inference for non-regular parameters in optimal dynamic treatment regimes.
\newblock {\em Statistical methods in medical research}, 19(3):317--343, 2010.

\bibitem{chernozhukov2018double}
Victor Chernozhukov, Denis Chetverikov, Mert Demirer, Esther Duflo, Christian Hansen, Whitney Newey, and James Robins.
\newblock Double/debiased machine learning for treatment and structural parameters, 2018.

\bibitem{chernozhukov2022locally}
Victor Chernozhukov, Juan~Carlos Escanciano, Hidehiko Ichimura, Whitney~K Newey, and James~M Robins.
\newblock Locally robust semiparametric estimation.
\newblock {\em Econometrica}, 90(4):1501--1535, 2022.

\bibitem{chernozhukov2015valid}
Victor Chernozhukov, Christian Hansen, and Martin Spindler.
\newblock Valid post-selection and post-regularization inference: An elementary, general approach.
\newblock {\em Annual Review of Economics}, 7(1):649--688, 2015.

\bibitem{chernozhukov2020adversarial}
Victor Chernozhukov, Whitney Newey, Rahul Singh, and Vasilis Syrgkanis.
\newblock Adversarial estimation of riesz representers.
\newblock {\em arXiv preprint arXiv:2101.00009}, 2020.

\bibitem{chernozhukov2022automatic}
Victor Chernozhukov, Whitney~K Newey, and Rahul Singh.
\newblock Automatic debiased machine learning of causal and structural effects.
\newblock {\em Econometrica}, 90(3):967--1027, 2022.

\bibitem{chernozhukov2022debiased}
Victor Chernozhukov, Whitney~K Newey, and Rahul Singh.
\newblock Debiased machine learning of global and local parameters using regularized riesz representers.
\newblock {\em The Econometrics Journal}, 25(3):576--601, 2022.

\bibitem{fan2017concordance}
Caiyun Fan, Wenbin Lu, Rui Song, and Yong Zhou.
\newblock Concordance-assisted learning for estimating optimal individualized treatment regimes.
\newblock {\em Journal of the Royal Statistical Society Series B: Statistical Methodology}, 79(5):1565--1582, 2017.

\bibitem{foster2019orthogonal}
Dylan~J Foster and Vasilis Syrgkanis.
\newblock Orthogonal statistical learning.
\newblock {\em The Annals of Statistics}, 51(3):879--908, 2023.

\bibitem{gunn2022adaptive}
Kevin Gunn, Wenbin Lu, and Rui Song.
\newblock Adaptive semi-supervised inference for optimal treatment decisions with electronic medical record data.
\newblock {\em arXiv preprint arXiv:2203.02318}, 2022.

\bibitem{haarnoja2017reinforcement}
Tuomas Haarnoja, Haoran Tang, Pieter Abbeel, and Sergey Levine.
\newblock Reinforcement learning with deep energy-based policies.
\newblock In {\em International conference on machine learning}, pages 1352--1361. PMLR, 2017.

\bibitem{Haarnoja2018}
Tuomas Haarnoja, Aurick Zhou, Pieter Abbeel, and Sergey Levine.
\newblock Soft actor-critic: Off-policy maximum entropy deep reinforcement learning with a stochastic actor.
\newblock In Jennifer Dy and Andreas Krause, editors, {\em Proceedings of the 35th International Conference on Machine Learning}, volume~80 of {\em Proceedings of Machine Learning Research}, pages 1861--1870. PMLR, 10--15 Jul 2018.

\bibitem{hasminskii1979nonparametric}
Rafail~Z Hasminskii and Ildar~A Ibragimov.
\newblock On the nonparametric estimation of functionals.
\newblock In {\em Proceedings of the Second Prague Symposium on Asymptotic Statistics}, volume 473, pages 474--482. North-Holland Amsterdam, 1979.

\bibitem{hernan2010causal}
Miguel~A Hern{\'a}n and James~M Robins.
\newblock Causal inference, 2010.

\bibitem{hirshberg2017augmented}
David~A Hirshberg and Stefan Wager.
\newblock Augmented minimax linear estimation.
\newblock {\em arXiv preprint arXiv:1712.00038}, 2017.

\bibitem{hirshberg2018debiased}
David~A Hirshberg and Stefan Wager.
\newblock Debiased inference of average partial effects in single-index models.
\newblock {\em arXiv preprint arXiv:1811.02547}, 2018.

\bibitem{huang2015characterizing}
Ying Huang, Eric~B Laber, and Holly Janes.
\newblock Characterizing expected benefits of biomarkers in treatment selection.
\newblock {\em Biostatistics}, 16(2):383--399, 2015.

\bibitem{ibragimov1981statistical}
I~Ibragimov and R~Has’minskii.
\newblock Statistical estimation, vol. 16 of.
\newblock {\em Applications of Mathematics}, 1981.

\bibitem{jankova2015confidence}
Jana Jankova and Sara Van De~Geer.
\newblock Confidence intervals for high-dimensional inverse covariance estimation.
\newblock {\em Electronic Journal of Statistics}, 9(1):1205--1229, 2015.

\bibitem{jankova2016confidence}
Jana Jankova and Sara Van De~Geer.
\newblock Confidence regions for high-dimensional generalized linear models under sparsity.
\newblock {\em arXiv:1610.01353}, 2016.

\bibitem{jankova2018semiparametric}
Jana Jankova and Sara Van De~Geer.
\newblock Semiparametric efficiency bounds for high-dimensional models.
\newblock {\em The Annals of Statistics}, 46(5):2336--2359, 2018.

\bibitem{javanmard2014confidence}
Adel Javanmard and Andrea Montanari.
\newblock Confidence intervals and hypothesis testing for high-dimensional regression.
\newblock {\em The Journal of Machine Learning Research}, 15(1):2869--2909, 2014.

\bibitem{javanmard2014hypothesis}
Adel Javanmard and Andrea Montanari.
\newblock Hypothesis testing in high-dimensional regression under the {G}aussian random design model: Asymptotic theory.
\newblock {\em IEEE Transactions on Information Theory}, 60(10):6522--6554, 2014.

\bibitem{javanmard2018debiasing}
Adel Javanmard and Andrea Montanari.
\newblock Debiasing the lasso: Optimal sample size for {G}aussian designs.
\newblock {\em The Annals of Statistics}, 46(6A):2593--2622, 2018.

\bibitem{kallus2020}
Nathan Kallus and Masatoshi Uehara.
\newblock Statistically efficient off-policy policy gradients.
\newblock In Hal~Daumé III and Aarti Singh, editors, {\em Proceedings of the 37th International Conference on Machine Learning}, volume 119 of {\em Proceedings of Machine Learning Research}, pages 5089--5100. PMLR, 13--18 Jul 2020.

\bibitem{pmlr-v139-karampatziakis21a}
Nikos Karampatziakis, Paul Mineiro, and Aaditya Ramdas.
\newblock Off-policy confidence sequences.
\newblock In Marina Meila and Tong Zhang, editors, {\em Proceedings of the 38th International Conference on Machine Learning}, volume 139 of {\em Proceedings of Machine Learning Research}, pages 5301--5310. PMLR, 18--24 Jul 2021.

\bibitem{klaassen1987consistent}
Chris~AJ Klaassen.
\newblock Consistent estimation of the influence function of locally asymptotically linear estimators.
\newblock {\em The Annals of Statistics}, pages 1548--1562, 1987.

\bibitem{kosorok2007introduction}
Michael~R Kosorok.
\newblock {\em Introduction to empirical processes and semiparametric inference}.
\newblock Springer Science \& Business Media, 2007.

\bibitem{laber2014dynamic}
Eric~B Laber, Daniel~J Lizotte, Min Qian, William~E Pelham, and Susan~A Murphy.
\newblock Dynamic treatment regimes: Technical challenges and applications.
\newblock {\em Electronic journal of statistics}, 8(1):1225, 2014.

\bibitem{laber2011adaptive}
Eric~B Laber and Susan~A Murphy.
\newblock Adaptive confidence intervals for the test error in classification.
\newblock {\em Journal of the American Statistical Association}, 106(495):904--913, 2011.

\bibitem{levit1976efficiency}
B~Ya Levit.
\newblock On the efficiency of a class of non-parametric estimates.
\newblock {\em Theory of Probability \& Its Applications}, 20(4):723--740, 1976.

\bibitem{lewis2020double}
Greg Lewis and Vasilis Syrgkanis.
\newblock Double/debiased machine learning for dynamic treatment effects via g-estimation.
\newblock {\em arXiv preprint arXiv:2002.07285}, 2020.

\bibitem{lu2013variable}
Wenbin Lu, Hao~Helen Zhang, and Donglin Zeng.
\newblock Variable selection for optimal treatment decision.
\newblock {\em Statistical methods in medical research}, 22(5):493--504, 2013.

\bibitem{luedtke2016statistical}
Alexander~R Luedtke and Mark~J Van Der~Laan.
\newblock Statistical inference for the mean outcome under a possibly non-unique optimal treatment strategy.
\newblock {\em Annals of statistics}, 44(2):713, 2016.

\bibitem{Moodie2007}
Erica E.~M. Moodie, Thomas~S. Richardson, and David~A. Stephens.
\newblock Demystifying optimal dynamic treatment regimes.
\newblock {\em Biometrics}, 63(2):447--455, 2007.

\bibitem{moodie2012q}
Erica~EM Moodie, Bibhas Chakraborty, and Michael~S Kramer.
\newblock Q-learning for estimating optimal dynamic treatment rules from observational data.
\newblock {\em Canadian Journal of Statistics}, 40(4):629--645, 2012.

\bibitem{moodie2010estimating}
Erica~EM Moodie and Thomas~S Richardson.
\newblock Estimating optimal dynamic regimes: Correcting bias under the null.
\newblock {\em Scandinavian Journal of Statistics}, 37(1):126--146, 2010.

\bibitem{moodie2007demystifying}
Erica~EM Moodie, Thomas~S Richardson, and David~A Stephens.
\newblock Demystifying optimal dynamic treatment regimes.
\newblock {\em Biometrics}, 63(2):447--455, 2007.

\bibitem{murphy2003optimal}
Susan~A Murphy.
\newblock Optimal dynamic treatment regimes.
\newblock {\em Journal of the Royal Statistical Society: Series B (Statistical Methodology)}, 65(2):331--355, 2003.

\bibitem{nachum2017bridging}
Ofir Nachum, Mohammad Norouzi, Kelvin Xu, and Dale Schuurmans.
\newblock Bridging the gap between value and policy based reinforcement learning.
\newblock {\em Advances in neural information processing systems}, 30, 2017.

\bibitem{newey1994asymptotic}
Whitney~K Newey.
\newblock The asymptotic variance of semiparametric estimators.
\newblock {\em Econometrica}, pages 1349--1382, 1994.

\bibitem{newey1998undersmoothing}
Whitney~K Newey, Fushing Hsieh, and James~M Robins.
\newblock Undersmoothing and bias corrected functional estimation.
\newblock Technical report, MIT Department of Economics, 1998.

\bibitem{newey2004twicing}
Whitney~K Newey, Fushing Hsieh, and James~M Robins.
\newblock Twicing kernels and a small bias property of semiparametric estimators.
\newblock {\em Econometrica}, 72(3):947--962, 2004.

\bibitem{newey2018cross}
Whitney~K Newey and James~R Robins.
\newblock Cross-fitting and fast remainder rates for semiparametric estimation.
\newblock {\em arXiv preprint arXiv:1801.09138}, 2018.

\bibitem{neykov2018unified}
Matey Neykov, Yang Ning, Jun~S Liu, and Han Liu.
\newblock A unified theory of confidence regions and testing for high-dimensional estimating equations.
\newblock {\em Statistical Science}, 33(3):427--443, 2018.

\bibitem{ning2017general}
Yang Ning and Han Liu.
\newblock A general theory of hypothesis tests and confidence regions for sparse high dimensional models.
\newblock {\em The Annals of Statistics}, 45(1):158--195, 2017.

\bibitem{pfanzagl1982lecture}
Johann Pfanzagl.
\newblock Lecture notes in statistics.
\newblock {\em Contributions to a general asymptotic statistical theory}, 13, 1982.

\bibitem{qian2011performance}
Min Qian and Susan~A Murphy.
\newblock Performance guarantees for individualized treatment rules.
\newblock {\em Annals of statistics}, 39(2):1180, 2011.

\bibitem{ren2015asymptotic}
Zhao Ren, Tingni Sun, Cun-Hui Zhang, and Harrison~H Zhou.
\newblock Asymptotic normality and optimalities in estimation of large {G}aussian graphical models.
\newblock {\em The Annals of Statistics}, 43(3):991--1026, 2015.

\bibitem{robins2007comment}
James Robins, Mariela Sued, Quanhong Lei-Gomez, and Andrea Rotnitzky.
\newblock Comment: Performance of double-robust estimators when" inverse probability" weights are highly variable.
\newblock {\em Statistical Science}, 22(4):544--559, 2007.

\bibitem{robins2004optimal}
James~M Robins.
\newblock Optimal structural nested models for optimal sequential decisions.
\newblock In {\em Proceedings of the Second Seattle Symposium in Biostatistics: analysis of correlated data}, pages 189--326. Springer, 2004.

\bibitem{robins1995semiparametric}
James~M Robins and Andrea Rotnitzky.
\newblock Semiparametric efficiency in multivariate regression models with missing data.
\newblock {\em Journal of the American Statistical Association}, 90(429):122--129, 1995.

\bibitem{robins1995analysis}
James~M Robins, Andrea Rotnitzky, and Lue~Ping Zhao.
\newblock Analysis of semiparametric regression models for repeated outcomes in the presence of missing data.
\newblock {\em Journal of the american statistical association}, 90(429):106--121, 1995.

\bibitem{robinson1988root}
Peter~M Robinson.
\newblock Root-n-consistent semiparametric regression.
\newblock {\em Econometrica: Journal of the Econometric Society}, pages 931--954, 1988.

\bibitem{rothenhausler2019incremental}
Dominik Rothenh{\"a}usler and Bin Yu.
\newblock Incremental causal effects.
\newblock {\em arXiv preprint arXiv:1907.13258}, 2019.

\bibitem{schulman2017equivalence}
John Schulman, Xi~Chen, and Pieter Abbeel.
\newblock Equivalence between policy gradients and soft q-learning.
\newblock {\em arXiv preprint arXiv:1704.06440}, 2017.

\bibitem{schulte2014q}
Phillip~J Schulte, Anastasios~A Tsiatis, Eric~B Laber, and Marie Davidian.
\newblock Q-and a-learning methods for estimating optimal dynamic treatment regimes.
\newblock {\em Statistical science: a review journal of the Institute of Mathematical Statistics}, 29(4):640, 2014.

\bibitem{shao1989general}
Jun Shao and CF~Jeff Wu.
\newblock A general theory for jackknife variance estimation.
\newblock {\em The annals of Statistics}, pages 1176--1197, 1989.

\bibitem{shi2018high}
Chengchun Shi, Alin Fan, Rui Song, and Wenbin Lu.
\newblock High-dimensional a-learning for optimal dynamic treatment regimes.
\newblock {\em Annals of statistics}, 46(3):925, 2018.

\bibitem{shi2022statistically}
Chengchun Shi, Shikai Luo, Yuan Le, Hongtu Zhu, and Rui Song.
\newblock Statistically efficient advantage learning for offline reinforcement learning in infinite horizons.
\newblock {\em Journal of the American Statistical Association}, pages 1--14, 2022.

\bibitem{singh2019biased}
Rahul Singh and Liyang Sun.
\newblock De-biased machine learning in instrumental variable models for treatment effects.
\newblock {\em arXiv preprint arXiv:1909.05244}, 2019.

\bibitem{song2017semiparametric}
Rui Song, Shikai Luo, Donglin Zeng, Hao~Helen Zhang, Wenbin Lu, and Zhiguo Li.
\newblock Semiparametric single-index model for estimating optimal individualized treatment strategy.
\newblock {\em Electronic journal of statistics}, 11(1):364, 2017.

\bibitem{song2015penalized}
Rui Song, Weiwei Wang, Donglin Zeng, and Michael~R Kosorok.
\newblock Penalized q-learning for dynamic treatment regimens.
\newblock {\em Statistica Sinica}, 25(3):901, 2015.

\bibitem{swanepoel1986note}
Jan~WH Swanepoel.
\newblock A note on proving that the (modified) bootstrap works.
\newblock {\em Communications in Statistics-Theory and Methods}, 15(11):3193--3203, 1986.

\bibitem{toth2016tmle}
Boriska Toth and Mark~J van~der Laan.
\newblock Tmle for marginal structural models based on an instrument.
\newblock 2016.

\bibitem{tsiatis2007semiparametric}
Anastasios Tsiatis.
\newblock {\em Semiparametric theory and missing data}.
\newblock Springer Science \& Business Media, 2007.

\bibitem{van2014asymptotically}
Sara Van~de Geer, Peter B{\"u}hlmann, Ya'acov Ritov, and Ruben Dezeure.
\newblock On asymptotically optimal confidence regions and tests for high-dimensional models.
\newblock {\em The Annals of Statistics}, 42(3):1166--1202, 2014.

\bibitem{van2011targeted}
Mark~J Van~der Laan, Sherri Rose, et~al.
\newblock {\em Targeted learning: causal inference for observational and experimental data}, volume~4.
\newblock Springer, 2011.

\bibitem{van2006targeted}
Mark~J Van Der~Laan and Daniel Rubin.
\newblock Targeted maximum likelihood learning.
\newblock {\em The international journal of biostatistics}, 2(1), 2006.

\bibitem{vaar:1991}
Aad Van Der~Vaart et~al.
\newblock On differentiable functionals.
\newblock {\em The Annals of Statistics}, 19(1):178--204, 1991.

\bibitem{van2000asymptotic}
Aad~W Van~der Vaart.
\newblock {\em Asymptotic Statistics}, volume~3.
\newblock Cambridge University Press, 2000.

\bibitem{watkins1992q}
Christopher~JCH Watkins and Peter Dayan.
\newblock Q-learning.
\newblock {\em Machine learning}, 8:279--292, 1992.

\bibitem{watkins1989learning}
Christopher John Cornish~Hellaby Watkins.
\newblock Learning from delayed rewards.
\newblock 1989.

\bibitem{zhang2012robust}
Baqun Zhang, Anastasios~A Tsiatis, Eric~B Laber, and Marie Davidian.
\newblock A robust method for estimating optimal treatment regimes.
\newblock {\em Biometrics}, 68(4):1010--1018, 2012.

\bibitem{zhang2014confidence}
Cun-Hui Zhang and Stephanie~S Zhang.
\newblock Confidence intervals for low dimensional parameters in high dimensional linear models.
\newblock {\em Journal of the Royal Statistical Society: Series B (Statistical Methodology)}, 76(1):217--242, 2014.

\bibitem{zhao2012estimating}
Yingqi Zhao, Donglin Zeng, A~John Rush, and Michael~R Kosorok.
\newblock Estimating individualized treatment rules using outcome weighted learning.
\newblock {\em Journal of the American Statistical Association}, 107(499):1106--1118, 2012.

\bibitem{zhao2009reinforcement}
Yufan Zhao, Michael~R Kosorok, and Donglin Zeng.
\newblock Reinforcement learning design for cancer clinical trials.
\newblock {\em Statistics in medicine}, 28(26):3294--3315, 2009.

\bibitem{zhao2011reinforcement}
Yufan Zhao, Donglin Zeng, Mark~A Socinski, and Michael~R Kosorok.
\newblock Reinforcement learning strategies for clinical trials in nonsmall cell lung cancer.
\newblock {\em Biometrics}, 67(4):1422--1433, 2011.

\bibitem{zhu2017breaking}
Yinchu Zhu and Jelena Bradic.
\newblock Breaking the curse of dimensionality in regression.
\newblock {\em arXiv:1708.00430.}, 2017.

\bibitem{zhu2018linear}
Yinchu Zhu and Jelena Bradic.
\newblock Linear hypothesis testing in dense high-dimensional linear models.
\newblock {\em Journal of the American Statistical Association}, 113(524):1583--1600, 2018.

\end{thebibliography}
\bibliographystyle{plain}

\appendix
\section*{Appendix}
The Appendix includes further details on the characterization and identification of the optimal dynamic policies and omitted proofs.
\section{Preliminary Definitions and Notation}

We introduce here some notation and basic definitions that will be used throughout the appendix. 

\begin{definition}[Uniform in $\beta$ Convergence]
Let $(X_n^{\beta})$ be a sequence of random variables. We say that $X_n^\beta=o_{p,\rm{unif}(\beta)}(1)$ if  for all $\epsilon>0$ the following holds: $$\lim_{n\to\infty}\sup_{\beta>0}\, pr\{|X_n^\beta|\ge \epsilon\}=0.$$ 
Similarly we say that $X_n^\beta=O_{p,\rm{unif}(\beta)}(1)$ if for all $\delta>0$ there is an $M>0$ such that $$\sup_{\beta>0,n\in \mathbb{N}} pr\{|X_n^\beta|>M\}\le \delta.$$ 
\end{definition}

The condition $X_{n}^\beta = o_{p, \rm{unif}(\beta)}(1)$ is significantly weaker than assuming that $\sup_{\beta>0}|X_n^{\beta}|=o_p(1)$, which would often prove to be a prohibitively strong assumption in many settings. Moreover, we say that a sequence of real numbers $(x_\beta)$ satisfies that $x_\beta = o_\beta(1)$ if $x_\beta\to 0$ as $\beta\to\infty.$

For any fixed or random vector $v$, we denote with $\|v\|_{p}=\{\sum |v_i|^p\}^{1/p}$  for any $p\ge 1$, and with $\|v\|_{\infty}=\max |v_i|$, where $v_i$ are the individual components of vector $v$. For any fixed or random matrix $A,$ we denote with $\|A\|_{\infty}=\max|A_{j, k}|$  where $A_{j, k}$ are the individual components of matrix $A$. For any random variable $X$, we denote with $\|X\|_{L_p}=\{\E(|X|^p)\}^{1/p}$ for any $p\ge 1$, and with $\|X\|_{L_\infty}=\inf\{c\geq 0: pr(|X| \le c) = 1\}$. Moreover, for any vector-valued function $g$ that takes as input a random variable $X$, we define $\|g\|_2 = [\E_X\{\|g(X)\|_2^2\}]^{1/2}$. For any random or fixed matrix $A$, we define $\|A\|_{op}$ as the operator norm of the matrix associated with matrix multiplication in $\ell_2$ space.

We also define the short-hand notation for the feature map of the softmax policy:
\begin{align*}
    \phi_{\beta,\psi}(X) = \sum_{\tau\in\mcT} W_{\tau}^{\beta,\psi} \phi(\tau, X), \quad
    W_\tau^{\beta,\psi} =~ \frac{\exp\{\beta\, \psi ^{\trans}\phi(\tau, X)\}}{\sum_{t\in\mcT} \exp\{\beta\, \psi ^{\trans}\phi(t, X)\}}, \\
    \phi_{\beta}(X) = \sum_{\tau\in\mcT} W_{\tau}^{\beta} \phi(\tau, X), \quad
    W_\tau^\beta = \frac{\exp\{\beta\, \psi_0 ^{\trans}\phi(\tau, X)\}}{\sum_{t\in\mcT} \exp\{\beta\, \psi_0 ^{\trans}\phi(t, X)\}} .
\end{align*}

\section{Main Bias Lemma: Controlling Softmax Error}
 
 One important component of the proof for Theorem~\ref{thm:main} is to show that the bias introduced by the softmax approximation vanishes faster than $n^{-1/2}$, if the temperature parameter $\beta$ grows faster than $n^{1/\{2(1+\delta)\}}$ for some appropriately chosen $\delta\in(0,1]$. We remind that $\theta_0^\beta$ is the solution to the softmax approximate moment condition $M(\theta, \psi_0, q^*, p_1^*, p_{2,\beta}^*;\beta)=0$ with respect to $\theta$ and $\theta_0$ is the true structural parameter of the first period blip effect, which can also equivalently be thought as the solution to the later moment condition for $\beta=\infty$. To control the bias $\sqrt{n} (\theta_0^\beta-\theta_0)$,
we need to control the discrepancy between $\max_{\tau\in \mcT} \psi_0^{\trans}\phi(\tau, X)$ and $\softmax_{\tau\in \mcT} \psi_0^{\trans}\phi(\tau, X)$. This is exactly the crux of the proof of the following bias lemma.

\begin{lemma}[Softmax Bias Control]\label{lem:bias}
For each treatment $\tau\in\mathcal{T}$, define random variable
\begin{align*}
    U_\tau= \psi_0^{\trans}\phi(\tau,X).
\end{align*}
Taking the maximum of $U_\tau$ over all treatments $\tau\in\mathcal{T}$, we define
\begin{align*}
    U_{\max}= \max_{t\in \mathcal{T}} \left\{\psi_0^{\trans}\phi(t,X)\right\}.
\end{align*}
Assume that there is a constant $c>0$ such that  for each treatment $\tau\in\mathcal{T}$, we have that $\|U_\tau\|_{\infty}<\infty$ and that on $(0,c)$ the random variable $U_{\max} - U_\tau$ admits a density $f_\tau$, that satisfies $f_\tau(x)\leq H/x^{1-\delta}$, for some $0\leq H<\infty$ and $0<\delta\leq 1$. Suppose that $\|\tilde{M}\|_{2}$ is uniformly bounded and that $\E(\tilde{M}\tilde{M}^{\trans})$ is strictly positive definite.
If $\beta = \omega(n^{1/\{2(1+\delta)\}})$ then we get that as $n\to \infty$:
\begin{align*}
    \sqrt{n}(\theta_0^\beta - \theta_0) = o(1).
\end{align*}
\end{lemma}

  To illustrate the main argument for the proof of Lemma~\ref{lem:bias}, consider the case when there are only two actions, i.e. $\mcT=\{0, 1\}$. In this binary treatment case, letting $U_1=\psi_0^{\trans}\phi(1, X)$ and that $U_0=\psi_0^{\trans}\phi(0,X)=0$, the term difference between the softmax and the max over the quantities $\{U_0, U_1\}$ is:
$$
     U_1 \left\{\mathbbm{1}_{\{U_1 \geq 0\}} - \frac{1}{1 + \exp(-\beta U_1)}\right\}
$$
which simplifies to:
 \begin{align*}
     \frac{|U_1|  \exp(-\beta U_1)}{1 + \exp\{-\beta U_1\}} \mathbbm{1}_{\{U_1 \geq 0\}} + \frac{|U_1|  \exp(\beta U_1)}{1 + \exp(\beta U_1)} \mathbbm{1}_{\{U_1 < 0\}} = \frac{|U_1|  \exp(-\beta |U_1|)}{1 +  \exp(-\beta |U_1|)}
     \leq~& |U_1|  \exp(-\beta |U_1|).
 \end{align*}
Thus the bias term $\sqrt{n}(\theta_0^{\beta} - \theta_0)$ will be roughly upper bounded by $O(\sqrt{n} \E\{|U_1|  \exp(-\beta |U_1|)\})$. Hence, if we can control quantities of the form $\E\{|U_1|  \exp(-\beta |U_1|)\}$, then we will be able to control the bias. For instance, if we can show that quantities of the form $\E\{|U_1|  \exp(-\beta |U_1|)\}$ decay faster than $1/\beta^{1+\delta}$, then we would need that $n^{1/2}/\beta^{1+\delta}=o(1)$, or equivalently, $\beta=\omega(n^{1/2(1+\delta)})$, which is our desired target result. To prove such a statement we develop the following key lemma. 

\begin{lemma}[Key Bias Building Block]\label{lem:second order}
Let $U$ be a non-negative random variable. Assume that there is a constant $c>0$ such that for all $x\in (0,c)$ the random variable $U$ admits a density $f$, that satisfies $f(x)\leq H/x^{1-\delta}$, for some $0\leq H<\infty$ and $0<\delta\leq 1$. Choose any $\epsilon> 0$. Then for any $\beta$ that is large enough, such that $\frac{(1+\epsilon)\log(\beta)}{\beta}\leq c$ and $(1+\epsilon)\log(\beta) > 1$, we have 
$$\E\{U\exp\left(-\beta U\right)\} \leq \frac{H}{\beta^{1+\delta}} + \frac{(1+\epsilon) \log(\beta)}{\beta^{2+\epsilon}}.$$
\end{lemma}
\begin{proof}
Let $\epsilon>0$. Suppose that $\beta$ is large enough such that $\alpha=\frac{(1+\epsilon)\log(\beta)}{\beta}\le c$ and such that $(1+\epsilon)\log(\beta)> 1$. Then we remark that the following decomposition holds
$$
\E\left\{\beta^{2}U\exp\left(-\beta U\right)\right\}
=A_\beta
+ B_\beta,$$
\begin{align*}
A_\beta:=~&\E\left[\beta^{2}U\exp\left(-\beta U\right)\mathbbm{1}\left\{0<U\leq \alpha\right\}\right], &
B_\beta:=~&\E\left[\beta^{2}U\exp\left(-\beta U\right)\mathbbm{1}\left\{U> \alpha\right\}\right].
\end{align*}
We will successively upper bound each of the terms $A_{\beta}$ and $B_{\beta}$. Firstly we remark that as ${(1+\epsilon)\log(\beta)}/{\beta}\le c$ then by assumption we know that there are constants $H<\infty$ and $0<\delta\leq 1$, such that the density of $U$ is upper bounded by $H/x^{1-\delta}$ for any $x\in (0, \alpha]$. Then the following holds
\begin{align*}
A_\beta =~& \E\left[\beta^{2}U\exp\left(-\beta U\right)\mathbbm{1}\left\{0 < U\leq \alpha\right\}\right]= \int
_{0}^{\alpha}\beta^{2}u
\exp(-\beta u) f(u) du\\
\leq~& H\int
_{0}^{\alpha}\beta^{2}u
\exp(-\beta u) \frac{1}{u^{1-\delta}} du
=~ H\int_{0}^{{(1+\epsilon)\log\beta} }v\exp(-v) \frac{\beta^{1-\delta}}{v^{1-\delta}} dv\\
\le~& H \beta^{1-\delta} \int_{0}^{\infty}v^{\delta}\exp(-v) dv \leq H \beta^{1-\delta}.
\end{align*}
The penultimate line follows by a change of variable: $\beta u=v$ and $dv = {\beta} du$. We now move on to bounding the term $B_{\beta}$. When $U > \alpha = (1+\epsilon)\log\beta/{\beta},$ we have
\begin{align*}
    \beta^{2}U\exp\left(-\beta U\right){ \le }\beta^{2}\alpha\exp\left(-\beta\alpha\right)
    = \beta^{-\epsilon} (1+\epsilon) \log\beta,
\end{align*}
where to obtain the first inequality we used the fact that the function $x\mapsto x\exp(-x)$ is decreasing for $x > 1$ and we assumed that $\beta$ was large enough such that $(1+\epsilon) \log(\beta) > 1$. 
\end{proof}

\begin{figure}
    \centering
\begin{tikzpicture}[scale=0.5]
\begin{axis}[
    axis lines=middle,
    xlabel=\(x\),
    ylabel=\(y\),
    xlabel style={at={(ticklabel* cs:1)}, anchor=north west},
    ylabel style={at={(ticklabel* cs:1)}, anchor=south east},
    enlargelimits,
    domain=0:5,
    samples=100,
    ytick={0.3679},
    yticklabels={\(\frac{1}{e}\)}
]
\addplot[name path=F, black] {x*exp(-x)} node[pos=0.6, above] {\(y = x e^{-x}\)};
\path[name path=Axis] (axis cs:0,0) -- (axis cs:5,0);
\addplot [
    thick,
    color=gray,
    fill=gray,
    fill opacity=0.2
] fill between[of=F and Axis];

\draw [dashed] (axis cs:1,0) -- (axis cs:1,{1*exp(-1)});
\draw [dashed] (axis cs:0,{exp(-1)}) -- (axis cs:1,{exp(-1)});

\end{axis}
\end{tikzpicture}
    \caption{Behavior of function $x \exp(-x)$.}
    \label{fig:graph}
\end{figure}

\begin{remark} 
Concluding this section, we want to remark the importance of Lemma~\ref{lem:second order} and more generally our main bias Lemma~\ref{lem:bias}. We note that a crude analysis of the difference between the max and the softmax would give a bound of the order of $1/\beta$. For instance, one way to analyze the difference between the max of the $\{U_{\tau}\}_{\tau\in \mcT}$ versus the $\beta$-softmax is to use the equivalence between the softmax and an entropic regularized maximum. In particular, the $\beta$-softmax is equivalent to the maximum of the entropic regularized objective $\max_{w} \{\sum_{\tau\in\mcT} w_{\tau} U_{\tau} - \beta^{-1}\sum_{\tau\in\mcT} w_{\tau} \log(w_{\tau})\}$, where the vector $w$ ranges over the $|\mcT|$-dimensional simplex. From this we see that the distortion in the objective is of the order of $1/\beta$. Thus very quickly we can argue that the difference between the maximum and the soft-maximum is of the order of $1/\beta$. However, such a bound is not strong enough and if used to bound the soft-max bias $\sqrt{n}(\theta_0^{\beta}-\theta_0)$, would result in requiring that $\beta=\omega(n^{1/2})$. However, the analysis of the variance part that we present in the subsequent sections, will be imposing that for asymptotic linearity of $\theta_0^\beta$, we need that $\beta=o(n^{1/2})$, which would result in a contradiction. Thus our more fine grained analysis of the difference between the max and the soft-max and our reduction to terms that look like the ones that are handled by Lemma~\ref{lem:second order}, as well as the fact that the refined analysis in Lemma~\ref{lem:second order} provides a bound that decays much faster as $\beta^{-(1+\delta)}$, is of significant importance for our main result.
\end{remark}

\section{Omitted Proofs from Section~\ref{sec: identiG}: Identification of Optimal Policies}

\subsection{Characterization of Optimal Policy}\label{app:charac}

By invoking the Markovianity of the policy and the conditional independencies implied by the intervention graphs, we can characterize the optimal dynamic policy as follows:
\begin{align*}
    \max_{\pi} \E\left\{Y^{(\pi)}\right\} =~& \max_{\pi} \E\left[ \E\left\{ Y^{(\pi)} \mid S\right\} \right]
    = \max_{\pi_2} \E\left[ \max_{\tau_1} \E\left\{ Y^{(\tau_1, \pi_2)} \mid S\right\} \right]
\end{align*}
where the second equality follows because the first period action $\tau_1$ is not allowed to depend on the second period policy $\pi_2$. 
Thus the optimal first period policy is defined as: 
$$
\pi_1^*(S) = \arg\max_{\tau_1} \E\left\{Y^{(\tau_1, \pi_2^*)}\mid S\right\}.
$$
Moreover, we can further write:
\begin{align*}
    \max_{\pi_2} \E\left[ \max_{\tau_1} \E\left\{ Y^{(\tau_1, \pi_2)} \mid S\right\} \right] =~& 
    \max_{\pi_2} \E\left[ \max_{\tau_1} \E\left\{ Y^{(\tau_1, \pi_2)} \mid T_1=\tau_1, S\right\}\right]\\
    =~& 
    \max_{\pi_2} \E\left[ \max_{\tau_1} \E\left\{Y^{(T_1, \pi_2)} \mid T_1=\tau_1, S\right\}\right]\\
    =~& 
    \max_{\pi_2} \E\left( \max_{\tau_1} \E\left[\E\left\{Y^{(T_1, \pi_2)}\mid X, T_1, S\right\} \mid T_1=\tau_1, S\right]\right)\\
    =~& 
    \max_{\pi_2} \E\left( \max_{\tau_1} \E\left[ \E\left\{Y^{(T_1, \pi_2)}\mid X\right\} \mid T_1=\tau_1, S\right]\right)\\
    =~& 
    \E\left( \max_{\tau_1} \E\left[ \max_{\tau_2} \E\left\{Y^{(T_1, \tau_2)}\mid X\right\} \mid T_1=\tau_1, S\right]\right).
\end{align*}
Thus the optimal second period policy is defined as 
$$
\pi_2^*(X) = \arg\max_{\tau_2} \E\left\{Y^{(T_1, \tau_2)}\mid X\right\}.
$$

\subsection{Blip Effect Decomposition of Policy Improvement}\label{app:blip}

The improvement that any counterfactual policy $\pi$ brings, as compared to the observed policy, can be decomposed as the sum of a sequence of improvements
\begin{align*}
    \E\left\{Y^{(\pi)} - Y\right\} =~& \E\left\{Y^{(\pi_1, \pi_2)} - Y^{(T_1, T_2)}\right\}
    = \E\left\{Y^{(\pi_1, \pi_2)} - Y^{(T_1, \pi_2)}\right\} + \E\left\{Y^{(T_1, \pi_2)} - Y^{(T_1, T_2)}\right\}.
\end{align*}
We can further center each of the terms around a baseline treatment, typically $0$ and write:
\begin{align*}
\E\left\{Y^{(T_1, \pi_2)} - Y^{(T_1, T_2)}\right\} = 
\E\left\{Y^{(T_1, \pi_2)} - Y^{(T_1, 0)}\right\} - \E\left\{Y^{(T_1, T_2)}-Y^{(T_1, 0)}\right\},\\
\E\left[Y^{(\pi_1, \pi_2)} - Y^{(T_1, \pi_2)}\right\} = 
\E\left\{Y^{(\pi_1, \pi_2)} - Y^{(0, \pi_2)}\right\} - \E\left\{Y^{(T_1, \pi_2)}-Y^{(0, \pi_2)}\right\}.
\end{align*}
The values of these various counterfactual random variables can also be easily depicted in a SWIG (Figure~\ref{fig:swig2}). We can further see easily from the intervention graphs that $Y^{(T_1, \pi_2)}\ci T_2 \mid X$ and $Y^{(T_1, 0)} \ci T_2 \mid X$ and $Y^{(\pi_1, \pi_2)}\ci T_1 \mid S$. 
\begin{figure}[H]
\centering
\begin{subfigure}[t]{0.3\textwidth}
    \centering
    \begin{tikzpicture} \scriptsize
    \node (S) at (0,0) {$S$};
    \node (T1) at (1, 1) {$T_1$};
    \node (X) at (2,0) {$X$};
    \node (T2) at (2.5,1) {$T_2$};
    \node (pi2) at (3.5,1) {$\pi_2(X)$};
    \node  (Y) at (4,0) {$Y^{(T_1, \pi_2)}$};
    \draw[thick, ->] (S) -- (T1);
    \draw[thick, ->] (S) -- (X);
    \draw[thick, ->] (T1) -- (X);
    \draw[thick, ->] (X) -- (Y);
    \draw[thick, ->] (X) -- (T2);
    \draw[thick, ->, dashed] (X) to[bend right] (pi2);
    \draw[thick, ->] (pi2) -- (Y);
    \end{tikzpicture}
    \caption{Intervention with $\pi$ in second period}
\end{subfigure}
\ \ \ 
\begin{subfigure}[t]{0.3\textwidth}
    \centering
    \begin{tikzpicture} \scriptsize
    \node (S) at (0,0) {$S$};
    \node (T1) at (1, 1) {$T_1$};
    \node (X) at (2,0) {$X$};
    \node (T2) at (2.5,1) {$T_2$};
    \node (pi2) at (3.5,1) {$0$};
    \node  (Y) at (4,0) {$Y^{(T_1, 0)}$};
    \draw[thick, ->] (S) -- (T1);
    \draw[thick, ->] (S) -- (X);
    \draw[thick, ->] (T1) -- (X);
    \draw[thick, ->] (X) -- (Y);
    \draw[thick, ->] (X) -- (T2);
    \draw[thick, ->] (pi2) -- (Y);
    \end{tikzpicture}
    \caption{Intervention with baseline in second period.}
\end{subfigure}
~~
\begin{subfigure}[t]{0.3\textwidth}
    \centering
    \begin{tikzpicture} \scriptsize
    \node (S) at (0,0) {$S$};
    \node (T1) at (.5, 1) {$T_1$};
    \node (pi1) at (1.5, 1) {$0$};
    \node (X) at (2,0) {$X^{(0)}$};
    \node (T2) at (2.5,1) {$T_2$};
    \node (pi2) at (4,1) {$\pi_2(X^{(0)})$};
    \node  (Y) at (4.5,0) {$Y^{(0, \pi_2)}$};
    \draw[thick, ->] (S) -- (T1);
    \draw[thick, ->] (S) -- (X);
    \draw[thick, ->] (pi1) -- (X);
    \draw[thick, ->] (X) -- (Y);
    \draw[thick, ->] (X) -- (T2);
    \draw[thick, ->, dashed] (X) to[bend right] (pi2);
    \draw[thick, ->] (pi2) -- (Y);
    \end{tikzpicture}
    \caption{Intervention with baseline in first and $\pi$ in second period.}
\end{subfigure}
\caption{Various interventions graphs that appear in improvement decomposition terms.}\label{fig:swig2}
\end{figure}

We can use these conditional independencies to further put more structure in the aforementioned improvements. Invoking the tower law of expectations and the conditional independencies, we can write:
\begin{align*}
    &\E\left\{Y^{(T_1, \pi_2)} - Y^{(T_1, 0)}\right\} = \E\left[\E\left\{Y^{(T_1, \pi_2)} - Y^{(T_1, 0)}\mid X\right\}\right]\\ 
    =& \E\left[\E\left\{Y^{(T_1, \pi_2)} - Y^{(T_1, 0)}\mid T_2=\pi_2(X), X\right\}\right],
\end{align*}
$$\E\left\{Y^{(T_1, T_2)}-Y^{(T_1, 0)}\right\}= \E\left[\E\left\{Y^{(T_1, T_2)}-Y^{(T_1, 0)}\mid T_2, X\right\}\right],$$
\begin{align*}
    &\E\left\{Y^{(\pi_1, \pi_2)} - Y^{(0, \pi_2)}\right\}= \E\left[\E\left\{Y^{(\pi_1, \pi_2)} - Y^{(0, \pi_2)}\mid S\right\}\right]\\
    &= \E\left[\E\left\{Y^{(\pi_1, \pi_2)} - Y^{(0, \pi_2)}\mid T_1=\pi_1(S), S\right\}\right],
\end{align*}
and $$\E\left\{Y^{(T_1, \pi_2)}-Y^{(0, \pi_2)}\right\} = \E\left[\E\left\{Y^{(T_1, \pi_2)}-Y^{(0, \pi_2)}\mid T_1, S\right\}\right].$$
Thus we see that it suffices to estimate the conditional expectation functions:
\begin{align*}
    \gamma_2(\tau_2, x) =~& \E\left\{Y^{(T_1, T_2)} - Y^{(T_1, 0)}\mid T_2=\tau_2, X=x\right\},\\
    \gamma_1^{(\pi_2)}(\tau_1, s) =~& \E\left\{Y^{(T_1, \pi_2)} - Y^{(T_1, 0)}\mid T_1=\tau_1, S=s\right\}.
\end{align*}
Then we can write the improvement of any policy $\pi$ as:
\begin{align*}
    \E\left\{Y^{(\pi)} - Y\right\} = \E\left\{\gamma_1^{(\pi_2)}(\pi_1(S), S) - \gamma_1^{(\pi_2)}(T_1, S)\right\} + \E\left\{\gamma_2(\pi_2(X), X) - \gamma_2(T_2, X)\right\}.
\end{align*}

\section{Omitted Proofs from Section~\ref{sec:mainthm}}
\subsection{Proof of Lemma~\ref{lem:simpler}}
\begin{proof}
Firstly we remark that for all $t$ we have $U_{\rm{max}}-U_t=\rm{max}_{s\in \mathcal{T}}\big(U_s-U_t\big)=\max\{\rm{max}_{s\ne t}\big(U_s-U_t\big),0\}$. Note that as the random variables $(U_t)_{t\in \mathcal{T}}$ admit a joint density then for all $t\in \mathcal{T}$ the random variables $(U_s-U_t)_{\substack{ s\in \mathcal{T}\\s\ne t}}$ also admit a joint density. Therefore the real-valued random variable $\max_{s\ne t}(U_s-U_t)$ is a continuous random variable. This implies that $(\max_{s\ne t}(U_s-U_t) )^+$ admits a density on $(0,\infty)$, and hence the last sentence of Lemma~\ref{lem:simpler}.
\end{proof}

\section{Omitted Proofs from Bias Appendix in Main Text}

\subsection{Proof of Lemma~\ref{lem:bias}}\label{app:bias}
\begin{proof} 
For simplicity define
\begin{align*}
 \tilde{Y} ~=~& Y - \E(Y\mid S)~=~Y - q^*(S),
\end{align*}
and define
\begin{align*}
    \tilde{\Phi}_\beta ~:&=~\psi_0^{\trans}\{\phi(T_2, X) - \phi_{\beta}(X)\} - \E[\psi_0^{\trans}\{\phi(T_2, X) - \phi_{\beta}(X)\} \mid S]\\
    &=~\psi_0^{\trans}\{\phi(T_2, X) - \phi_{\beta}(X)\} - p_{2,\beta}^*(S).
\end{align*}
Then by definition of $\theta_0^\beta$ and $\theta_0$ we have
\begin{align*}
    M(\theta_0^\beta, \psi_0, q^*, p_1^*, p_{2,\beta}^*;\beta) ~=~\E\left\{\bigg(\tilde{Y} - \tilde{\Phi}_\beta\bigg)\, \tilde{M}\right\}- \E\left(\tilde{M}\tilde{M}^{\trans}\right)\theta_0^\beta ~=~ 0
\end{align*}
and
\begin{align*}
    M(\theta_0, \psi_0, q^*, p_1^*, p_{2,\infty}^*;\infty) ~=~\E\left\{\bigg(\tilde{Y} - \tilde{\Phi}_\infty\bigg)\, \tilde{M}\right\}- \E\left(\tilde{M}\tilde{M}^{\trans}\right)\theta_0 ~=~ 0.
\end{align*}
By subtracting the two we obtain that:
\begin{align*}
    \E(\tilde{M}\tilde{M}^{\trans})\left(\theta_0^\beta - \theta_0\right)~=~\E\left\{\left(\tilde{\Phi}_\infty - \tilde{\Phi}_\beta\right)\, \tilde{M}\right\}.
\end{align*}
That is, we get that
 \begin{align*}
     \theta_0^\beta - \theta_0 = \E(\tilde{M}\tilde{M}^{\trans})^{-1} \E\left[ \psi_0^{\trans}\{\phi_{\beta}(X) - \phi_{\infty}(X)\}\tilde{M} + \{p_{2,\beta}^*(S) - p_{2,\infty}^*(S)\}\tilde M\right].
 \end{align*}
 Now by the tower law, we have
 \begin{align*}
     E\left[\{p_{2,\beta}^*(S) - p_{2,\infty}^*(S)\}\tilde M\right]
     =E\left[\{p_{2,\beta}^*(S) - p_{2,\infty}^*(S)\}E(\tilde M\mid S)\right]=0,
 \end{align*}
 Hence, we obtain that
  \begin{align*}
     \theta_0^\beta - \theta_0 ~&=~\E[\tilde{M}\tilde{M}^{\trans}]^{-1} \E\left[\psi_0^{\trans}\{\phi_{\beta}(X) - \phi_{\infty}(X)\}\tilde{M} \right]
\end{align*}
Note that $\psi_0^{\trans}\phi_{\beta}(X) = \sum_{\tau\in\mathcal{T}} W_\tau^\beta U_\tau$ and that $\psi_0^{\trans}\phi_{\infty}(X) = \max_{\tau} U_{\tau}=U_{\max}$. Thus we have:
\begin{align*}
      &\theta_0^\beta - \theta_0 =\E(\tilde{M}\tilde{M}^{\trans})^{-1} \E\left\{\tilde{M} \left(\sum_{\tau\in\mathcal{T}} W_\tau^\beta U_\tau - U_{\max}\right) \right\}\\
      &=\E(\tilde{M}\tilde{M}^{\trans})^{-1} \E\left\{\tilde{M} \sum_{\tau\in\mathcal{T}} W_\tau^\beta (U_\tau - U_{\max}) \right\}.
 \end{align*}
 where we used the simple fact that $\sum_{\tau \in \mcT} W_{\tau}^{\beta}=1$ in the last equality. Since we assumed that $\|\tilde{M}\|_{2}$ is uniformly bounded, and that the covariance matrix $\E[\tilde{M}\tilde{M}^{\trans}]$ is positive definite, we have that for some constant $C$:
 \begin{align*}
    \left\|\theta_0^\beta - \theta_0\right\|_2 &\leq C\, \E\left\{\left|\sum_{\tau\in\mathcal{T}} W_\tau^\beta (U_\tau - U_{\max})\right|\right\}\\
    &\leq C\, \E\left\{\sum_{\tau\in\mathcal{T}} W_\tau^\beta \left|U_\tau - U_{\max}\right|\right\}
    = C\, \sum_{\tau\in\mathcal{T}} \E\left(W_\tau^\beta \left|U_\tau - U_{\max}\right|\right).
 \end{align*}
 Finally, note that:
 \begin{align*}
    \E\left(W_\tau^\beta \left|U_\tau - U_{\max}\right|\right) =~&\E\left\{\frac{\exp(\beta U_{\tau})\left(U_{\max} - U_\tau\right)}{\sum_{t\in\mcT} \exp(\beta U_t)} \right\}\\
    \leq~& \E\left\{\frac{\exp(\beta U_{\tau})\left(U_{\max} - U_\tau\right)}{\exp(\beta U_{\max})} \right\}\\
    =~& \E\left[\exp\{-\beta (U_{\max} - U_{\tau})\} \left(U_{\max} - U_\tau\right)\right].
 \end{align*}
 Applying Lemma~\ref{lem:second order}, we have that for $n$ sufficiently large, such that $(1+\epsilon)\log(\beta)/ \beta\leq c$ and $(1+\epsilon)\log(\beta) \geq 1$:
 \begin{align*}
    \E\left[\exp\{-\beta (U_{\max} - U_{\tau})\} \left(U_{\max} - U_\tau\right)\right] \leq \frac{H}{\beta^{1+\delta}} + \frac{(1+\epsilon)\log(\beta)}{\beta^{2+\epsilon}}
 \end{align*}
 for appropriately chosen $H$ and $\delta$.
We therefore can conclude that
 \begin{align*}
     \sqrt{n}\left|\theta_0^\beta - \theta_0\right| \leq C|\mcT| \left(\frac{\sqrt{n} H}{\beta^{1+\delta}} + \frac{\sqrt{n}(1+\epsilon)\log(\beta)}{\beta^{2+\epsilon}}\right).
 \end{align*}
If $\beta=\omega(n^{1/\{2(1+\delta)\}})$, then the latter upper bound converges to $0$ as $n\to\infty$.
\end{proof}

\section{Asymptotic Linearity for Moment Equations with Growing Parameters}\label{sec:general}

Before moving to the proof of our main theorem, we present a more general result on asymptotic linearity of estimates defined as the solutions to empirical  moment restrictions with growing parameters. We will use this general result as our main workhorse in proving asymptotic linearity in the next section, when we prove our main Theorem~\ref{thm:main}.

Our problem falls into a general class of semiparametric inference problems, defined as solutions to moment equations that apart from the target parameter, also depend on a set of auxiliary or nuisance parameters. One additional element that we need to add to the classical setting of semi-parametric inference with moment restrictions is that in our problem the moments themselves are parameterized by quantities (e.g. the temperature parameter $\beta$) that grow with the sample size. En route to our main theorem, we will analyze such types of moment problems in their full generality and then instantiate the general theorem to the setting of inference on optimal dynamic treatment regimes.

  In this section, we consider the following generalized method of moments framework:
\begin{align*}
M(\theta, g, h;\beta) = \E_Z\{m_\beta(Z; \theta, g, h)\},\quad
    M(\theta_0^{\beta_n}, g_0^{\beta_n}, h_0;\beta_n) = 0
\end{align*}
where $Z\in \mathcal{Z}$ is a vector of random variables. Apart from some constants $\beta_n$ that grows to infinity as $n\to\infty$ and the target parameter $\theta_0^{\beta_n}\in \Theta\subset\mathbb{R}^p$ of interest, the above moment is also a function of unknown nuisance parameter $h_0\in \R^q$ and unknown nuisance functions $g_0^{\beta_n}\in \mathcal{G}$, which need to be estimated from the data. 

  For simplicity we will drop the subscript of $\beta_n$  and write $\beta=\beta_n$, but it is worth noting that $\beta$ is a function of $n.$ Denote $\theta_0=\theta_0^\infty$ and $g_0=g_0^\infty$.
The estimator satisfies:
\begin{align*}
    M_{n}(\hat{\theta}^{\beta}, \hat{g}^{\beta}, \hat{h};\beta) =o_p(n^{-1/2}),\quad M_n(\theta, g,h;\beta) = n^{-1} \sum_{i=1}^n m_\beta(Z; \theta, g, h) .
\end{align*}

  To simplify the regularity assumptions required for asymptotic normality, we focus on the case where $m_\beta(Z; \theta, g, h)$ is linear in $\theta$, i.e.
\begin{align*}
    m_\beta(Z; \theta, g, h) = a_\beta(Z;g,h)\, \theta + \nu_\beta(Z;g,h)
\end{align*}
where $a_\beta(Z;g,h)\in \mathbb{R}^{p\times p}$ is a $p\times p$ matrix and $\nu_\beta(Z;g,h)\in \mathbb{R}^p$ is a $p$-vector, and we denote with:
\begin{align*}
    A(g,h;\beta) = \E_Z\{a_\beta(Z;g,h)\}, \quad 
    A_n(g) =  \E_n\{a_\beta(Z;g,h)\},\\
    V(g,h;\beta) = \E_Z\{\nu_\beta(Z;g,h)\},\quad
    V_n(g,h;\beta) = \E_n\{\nu_\beta(Z;g,h)\}.
\end{align*}

  Our general asymptotic linearity theorem will be based on a series of high-level assumptions that we provide next. In subsequent sections, we will verify each of these high-level assumptions from more primitive conditions, when applying our general asymptotic linearity theorem to the inference problem that arises in the optimal dynamic treatment regime setting. 

\begin{assumption}[Influence of $h$]\label{ass:influence} 
The functions $\big(\partial_hm_{\beta}(z;\cdot,\cdot,h_0)\big)$ are equicontinuous: for all $\epsilon>0$ there is a $\delta>0$ such that for all $\|g_1-g_2\|_2\le \delta$, and all $\|\theta_1-\theta_2\|_2\le \delta$  the following holds 
$$\sup_{z\in \mathcal{Z},\beta>0 }\|\partial_hm_{\beta}(z;\theta_1,g_1,h_0)-\partial_hm_{\beta}(z;\theta_2,g_2,h_0)\|_{\infty}\le \epsilon.$$
Moreover the expected gradient $\partial_h m_{\beta}$ converges to some finite limit as $\beta\to \infty$:
$$E\{\partial_hm_{\beta}(Z_1;\theta_0^\beta,g_0^\beta,h_0)\}\to J_*$$
for some $\|J_*\|_\infty<\infty.$ 
Furthermore the derivative is uniformly bounded, i.e. almost surely,
\begin{align*}
\sup_{\beta}\left\|\|\partial_h m_{\beta}(Z_1;\theta^\beta_0,g^\beta_0,h_0)\|_{\infty}\right\|_{L_2}<~&\infty.
\end{align*}
Finally, the Hessian of each coordinate $t$ of the moment vector $m_{\beta}$ with respect to $h$, is uniformly bounded:
\begin{align*}
    \sup_{z\in {\cal Z}, h,\theta,g} \left\|\partial_{hh} m_{\beta,t}(z; \theta, g, h)\right\|_{op} =~& o(n^{1/2}).
\end{align*}
\end{assumption}

\begin{assumption}[Limits of $\beta$]\label{ass:mbeta-smooth}
The parameter $\beta$ grows at an appropriate rate as $n$ grows, such that functions $m_\beta(Z_1;\theta_0^\beta, g_0^\beta, h_0)$ and $A(g_0^\beta, h_0;\beta)$ each has a limit as $\beta\to\infty.$ That is, we have that for some limit functions $m_*$ and $A_*$, 
\begin{align*}&
    m_\beta(Z_1;\theta^\beta_0,g^\beta_0,h_0) - m_*(Z_1;\theta_0,g_0,h_0)~= o_p(1),
  \\&\|A(g_0^\beta, h_0;\beta) - A_*(g_0, h_0)\|_{op}=o(1).
\end{align*} 
\end{assumption}

\begin{assumption}[Orthogonality in $g$]\label{ass:ortho}
The moment satisfies the Neyman orthogonality condition with respect to nuisance $g$: for all $g\in \mcG$
\begin{align*}
    D_g M(\theta_0^\beta, g_0^\beta, h_0;\beta)[g-g_0^\beta] = \frac{\partial}{\partial t} M(\theta_0^\beta, g_0^\beta + t\, (g-g_0^\beta), h_0; \beta)\big|_{t=0} ~=~& 0
\end{align*}
and a second-order smoothness condition: for all $g\in \mcG$,  we have
\begin{align*}
   \sup_{t_0\in[0,1]}\left\{ \frac{\partial^2}{\partial t^2} M(\theta_0^\beta, g_0 + t\, (g-g_0^\beta), h_0;\beta)\big|_{t=t_0}\right\} ~=~ O_{\rm{unif}(\beta)}\left(\|g-g_0^\beta\|_2^2\right).
\end{align*}
\end{assumption}

  As $\beta$ increases the estimators and moment functions change. To be able to control all of this we need to control the impact of $\beta$ on the convergence rate.
\begin{assumption}[Rates for $g$]\label{ass:rates}
Suppose that the nuisance estimates $\hat{g}^\beta\in \mathcal{G}$ satisfy the following consistency rate:
\begin{align*}
\begin{aligned}
    \|\hat{g}^\beta - g_0^\beta\|^2_2= \E_X\{\|\hat{g}^\beta(X) - g_0^\beta(X)\|_2^2\} =o_{p,\rm{unif}(\beta)}\left(n^{-1/2}\right).
\end{aligned}
\end{align*}

\end{assumption}

\begin{assumption}[Equicontinuity]\label{ass:equicont}
Suppose that $\beta$ grows at rate such that the moment $m$ satisfies the stochastic equicontinuity conditions: 
\begin{align}\label{eqn:equicont}
\begin{aligned}
    n^{1/2}\left\|A(\hat g^{\beta},h_0;\beta) - A(g_0^\beta, h_0;\beta) - \{A_n(\hat g^{\beta},h_0;\beta) - A_n(g_0^\beta, h_0;\beta)\}\right\|_{op} =~& o_{p}(1)\\
    n^{1/2} \left\|V(\hat g^{\beta},h_0;\beta) - V(g_0^\beta, h_0;\beta) - \{V_n(\hat g^{\beta},h_0;\beta) - V_n(g_0^\beta, h_0;\beta)\}\right\|_2 =~& o_p(1).
\end{aligned} 
\end{align}
\end{assumption}

\begin{assumption}[Regularity]\label{ass:regularity}
Assume that $A_*(g_0, h_0)^{-1}$ exists, that for any $j,k\le p$ and for some $\epsilon > 0$,
\begin{align*}
\sup_{\beta>0}\|\theta_0^\beta\|_2,\sup_{\beta>0}\|a_{\beta,j,k}(Z;g_0^\beta, h_0)\|_{L_{2+\epsilon}}, \sup_{\beta>0}\|\nu_{\beta,j,k}(Z;g_0^\beta, h_0)\|_{L_{2+\epsilon}} <\infty.
\end{align*}
Assume that for any $i,j\in [p]\times[p]$, the random variables $\left(a_{\beta,i,j}(Z;g_0,h_0)\right)_{i,j}$ has bounded variances. Moreover, assume that for any $g, g'\in\mathcal{G}$:
\begin{align*}
    \sup_{\beta>0}\|A(g,h_0;\beta) - A(g',h_0;\beta)\|_{op} = O(\|g - g'\|_2).
\end{align*}
\end{assumption}

\begin{theorem}\label{thm:general-linearity} Under Assumptions~\ref{ass:influence}, \ref{ass:mbeta-smooth}, \ref{ass:ortho}, \ref{ass:rates}, \ref{ass:equicont}, \ref{ass:regularity}, if the nuisance parameter estimate $\hat{h}$ is asymptotically linear with some influence function $f_h$ such that $\E[f_h(Z_i)] = 0$:
\begin{align*}
    \hat h - h_0 = n^{-1}\sum_{i=1}^n f_h(Z_i) + o_p(n^{-1/2})
\end{align*}
then the parameter estimate $\hat{\theta}^\beta$ is asymptotically linear around $\theta_0^\beta$:
\begin{align*}
    n^{1/2} (\hat{\theta}^\beta - \theta_0^\beta) = n^{-1/2} \sum_{i=1}^n \rho_{\theta}(Z_i) + o_p(1)
\end{align*}
with influence function:
\begin{align*}
    \rho_{\theta}(Z) = -  A_*(g_0,h_0)^{-1} \{m_*(Z; \theta_0, g_0,h_0) + (J^{*})^{\trans}f_h(Z)\}.
\end{align*}
\end{theorem}
\begin{remark}
    Note here if we instead exploit a cross-fitting approach which trains $\hat g^\beta$ and $\hat h$ on one half of the entire dataset while evaluating the empirical moment on the other half of the data, the theorem still applies. We will formally present a statement of Theorem~\ref{thm:general-linearity} adapted to cross-fitting approach in Appendix~\ref{app:cfit}.
\end{remark}
\subsection{Proof of Theorem~\ref{thm:general-linearity}}
\begin{proof}
By the linearity of the moment with respect to $\theta$, we have
\begin{align*}
    A(\hat g^{\beta},h_0;\beta) (\hat{\theta}^\beta-\theta_0^\beta) =~&  M(\hat{\theta}^\beta, \hat g^{\beta},h_0;\beta) - M({\theta}_0^\beta, \hat g^{\beta},h_0;\beta)\\
    =~& M({\theta}_0^\beta, g_0^{\beta},h_0;\beta) - M({\theta}_0^\beta, \hat g^{\beta},h_0;\beta)\\
    +~& M(\hat{\theta}^\beta, \hat g^{\beta},h_0;\beta) - M_n(\hat{\theta}^\beta, \hat g^{\beta},h_0;\beta)\\
    +~& M_n(\hat{\theta}^\beta, \hat g^{\beta},h_0;\beta) - M_n(\hat{\theta}^\beta, \hat g^{\beta},\hat h;\beta)\\
    +~& M_n(\hat{\theta}^\beta, \hat g^{\beta},\hat h;\beta) - M({\theta}_0^\beta, g_0^\beta,h_0;\beta)    .
\end{align*}
Note that by definition of $\hat \theta ^\beta$
 and of $\theta_0^\beta$ we have $$M({\theta}_0^\beta, g_0^\beta,h_0;\beta) ~=~ 0$$
and
$$M_n(\hat{\theta}^\beta, \hat g^{\beta},\hat h;\beta) ~=~ 0. $$
Therefore we obtain that \begin{align*}
    A(\hat g^{\beta},h_0;\beta)  (\hat{\theta}^\beta-\theta_0^\beta) 
    =~& M({\theta}_0^\beta, g_0^{\beta},h_0;\beta) - M({\theta}_0^\beta, \hat g^{\beta},h_0;\beta)\\
    +~& M(\hat{\theta}^\beta, \hat g^{\beta},h_0;\beta) - M_n(\hat{\theta}^\beta, \hat g^{\beta},h_0;\beta)\\
    +~& M_n(\hat{\theta}^\beta, \hat g^{\beta},h_0;\beta) - M_n(\hat{\theta}^\beta, \hat g^{\beta},\hat h;\beta)
\end{align*} We will analyze separately each term on the right-hand side. 
In this goal and for ease of notations, we write
\begin{align*}
     I_{1,n} = M({\theta}_0^\beta, g_0^{\beta},h_0;\beta) - M({\theta}_0^\beta, \hat g^{\beta},h_0;\beta),\\
    I_{2,n} = M(\hat{\theta}^\beta, \hat g^{\beta},h_0;\beta) - M_n(\hat{\theta}^\beta, \hat g^{\beta},h_0;\beta),\\
    I_{3,n} = M_n(\hat{\theta}^\beta, \hat g^{\beta},h_0;\beta) - M_n(\hat{\theta}^\beta, \hat g^{\beta},\hat h;\beta).
\end{align*}
We will prove that the terms $I_{1,n}$ is negligible and re-express the terms $I_{2,n}$ and $I_{3,n}.$ Moreover, we remark that by assumption~\ref{ass:mbeta-smooth} we have $\|A(g_{0}^{\beta},h_0;\beta)  - A_*(g_0, h_0)\|_{op}=o(1)$ for some limit function $A_*$, and hence by triangle inequality
\begin{align*}
    &~\|A(\hat g^{\beta},h_0;\beta) - A_*(g_0, h_0)\|_{op} \\
    \le &~\|A(\hat g^{\beta},h_0;\beta) - A(g_{0}^{\beta},h_0;\beta)\|_{op} + \|A(g_{0}^{\beta},h_0;\beta) - A_*(g_0, h_0)\|_{op}\\
    = &~ O_{p,\rm{unif}(\beta)}(\|\hat g^{\beta} - g_{0}^{\beta}\|_2) + o(1) ~=~o_{p,\rm{unif}(\beta)}(1),
\end{align*}
where the penultimate equality follows from Assumption~\ref{ass:regularity}, and the last equality holds by Assumption~\ref{ass:rates}. 
Therefore the following holds
\begin{align*}
    A(\hat g^{\beta},h_0;\beta)  (\hat{\theta}^\beta-\theta_0^\beta) =~&  A_*(g_0, h_0) (\hat{\theta}^\beta-\theta_0^\beta) + \{A(\hat g^{\beta},h_0;\beta) - A_*(g_0, h_0)\}(\hat{\theta}^\beta-\theta_0^\beta)\\
    =~& A_*(g_0, h_0) (\hat{\theta}^\beta-\theta_0^\beta) + o_{p,\rm{unif}(\beta)}(\|\hat{\theta}^\beta-\theta_0^\beta\|_2).
\end{align*}
Hence, we have
\begin{align*}
 A_*(g_0, h_0)  (\hat{\theta}^\beta-\theta_0^\beta)~=~ I_{1,n}+I_{2,n}+I_{3,n}+o_{p,\rm{unif}(\beta)}(\|\hat{\theta}^\beta-\theta_0^\beta\|_2).   
\end{align*}
In the following, we analyze the asymptotic behaviors of $I_{1,n}, I_{2,n}, I_{3,n}$ separately, starting with $I_{1,n}$. By exploiting the Neyman orthogonality assumption and the Smoothness hypothesis, we obtain that:
\begin{align*}
     M({\theta}_0^\beta, \hat g^{\beta},h_0;\beta) - M({\theta}_0^\beta, g_0^{\beta},h_0;\beta) ~=&~ D_g M({\theta}_0^\beta, g_0^\beta, h_0;\beta)[\hat g^{\beta}-g_0^\beta] + O_{p,\rm{unif}(\beta)}\left(\|\hat g^{\beta}-g_0^\beta\|_2^2\right) \\
    =&~ O_{p,\rm{unif}(\beta)}\left(\|\hat g^{\beta}-g_0^\beta\|_2^2\right) ~=~o_{p,\rm{unif}(\beta)}(n^{-1/2}).
\end{align*}
where to get the first equality we used the Smoothness condition and where to obtain the second equality we used the Neyman orthogonality assumption. This implies that the first term $I_{1,n}$ is negligible meaning that
\begin{align*}
    I_{1,n} =~& o_{p,\rm{unif}(\beta)}(n^{-1/2}).
\end{align*}

  We now move on to analyzing $I_{2,n}$. In this goal, let $G_n(\theta, g, h;\beta) = M(\theta, g,h;\beta) - M_n(\theta, g,h;\beta).$ Then we remark that $I_{2,n}$ can be reformulated as 
\begin{align*}
I_{2,n}=G_n(\hat\theta^\beta, \hat g^{\beta}, h_0;\beta).
\end{align*}

Now we decompose this empirical process into an asymptotically normal component and an asymptotically negligible part. To achieve this we use the shorthand notation:
\begin{align*}
    G_n(\theta_0, g_0, h_0;\infty) = \E\{m_*(Z; \theta_0, g_0, h_0)\}-n^{-1}\sum_{i=1}^n m_*(Z_i;\theta_0, g_0, h_0) 
\end{align*}
Indeed we remark that 
\begin{align*}
G_n(\hat\theta^\beta, \hat g^{\beta}, h_0;\beta) =~& G_n(\theta_0, g_0, h_0;\infty) 
+
\{G_n(\theta_0^\beta, g_0^\beta, h_0;\beta) - G_n(\theta_0, g_0, h_0;\infty)\} \\
+~& 
\{G_n(\theta_0^\beta, \hat g^{\beta}, h_0;\beta) - G_n(\theta_0^\beta, g_0^\beta, h_0;\beta)\} \\
+~& 
\{G_n(\hat\theta^\beta, \hat g^{\beta}, h_0;\beta) - G_n(\theta_0^\beta, \hat g^{\beta}, h_0;\beta)\}. \label{last_term}
\end{align*}
We analyze each one of those terms separately. Firstly we remark
that the second term:
\begin{align*}
G_n(\theta_0^\beta, g_0^\beta, h_0;\beta) - G_n(\theta_0, g_0, h_0;\infty) 
\end{align*}
is of the form:
\begin{align*}
    n^{-1}\sum_{i=1}^{n}\left\{m_*(Z_i;\theta_0,g_0,h_0)-m_\beta(Z_i;\theta^\beta_0,g^\beta_0,h_0)\right\}-E\Big\{m_*(Z_i;\theta_0,g_0,h_0)-m_\beta(Z_i;\theta^\beta_0,g^\beta_0,h_0)\Big\}.
\end{align*}
We easily remark that $\{m_*(Z_i;\theta_0,g_0,h_0)-m_\beta(Z_i;\theta^\beta_0,g^\beta_0,h_0)\}$ is a sequence of i.i.d observations. Therefore we have that 
\begin{align*}
&\| G_n(\theta_0^\beta, g_0^\beta, h_0;\beta) - G_n(\theta_0, g_0, h_0;\infty) \|_{L_2}\\
    &=\left(var\left[n^{-1}\sum_{i=1}^{n}\left\{m_*(Z_i;\theta_0,g_0,h_0)-m_\beta(Z_i;\theta^\beta_0,g^\beta_0,h_0)\right\}\right]\right)^{1/2}\\
    &=\left(n^{-1}var\left[m_*(Z_i;\theta_0,g_0,h_0)-m_\beta(Z_i;\theta^\beta_0,g^\beta_0,h_0)\right]\right)^{1/2}\\
    &\le n^{-1/2}\Big\|m_*(Z_i;\theta_0,g_0,h_0)-m_\beta(Z_i;\theta^\beta_0,g^\beta_0,h_0)\Big\|_{L_2} = o_\beta(n^{-1/2}),
\end{align*}
where to obtain the last identity we exploited the assumption that 
\begin{align*}
    m_\beta(Z_1;\theta^\beta_0,g^\beta_0,h_0) - m_*(Z_1;\theta_0,g_0,h_0)~= o_p(1)
\end{align*}
and Assumption~\ref{ass:regularity}. 

Therefore as long as we take $\beta\to\infty$ we obtain that
\begin{align*}
G_n(\theta_0^\beta, g_0^\beta, h_0;\beta) - G_n(\theta_0, g_0, h_0;\infty) =o_{p}(n^{-1/2}).
\end{align*}

  By the linearity of the moment, we can also write
\begin{align*}
   G_n(\hat\theta^\beta, \hat g^{\beta}, h_0;\beta) - G_n(\theta_0^\beta, \hat g^{\beta}, h_0;\beta) = \left\{A(\hat g^{\beta}, h_0;\beta) - A_n(\hat g^{\beta}, h_0;\beta)\right\}^{\trans}(\hat\theta^\beta-\theta_0^\beta).
\end{align*}

  To successfully upper-bound this term we need to upper bound $\left\|A(\hat g^{\beta}, h_0;\beta) - A_n(\hat g^{\beta}, h_0;\beta)\right\|_{op} $. In this goal, note that by the triangle inequality, the following holds
\begin{equation}\label{tr2}
\begin{aligned}
   &\left\|A(\hat g^{\beta}, h_0;\beta) - A_n(\hat g^{\beta}, h_0;\beta)\right\|_{op}\\ \leq~& \left\|A(g_0^\beta, h_0;\beta)-A_n(g_0^\beta, h_0;\beta)\right\|_{op} \\
+~&  \left\|A(\hat g^{\beta}, h_0;\beta) - A(g_0^\beta, h_0;\beta) - \{A_n(\hat g^{\beta}, h_0;\beta) - A_n(g_0^\beta, h_0;\beta)\}\right\|_{op}.
\end{aligned} 
\end{equation}
We can prove that each term on the right-hand side of \eqref{tr2} is negligible. In this goal, note that for all $j,k\le p$  we have that for the $(j, k)$-th component $a_{\beta,j,k}$ of function $a_\beta$
\begin{align*}
    &\Big\|n^{-1}\sum_{i=1}^na_{\beta,j,k}(Z_i;g_0^\beta, h_0)-E
\{a_{\beta,j,k}(Z_i;g_0^\beta, h_0)\}\Big\|_{L_2}
=\left[var\left\{n^{-1}\sum_{i=1}^na_{\beta,j,k}(Z_i;g_0^\beta, h_0)\right\}\right]^{1/2} \\
&=\left[n^{-1}var\left\{a_{\beta,j,k}(Z_i;g_0^\beta, h_0)\right\}\right]^{1/2}
\le  n^{-1/2} \|a_{\beta,j,k}(Z_1;g_0^\beta, h_0)\|_{L_2}=o(1),
\end{align*}
where the last equality is in fact a convergence that is uniform in $\beta.$
Therefore this directly implies that 
\begin{align*}
    \left\|A(g_0^\beta, h_0;\beta)-A_n(g_0^\beta, h_0;\beta)\right\|_{op} =
    o_{p,\rm{unif}(\beta)}(1)
\end{align*}
Moreover, for the sequence of $\beta$ that we chose,
according to our stochastic equicontinuity condition, we have that:
\begin{align*}
    \left\|A(\hat g^{\beta}, h_0;\beta) - A(g_0^\beta, h_0;\beta) - \left\{A_n(\hat g^{\beta}, h_0;\beta) - A_n(g_0^\beta, h_0;\beta)\right\}\right\|_{op} = o_{p}(n^{-1/2}) = o_{p}(1).
\end{align*}
Thus we get that $\left\|A(\hat g^{\beta}, h_0;\beta) - A_n(\hat g^{\beta}, h_0;\beta)\right\|_{op}=o_{p}(1)$, and therefore:
\begin{align*}
    G_n(\hat\theta^\beta, \hat g^{\beta}, h_0;\beta) - G_n(\theta_0^\beta, \hat g^{\beta}, h_0;\beta) =~& o_{p}(\|\hat\theta^\beta-\theta_0^\beta\|_2).
\end{align*}
Moreover, by using the triangle inequality, the definition of the operator norm, the condition that $\sup_{\beta>0}\|\theta_0^\beta\|_2$ is finite, and the stochastic equicontinuity conditions
we obtain that 
\begin{align*}
    &\left\|G_n(\theta_0^\beta, \hat g^{\beta}, h_0;\beta) - G_n(\theta_0^\beta, g_0^\beta, h_0;\beta)\right\|_2\\
    \le~&\left\|A(\hat g^{\beta}, h_0;\beta) - A(g_0^\beta, h_0;\beta) - \{A_n(\hat g^{\beta}, h_0;\beta) - A_n(g_0^\beta, h_0;\beta)\}\right\|_{op}\, \|\theta_0^\beta\|_2 \\
    ~& + \left\|V(\hat g^{\beta},h_0;\beta) - V(g_0^\beta, h_0;\beta) - \{V_n(\hat g^{\beta},h_0;\beta) - V_n(g_0^\beta, h_0;\beta)\}\right\|_2\\
    =~& o_{p}(n^{-1/2}). 
\end{align*}
Altogether, we obtain that
$$I_{2,n} = G_n(\theta_0, g_0, h_0;\infty) + o_{p}(n^{-1/2} + \|\hat\theta^\beta - \theta_0^\beta\|_2).$$
Finally we want to analyze the term $I_{3,n}.$ By a second order Taylor expansion of each coordinate of the moment vector and our assumption on the Hessian of each coordinate of the moment vector, we have that 
\begin{align*}
    I_{3,n} =~& M_n(\hat{\theta}^\beta, \hat g^{\beta},h_0;\beta) - M_n(\hat{\theta}^\beta, \hat g^\beta,\hat h;\beta)\\
    &=\{\partial_h M_n(\hat{\theta}^\beta, \hat g^{\beta},h_0;\beta)\}^{\trans}(h_0 - \hat h) + o(n^{1/2}) \|h_0 - \hat{h}\|_2^2.\\
    &=\left\{n^{-1}\sum_{i=1}^n \partial_h m_\beta(Z_i;\hat{\theta}^\beta, \hat g^{\beta},h_0)\right\}^{\trans}(h_0 - \hat h) + o_p(n^{-1/2}).
\end{align*}
We will first show that $\{n^{-1}\sum_{i=1}^n \partial_h m_\beta(Z_i;\hat{\theta}^\beta, \hat g^{\beta},{h}_0)\}$ concentrates to a deterministic quantity. In this goal, by exploiting the uniform continuity in $\theta, g$ assumption and linearity of the moment function we obtain that 
 \begin{align*}
 n^{-1}\sum_{i=1}^n \partial_h m_\beta(Z_i;\hat{\theta}^\beta, \hat g^\beta,{h}_0)=n^{-1}\sum_{i=1}^n \partial_h m_\beta(Z_i;{\theta}^{\beta}_0,  g^{\beta}_0, {h}_0)+O_p(\|\hat\theta^{\beta}-\theta^{\beta}_0\|_2)+o_{p,\rm{unif}(\beta)}(1).
\end{align*}
Moreover notice that the observations $\{\partial_h m_\beta(Z_i;{\theta}^{\beta}_0,  g^{\beta}_0,h_0)\}$ are i.i.d. Hence, we have that
\begin{align*}
  &\left\|n^{-1}\sum_{i=1}^n \partial_h m_\beta(Z_i;{\theta}^{\beta}_0,  g^{\beta}_0,{h}_0)-E\{\partial_h m_\beta(Z_1;{\theta}^{\beta}_0,  g^{\beta}_0,{h}_0)\}\right\|_{L_2}\\
  &=\left[var\left\{n^{-1}\sum_{i=1}^n \partial_h m_\beta(Z_i;{\theta}^{\beta}_0,  g^{\beta}_0,{h}_0)\right\}\right]^{1/2}=n^{-1/2}\left[var\{\partial_h m_\beta(Z_i;{\theta}^{\beta}_0,  g^{\beta}_0,{h}_0)\}\right]^{1/2}\\
  &\le n^{-1/2}\left\|\partial_h m_\beta(Z_1;{\theta}^{\beta}_0,  g^{\beta}_0,{h}_0)\right\|_{L_2}\leq \frac{C}{\sqrt{n}}
\end{align*}
for some universal constant $C$.
Finally, by assumption, we have assumed that
\begin{align*}
 E\{\partial_h m_\beta(Z_1;{\theta}_0^{\beta},  g^{\beta}_0,{h}_0)\}= J_* + o_\beta(1).
\end{align*}
All of this combined together implies that
\begin{align*}
 n^{-1}\sum_{i=1}^n \partial_h m_\beta(Z_i;\hat{\theta}^\beta, \hat g^{\beta},{h}_0)=J_*+o_{p,\rm{unif}(\beta)}(1)+O_p(\|\hat \theta^{\beta}-\theta_0^\beta\|).
\end{align*}
Therefore, we obtain that
\begin{align*}
I_{3,n} &= -J_*^{\trans}(\hat h - h_0) + O_p(\|h_0 - \hat h\|_2\, \|\hat \theta^{\beta}-\theta_0^\beta\|_2) + o_p(n^{-1/2}) \\
&= -J_*^{\trans}(\hat h - h_0) + o_p(n^{-1/2} + \|\hat \theta^{\beta}-\theta_0^\beta\|_2).
\end{align*}
According to the asymptotic linearity assumption we know that
$$\hat h - h_0 = n^{-1}\sum_{i=1}^n f_h(Z_i)+ o_p\left(n^{-1/2}\right). $$

This implies that 
\begin{align*}
    I_{3,n}= - n^{-1}\sum_{i=1}^n J_*^{\trans}f_h(Z_i) +o_p(n^{-1/2}+\|\hat \theta^\beta-\theta_0^{\beta}\|_2).
\end{align*}

By combining the analysis of $I_{1,n}, I_{2,n}$, and $I_{3,n}$ together, we obtain that 
\begin{align*}
    &A_*(g_0, h_0)
    \left(\hat{\theta}^\beta-\theta_0^\beta\right)\\
    ~=~& G_n(\theta_0,g_0,h_0,\infty)- n^{-1}\sum_{i=1}^n J_*^{\trans}f_h(Z_i) +o_p\left(n^{-1/2} + \|\hat{\theta}^\beta-\theta_0^\beta\|_2\right)\\
    =~&  -n^{-1}\sum_{i=1}^n (m_*(Z_i;\theta_0, g_0, h_0) + J_*^{\trans}f_h(Z_i)) +o_p\left(n^{-1/2} + \|\hat{\theta}^\beta-\theta_0^\beta\|_2\right).
\end{align*}
By assumption the matrix $A_*(g_0,h_0)$ is invertible and therefore we obtain that \begin{align*}
     \hat{\theta}^\beta-\theta_0^\beta= -& A_*(g_0, h_0)^{-1} \Big[n^{-1}\sum_{i=1}^n  \left\{m_*(Z_i;\theta_0, g_0, h_0) + J_*^{\trans}f_h(Z_i)\right\} \Big] \\
     +& o_p\left(n^{-1/2} + \|\hat{\theta}^\beta-\theta_0^\beta\|_2\right).
\end{align*}
The desired result immediately follows.
\end{proof}

\section{Proof of Main Theorem~\ref{thm:main}}

In this section we provide the full proof of the main theorem. The proof is divided into several steps. First we invoke the general asymptotic linearity Theorem~\ref{thm:general-linearity} to the dynamic treatment regime setting of Theorem~\ref{thm:main}. This requires verifying the set of high-level assumptions required in Theorem~\ref{thm:general-linearity}. We provide lemmas that verify each of these assumptions and defer the proof of each of these lemmas to the end of the section. Then we combine the result of this instantiation with the main bias lemma to provide a complete proof of Theorem~\ref{thm:main}.

\subsection{Instantiating Asymptotic Linearity Theorem}\label{sec:instant}

Our goal in this section is to apply Theorem~\ref{thm:general-linearity} to the problem of estimating the first period structural parameter $\theta_0$ of the blip function that corresponds to the optimal regime. Note that estimating $\theta_0$ falls exactly into the framework of the previous section, with $\theta_0$ and more generally $\theta_0^\beta$ being the corresponding first period structural parameters, $h_0$ being the second period structural parameter $\psi_0$ and $g_0^\beta$ being the first period nuisance functions $g_0^\beta = (q^*, p_1^*, p_{2,\beta}^*)$. Finally, the moment function is the moment $m_\beta$ presented in Equation~\eqref{eqn:apx-moment-def} in the main text and the quantities $a_{\beta}(Z; g, h)$ and $\nu_{\beta}(Z;g,h)$ correspond to:
\begin{align*}
    a_{\beta}(Z; \psi, g) =~& -\{\mu(T_1, S) - p_1(S)\}\, \{\mu(T_1, S) - p_1(S)\}^{\trans}\\
    \nu_{\beta}(Z; \psi, g) =~& \{Y - \psi^{\trans}\phi(T_2, X) + \softmax_{\tau_2\in \mcT} \psi^{\trans}\phi(\tau_2, X) - q(S) + p_2(S)\}\, \{\mu(T_1, S) - p_1(S)\}.
\end{align*}
Thus to apply Theorem~\ref{thm:general-linearity}, we need to show that all Assumptions~\ref{ass:influence}, \ref{ass:mbeta-smooth}, \ref{ass:ortho}, \ref{ass:rates}, \ref{ass:equicont}, \ref{ass:regularity} are satisfied, under the conditions of our main Theorem~\ref{thm:main}. We present lemmas verifying each of these assumptions and conclude with a corollary that is an instantiation of Theorem~\ref{thm:main} to the problem of estimating the first period structural parameter $\theta_0$. For proofs that we present in this section, we will drop all data split indices for the cross-fitting approach for simplicity.

\begin{lemma}[Verifying Assumption~\ref{ass:influence}: Influence of $\psi$]\label{lem:variance}
Under the conditions of Theorem~\ref{thm:main}, we have that Assumption~\ref{ass:influence} is satisfied for the problem of estimating the first period structural parameter. Moreover, the limit $J_*$ is of the form:
\begin{equation}\label{eqn:specific-Jstar}
    J_* = \E[\{\phi_\infty(X) - \phi(T_2, X)\}\tilde{M}^{\trans}]
\end{equation}
where if we denote $\mathcal{M}(X)=\{\tau: \psi_0^{\trans}\phi(\tau, X) = \max_{t}\psi_0^{\trans}\phi(t, X)\}$,
\begin{align*}
    \phi_\infty(X)=|\mathcal{M}(X)|^{-1}\sum_{\tau\in\mathcal{M}(X)}\phi(\tau, X).
\end{align*}
\end{lemma}
  We will highlight the proof of Lemma~\ref{lem:variance} in Section~\ref{sec:main-var}.

\begin{lemma}[Verifying Assumption~\ref{ass:mbeta-smooth}: Limits of $\beta$]\label{lem: limitbeta}
Under the conditions of Theorem~\ref{thm:main}, we have that Assumption~\ref{ass:mbeta-smooth} is satisfied for the problem of estimating the first period structural parameter, with limit functions:
\begin{align}
m_*(Z;\theta, \psi, g) =~& \{\epsilon_1(\theta, \psi) - q(S) + p_2(S) + \theta^{\trans}p_1(S) \} \{\mu(T_1, S) - p_1(S)\} \label{eqn:specific-mstar} \\
A_*(\psi, g) =~& -\E\left[\{\mu(T_1, S) - p_1(S)\}\, \{\mu(T_1, S) - p_1(S)\}^{\trans}\right] 
\end{align}
where $\epsilon_1(\theta, \psi)$ is defined in Equation~(2).
\end{lemma}
\begin{proof}
Note that for any $Z$ we have that $m_\beta(Z; \theta_0, g_0, h_0)$ converges as $\beta\to \infty$ to the original moment with the maximum instead of the softmax, simply because the softmax converges to the max as $\beta\to \infty$, i.e. $\softmax_{\tau_2\in \mcT} \psi^{\trans}\phi(\tau_2, X)\to \max_{\tau_2\in \mcT} \psi^{\trans}\phi(\tau_2, X)$. Moreover, note that in this case the quantity $a_{\beta}$ is independent of $\beta$, thus the second property in Assumption\ref{ass:mbeta-smooth} is trivially satisfied. 
\end{proof}

\begin{lemma}[Verifying Assumption~\ref{ass:ortho}: Orthogonality in $g$]\label{lem: neyman}
Under the conditions of Theorem~\ref{thm:main}, we have that Assumption~\ref{ass:ortho} is satisfied for the problem of estimating the first period structural parameter.
\end{lemma}
  We will provide a proof of Lemma~\ref{lem: neyman} in Appendix~\ref{sec:neyman}.

\begin{lemma}[Verifying Assumption~\ref{ass:rates}: Rates for $g$]\label{lem:ratesg}
Under the conditions of Theorem~\ref{thm:main}, we have that Assumption~\ref{ass:rates} is satisfied for the problem of estimating the first period structural parameter.
\end{lemma}
  This proof of this lemma is not immediate. For $\hat p_{2,\beta}$, we note that $\hat p_{2,\beta} = \hat{p}_{2,\beta,\hat\psi}$ is not a direct estimate of 
$p_{2,\beta}^*(S)=p_{2,\beta,\psi_0}^*(S)=\E[\psi_0^{\trans}\{\phi(T_2, X)-\phi_{\beta,\psi_0}(X)\}\mid S]$, since it has used $\hat\psi$ in place of $\psi_0$ in the regression algorithm.
We will present a lemma in Section~\ref{sec:betaun} that states that the consistency rate condition $\|\hat p_{2,\beta} - p_{2,\beta}^*\|_2 = o_{p,\rm{unif}(\beta)}(n^{-1/4})$ is in fact implied by a uniform consistency rate of the regression algorithm itself, as is assumed in the main Theorem~\ref{thm:main}.

\begin{lemma}[Verifying Assumption~\ref{ass:equicont}: Equicontinuity]\label{lem:stab-stoch}
Under the conditions of Theorem~\ref{thm:main}, we have that Assumption~\ref{ass:equicont} is satisfied for the problem of estimating the first period structural parameter.
\end{lemma} 
  In Section~\ref{sec:equicont}, we will establish that when the nuisance function space is of low statistical complexity or when we use the cross-fitted estimation approach, we will automatically have that stochastic equicontinuity holds, and hence Lemma~\ref{lem:stab-stoch} holds.

\begin{lemma}[Verifying Assumption~\ref{ass:regularity}: Regularity]\label{lem:regul}
Under the conditions of Theorem~\ref{thm:main}, we have that Assumption~\ref{ass:regularity} is satisfied for the problem of estimating the first period structural parameter.
\end{lemma}
\begin{proof}
That $A_*(g_0, h_0)^{-1}$ exists follows immediately from the fact that $\E(\tilde{M}\tilde{M}^{\trans})$ is strictly positive definite. The boundedness conditions in Assumption~\ref{ass:regularity} follow from triangle inequality and our boundedness conditions listed in out main theorems. The bounded variances conditions naturally follow since $\E(\tilde{M}\tilde{M}^{\trans})$ is bounded. The last condition in Assumption~\ref{ass:regularity} follows from almost sure boundedness of $\tilde{M}.$
\end{proof}

\begin{corollary}\label{thm:specific-linearity}
Suppose estimator $\hat\psi$ has influence function $\rho_\psi$
$$n^{1/2}(\hat\psi - \psi_0) = n^{-1/2}\sum_{i = 1}^n\rho_\psi(Z_i) + o_p(1).$$
with $\E\{\rho_\psi(Z)\}=0$. Under the conditions of Theorem~\ref{thm:main}, the estimate $\hat{\theta}^\beta$ is asymptotically linear around $\theta_0^\beta$:
\begin{align*}
    n^{1/2} (\hat{\theta}^\beta - \theta_0^\beta) = n^{-1/2} \sum_{i=1}^n \rho_{\theta}(Z_i) + o_p(1)
\end{align*}
with influence function:
\begin{align*}
    \rho_{\theta}(Z) =   \E(\tilde{M} \tilde{M}^{\trans})^{-1} \{m_*(Z; \theta_0, \psi_0, g_0) + J_*^{\trans}\rho_{\psi}(Z)\}
\end{align*}
where $m_*$ and $J_*$ are limits as defined in Equation~\eqref{eqn:specific-mstar} and Equation~\eqref{eqn:specific-Jstar} respectively.
\end{corollary}
\begin{proof}
By Lemmas~\ref{lem:variance}, \ref{lem: limitbeta}, \ref{lem: neyman}, \ref{lem:ratesg}, \ref{lem:stab-stoch}, \ref{lem:regul}, we know that under the conditions of Theorem~\ref{thm:main}, all Assumptions~~\ref{ass:influence}, \ref{ass:mbeta-smooth}, \ref{ass:ortho}, \ref{ass:rates}, \ref{ass:equicont}, \ref{ass:regularity} are satisfied. Hence Corollary~\ref{thm:specific-linearity} follows.
\end{proof}
\begin{remark}
    We comment that previous work in the literature (e.g. \cite{chernozhukov2020adversarial}) has established that either under the conditions of small critical radius in Theorem~\ref{thm:main} or under the cross-fitting approach in Theorem~\ref{thm:main2}, the asymptotic normality of second period structural parameter estimator $\hat\psi$ is indeed true. We will include more details in Sections~\ref{sec: pf_main} and \ref{sec: pf_main2}.
\end{remark}

\subsection{Proof of Theorem~\ref{thm:main}: Asymptotic Linearity of First-Period Parameter}\label{sec: pf_main}
\begin{proof}
    The results of \cite{chernozhukov2020adversarial} establishes asymptotic linearity of $\hat\psi$:
    \begin{equation}\label{eq: an1}
    \sqrt{n}(\hat\psi-\psi_0) = n^{-1/2}\sum_{i=1}^n\rho_\psi(Z_i) + o_p(1)
    \end{equation}
    where if we define $\tilde{P}=\phi(T_2, X) - \E\{\phi(T_2, X)\mid X\}$:
    $$\rho_\psi(Z)=\E(\tilde{P}\tilde{P}^{\trans})^{-1}\{Y-\E(Y\mid X) - \psi_0^{\trans}\tilde{P}\}\tilde{P}.$$
    Moreover, Lemma~\ref{lem:bias} gives us the softmax bias control that
    $$\sqrt{n}(\theta_0^\beta-\theta_0) = o(1),$$
    and Corollary~\ref{thm:specific-linearity} gives us asymptotic linearity of $\hat{\theta}^\beta$ around $\theta_0^\beta$:
    \begin{align*}
    \sqrt{n} \left(\hat{\theta}^\beta - \theta_0^\beta\right) = n^{-1/2} \sum_{i=1}^n \rho_{\theta}(Z_i) + o_p(1)
\end{align*}
where
\begin{align*}
    \rho_{\theta}(Z) =   \E(\tilde{M} \tilde{M}^{\trans})^{-1} \{m_*(Z; \theta_0, g_0,h_0) + J_*^{\trans}\rho_{\psi}(Z)\}.
\end{align*}
Combining these two results, we overall can conclude that
\begin{equation}\label{eq: an2}
    \sqrt{n} \left(\hat{\theta}^\beta - \theta_0\right) = n^{-1/2} \sum_{i=1}^n \rho_{\theta}(Z_i) + o_p(1).
\end{equation}
We can hence use the two asymptotic linearity statements \eqref{eq: an1} and \eqref{eq: an2} to construct confidence intervals for $\psi_0$ and $\theta_0.$
\end{proof}

\subsection{Proof of Lemma~\ref{lem:variance}: Main Variance Lemma}\label{sec:main-var}

We need to verify that the functions $\{\partial_\psi m_{\beta}(z;\cdot,\psi_0,\cdot)\}$ are equicontinuous and that the expected gradient $\E\{\partial_\psi m_{\beta}(Z;\theta_0^\beta, \psi_0, g_0^\beta)\}$ converges to some finite limit as $\beta\to \infty$, and that the $L_2$ norm of the derivative is bounded uniformly in $\beta$, i.e.
\begin{align*}
\sup_{\beta}\left\|\|\partial_\psi m_{\beta}(Z_1;\theta^\beta_0, \psi_0, g^\beta_0)\|_\infty\right\|_{L_2}<~&\infty,
\end{align*}
and, finally, that the Hessian of each coordinate $t$ of the moment vector $m_{\beta}$ with respect to $\psi$, is uniformly bounded: almost surely,
\begin{align*}
    \sup_{z\in {\cal Z}, \psi,\theta,g} \left\|\partial_{\psi\psi} m_{\beta,t}(z; \theta, \psi,g)\right\|_{op} =~& o(n^{1/2}).
\end{align*}
We note that $\psi$ appears in the moment $m_{\beta}(Z;\theta, \psi, g)$ in a term of the form:
\begin{align*}
    u_{\beta}(Z; \theta, \psi, g) = \{\mu(T_1, S) - p_1(S)\}\, \{\softmax_{\tau\in \mcT} \psi^{\trans}\phi(\tau, X) - \psi^{\trans}\phi(T_2,X)\}\, 
\end{align*}
Thus to understand $\partial_{\psi} m_{\beta}$, we need to analyze the derivative of the softmax, with respect to the vector $\psi$. To express this derivative it helps to recall the following definition for the weights that the softmax assigns to each treatment: 
\begin{align*}
    W_\tau^{\beta,\psi} =~& \frac{\exp\{\beta\, \psi^{\trans}\phi(\tau, X)\}}{\sum_{t\in\mcT} \exp\{\beta\, \psi^{\trans}\phi(t, X)\}}.
\end{align*}
And introduce the shorthand notation for the softmax random variable:
\begin{equation}\label{eq: V_smax}
    V_{\beta,\psi} =~ \softmax_{\tau \in \mcT} \psi^{\trans}\phi(\tau, X) = \sum_{\tau \in \mcT} W_{\tau}^{\beta, \psi} \psi^{\trans}\phi(\tau, X).
\end{equation}
Then we will use throughout the following useful lemma:
\begin{lemma}[Derivative of Softmax]
The derivative of the softmax $V_{\beta,\psi}(X)$ with respect to $\psi$ is of the form:
\begin{equation}\label{eq: decom}
    \partial_{\psi} V_{\beta,\psi} = \sum_{\tau\in \mcT} W_{\tau}^{\beta, \psi} \phi(\tau, X) + \sum_{\tau\in \mcT}  W_{\tau}^{\beta,\psi} \beta \psi^{\trans}\phi(\tau, X)  \left\{\phi(\tau, X) - \sum_{t\in\mcT} W_{t}^{\beta,\psi}\phi(t, X)\right\}
\end{equation}
where we denote the first term on the right hand side of the equation, i.e. the soft argmax of feature map, as:
$$A_{\beta,\psi}=\sum_{\tau\in \mcT} W_{\tau}^{\beta, \psi} \phi(\tau, X)$$
and denote the second term, i.e. the derivative of softmax weights, as:
$$Q_{\beta,\psi}=\sum_{\tau\in \mcT}  W_{\tau}^{\beta,\psi} \beta \psi^{\trans}\phi(\tau, X)  \left\{\phi(\tau, X) - \sum_{t\in\mcT} W_{t}^{\beta,\psi}\phi(t, X)\right\}.$$
\end{lemma}
\begin{proof}
The proof follows by a simple chain rule of differentiation and numerical manipulations.
\end{proof}
  From this lemma we get that:
\begin{align*}
    \partial_{\psi} m_{\beta}(Z;\theta, \psi, g) =  \{\mu(T_1, S) - p_1(S)\}\, \{A_{\beta,\psi} + Q_{\beta, \psi} - \phi(T_2, X)\}^{\trans}
\end{align*}
Our goal is to analyze the properties of this function. First note that the term $\|A_{\beta,\psi}\|_2$ is upper bounded by $\max_{\tau}\|\phi(\tau, X)\|_2$. Second we show a similar upper bound for the quantity $Q_{\beta,\psi}$. In that respect it will be easy to re-write the derivative of the soft-max in a manner that is more amenable to analysis:
\begin{lemma}[Re-writting $Q_{\beta,\psi}$]\label{lem:rewritting} Let $U_{\tau,\psi}^\beta=\beta \psi^{\trans}\phi(\tau,X)$. Then we have:
\begin{align*}
    Q_{\beta,\psi}&
    = \sum_{\tau\in\mcT} \phi(\tau, X) \sum_{t\in\mcT} (U_{\tau,\psi}^\beta - U_{t,\psi}^\beta) \frac{\exp(U_{t,\psi}^\beta)\, \exp(U_{\tau,\psi}^\beta)}{\left\{\sum_{t\in\mcT} \exp(U_{t,\psi}^\beta)\right\}^2}\, 
\end{align*}
\end{lemma}
\begin{proof}
For simplicity, for this proof we let $U_\tau =  \beta \psi^{\trans}\phi(\tau, X)$ and  $Q_{\beta}=Q_{\beta,\psi}$ and $W_{\tau}^{\beta} = W_{\tau}^{\beta,\psi}$. First we note that $Q_{\beta}$ can be re-written in a more convenient manner for analysis, by simply re-arranging the sums and doing a change of variable names:
\begin{align*}
    Q_{\beta} =~& \sum_{\tau\in\mcT} W_{\tau}^{\beta}\, U_\tau\, \left\{\phi(\tau, X) - \sum_{\tau'\in\mcT} W_{\tau'}^{\beta} \phi(\tau', X)\right\} \\
     =~& \sum_{\tau\in\mcT} W_{\tau}^{\beta}\, U_\tau\, \phi(\tau, X) - \sum_{\tau\in\mcT} W_{\tau}^{\beta}\, U_\tau \sum_{\tau'\in\mcT} W_{\tau'}^{\beta} \phi(\tau', X) \\
     =~& \sum_{\tau\in\mcT} W_{\tau}^{\beta}\, U_\tau\, \phi(\tau, X) - \sum_{\tau'\in\mcT} W_{\tau'}^\beta \phi(\tau', X) \sum_{\tau\in\mcT} W_{\tau}^{\beta}\, U_\tau \\
     =~& \sum_{\tau\in\mcT} W_{\tau}^{\beta}\, U_\tau\, \phi(\tau, X) - \sum_{\tau\in\mcT} W_{\tau}^\beta \phi(\tau, X) \sum_{\tau'\in\mcT} W_{\tau'}^{\beta}\, U_{\tau'} \\
     =~& \sum_{\tau\in\mcT} W_{\tau}^{\beta}\, \left\{U_\tau - \sum_{\tau'\in\mcT} W_{\tau'}^{\beta}\, U_{\tau'}\right\}\, \phi(\tau, X).
\end{align*}
Expanding the definition of the softmax weights $W_{\tau}^\beta$ we get:
\begin{align*}
    J_\tau^\beta :=~& W_{\tau}^{\beta}\, \left\{U_\tau - \sum_{\tau'} W_{\tau'}^{\beta}\, U_{\tau'}\right\}\\ =~&\frac{\exp(U_\tau)}{\sum_{t\in\mcT} \exp(U_t)}\, \left\{U_\tau - \sum_{t\in\mcT} \frac{U_t \exp(U_t)}{\sum_{t'\in\mcT} \exp(U_t)}\right\}
    = \sum_{t\in\mcT} (U_\tau - U_t) \frac{\exp(U_t)\, \exp(U_\tau)}{\left\{\sum_{t'\in\mcT} \exp(U_t)\right\}^2}
\end{align*}
\end{proof}

\begin{lemma}[Boundedness of $Q_{\beta,\psi}$] The quantity $Q_{\beta,\psi}$ is absolutely bounded as:
\begin{align*}
    \|Q_{\beta,\psi}\|_2 \leq \frac{K}{e}\, \sum_{\tau\in\mcT} \|\phi(\tau, X)\|_2
\end{align*}
where $K=|\mathcal{T}|$ is the number of treatments.
\end{lemma}
\begin{proof}
For simplicity, for this proof we let $U_\tau =  \beta \psi^{\trans}\phi(\tau, X)$ and  $Q_{\beta}=Q_{\beta,\psi}$ and $W_{\tau}^{\beta} = W_{\tau}^{\beta,\psi}$. We note that:
\begin{equation}\label{J}
\begin{aligned}
    |J_\tau^\beta| =~& \left|\sum_{t\in\mcT} (U_\tau - U_t) \frac{\exp(U_t)\, \exp(U_\tau)}{\left\{\sum_{t\in\mcT} \exp(U_t)\right\}^2}\right|\\
    \leq~& \sum_{t\in\mcT} |U_\tau - U_t| \frac{\exp(U_t)\, \exp(U_\tau)}{\left\{\sum_{t\in\mcT} \exp(U_t)\right\}^2}\\
    \leq~& \sum_{t\in\mcT} |U_\tau - U_t| \frac{\exp(U_t + U_\tau)}{\exp(2 U_t) + \exp(2 U_\tau)}\\
    \leq~& \sum_{t\in\mcT} |U_\tau - U_t| \frac{1}{\exp(U_t - U_\tau) + \exp(U_\tau - U_t)}\\
    \leq~& \sum_{t\in\mcT}  \frac{|U_\tau - U_t|}{\exp(|U_t - U_\tau|)} \leq \frac{K}{e}.
\end{aligned}
\end{equation}
Thus we have:
\begin{align*}
    \|Q_{\beta}\|_2 \leq \sum_{\tau\in\mcT} |J_\tau^\beta| \, \|\phi(\tau, X)\|_2 \leq \frac{K}{e}\, \sum_{\tau\in\mcT} \|\phi(\tau, X)\|_2.
\end{align*}
\end{proof}

  Thus we get that the derivative $\partial_{\psi} m_\beta$ is Lipschitz in $p_1$ with Lipschitz constant:
\begin{align*}
    2\max_{\tau} \|\phi(\tau, X)\|_2 + \frac{K}{e} \sum_{\tau\in\mcT} \|\phi(\tau, X)\|_2 < C
\end{align*}
for some universal constant $C$, by the assumptions of our main Theorem~\ref{thm:main}. Moreover, note that the gradient $\partial_{\psi} m_\beta$ is independent of $\theta$. Thus we get that the gradient is equicontinuous in both $\theta$ and $g$, which is the first condition of Assumption~\ref{ass:influence}.
Moreover, by the aforementioned bounds we get that:
\begin{align*}
    \|\partial_\psi m_{\beta}(Z;\theta_0^\beta, \psi_0, g_0^\beta)\|_\infty =~& \|\tilde{M} \{A_{\beta,\psi_0} + Q_{\beta, \psi_0}-\phi(T_2,X)\}^{\trans}\|_\infty \\
    \leq~& \|\tilde{M}\|_{\infty} \|A_{\beta,\psi_0} + Q_{\beta, \psi_0}-\phi(T_2,X)\|_\infty\\
    \leq~& \|\tilde{M}\|_{\infty}\, C \leq C'
\end{align*}
by our assumption of our main Theorem~\ref{thm:main} on the boundedness of $\|\tilde{M}\|_{2}$. Thus we get that the third condition of Assumption~\ref{ass:influence} is satisfied.

  Next we argue the limit behavior of $A_{\beta,\psi}$ and $Q_{\beta,\psi}$. Note that the first term $A_{\beta,\psi}$ in the derivative converges trivially, by the definition of the softmax to the feature map of the best action under $\psi$; or a uniform distribution over the best actions if there are ties. Hence, this implies that as $\beta\to \infty$, $ A_{\beta,\psi} \to \phi_\infty(X)$ where
\begin{align*}
    \phi_\infty(X) = \frac{1}{|\arg\max_{\tau} \psi^{\trans}\phi(\tau, X)|} \sum_{\tau^* \in \arg\max_{\tau} \psi^{\trans}\phi(\tau, X)} \phi(\tau^*, X).
\end{align*}
The second term $Q_{\beta,\psi}$ is the problematic one, as it involves how the argmax changes as a function of $\psi$. However, we show that the second term $Q_{\beta,\psi}$ converges to zero.

\begin{lemma}[Key Variance Building Block]
The term $Q_{\beta,\psi}$ in the derivative of the softmax $\partial_{\psi} V_{\beta,\psi}$, converges to zero, i.e. as $\beta\to \infty$:
\begin{align*}
    Q_{\beta,\psi} \to 0.
\end{align*}
\end{lemma}
\begin{proof}
For simplicity, for this proof we let $U_\tau =  \beta \psi^{\trans}\phi(\tau, X)$ and  $Q_{\beta}=Q_{\beta,\psi}$ and $W_{\tau}^{\beta} = W_{\tau}^{\beta,\psi}$ and
\begin{align*}
    J_\tau^\beta =~& \sum_{t\in\mcT} (U_\tau - U_t) \frac{\exp(U_t)\, \exp(U_\tau)}{\left\{\sum_{t'\in\mcT} \exp(U_{t'})\right\}^2}
\end{align*}
It suffices to show that:
\begin{align*}
    \lim_{\beta \to \infty} |J_\tau^\beta | = 0
\end{align*}
Since:
\begin{align*}
    0 \leq |J_\tau^\beta| \leq~& \sum_{t\in\mcT}  \frac{|U_\tau - U_t|}{\exp(|U_t - U_\tau|)} = \sum_{t\in\mcT}  \frac{\beta\, |Z_\tau - Z_t|}{\exp(\beta\, |Z_\tau - Z_t|)}
\end{align*}
where $Z_\tau = \psi^{\trans}\phi(\tau, X)$. Since almost surely:
\begin{align*}
    \lim_{\beta \to \infty} \frac{\beta\, |Z_\tau - Z_t|}{\exp(\beta\, |Z_\tau - Z_t|)} = 0
\end{align*}
we get the desired statement that $\lim_{\beta \to \infty} |J_\tau^\beta| = 0$.
\end{proof}

  Thus we have shown that the gradient converges to the limit:
\begin{align*}
    \lim_{\beta\to \infty} \partial_{\psi} m_{\beta}(Z; \theta_0^\beta, \psi_0, g_0^\beta) = \tilde{M} \{\phi(\pi_{2}^*(X), X) - \phi(T_2, X)\}
\end{align*}
and is uniformly upper bounded by a constant. Note again that if there are ties for maximizing $\psi_0^{\trans}\phi(\tau,X)$ then we take the average: 
$$\phi(\pi_{2}^*(X), X)= \frac{1}{|\arg\max_{\tau} \psi_0^{\trans}\phi(\tau, X)|} \sum_{\tau^* \in \arg\max_{\tau} \psi_0^{\trans}\phi(\tau, X)} \phi(\tau^*, X).$$
Thus by the dominated convergence theorem we get that:
\begin{align*}
    \lim_{\beta\to \infty} \E\left\{\partial_{\psi} m_{\beta}(Z; \theta_0^\beta, \psi_0, g_0^\beta)\right\} = \E[\tilde{M} \{\phi(\pi_{2}^*(X), X) - \phi(T_2, X)\}]
\end{align*}
Thus we conclude that the second condition in Assumption~\ref{ass:influence} is satisfied, with:
\begin{align*}
    J_* = \E[\tilde{M} \{\phi(\pi_{2}^*(X), X) - \phi(T_2, X)\}].
\end{align*}

  It remains to show the fourth condition of Assumption~\ref{ass:influence}.
\begin{lemma}[Order of Hessian]\label{lem: hes}
     When $\beta$ grows as $o(n^{1/2})$, the Hessian of each coordinate $t$ of the moment vector $m_{\beta}$ with respect to $\psi$, grows as:
\begin{align*}
    \sup_{z\in {\cal Z}, \psi,\theta,g\in\mathcal{G}} \left\|\partial_{\psi\psi} m_{\beta,t}(z; \theta, \psi, g)\right\|_{op} =~& o(n^{1/2})
\end{align*}
almost surely.
\end{lemma}
  For succinctness, we omit the proof here and include it in Appendix~\ref{sec: order_hes}.

\subsection{Proof of Lemma~\ref{lem: neyman}: Verifying Neyman Orthogonality}\label{sec:neyman}
\begin{proof}
To prove Neyman orthogonality it suffices to show that if we view the moment function as a function of the output of each nuisance function, then the derivative with respect to that finite dimensional vector output, conditional on the variables that correspond to the input of each function is equal to zero. In other words if we write:
\begin{align*}
    m_{\beta}(Z; \theta, \psi, g) = \tilde{m}_{\beta}(Z; \theta, \psi, (q(S), p_1(S), p_{2,\beta}(S)))
\end{align*}
then we need the function $\tilde{m}_{\beta}(z;\theta, \psi, \gamma)$ to satisfy:
\begin{align*}
    J_g(S) = \E\{\partial_{\gamma} \tilde{m}_{\beta}(Z; \theta_0^\beta, \psi_0, g_0^\beta(S))\mid S\} = 0
\end{align*}
We verify this property for each component $g=(q, p_1, p_{2,\beta})$ and each coordinate $t$ of the moment vector separately: for $t'\neq t$
\begin{align*}
    J_{q,t}(S) =~& \E[-\{\mu_t(T_1, S) - p_{1,t}^*(S)\}\mid S]=0\\
    J_{p_{1,t'},t}(S) =~& \E[
        \{\mu_t(T_1, S) - p_{1,t}^*(S)\}\, \theta_{0,t'}^\beta \mid S]=0\\
    J_{p_{1,t},t}(S) =~& \E[
        -\{\epsilon_1^\beta(\theta_0^\beta, \psi_0) - q^*(S) + p_{2,\beta}^*(S) + (\theta_0^\beta)^\top p_1^*(S)\} + \{\mu_t(T_1, S) - p_{1,t}^*(S)\}\, \theta_{0,t}^\beta \mid S]\\
        =~& \E(
        -[\epsilon_1^\beta(\theta_0^\beta, \psi_0) - \E\{\epsilon_1^\beta(\theta_0^\beta, \psi_0)\mid S\}] + \{\mu_t(T_1, S) - p_{1,t}^*(S)\}\, \theta_{0,t}^\beta \mid S)=0\\
    J_{p_{2,\beta}, t}(S) =~& \E\{\mu_t(T_1, S) - p_{1, t}^*(S)\mid S\} = 0
\end{align*}
Similarly, for the smoothness part, it suffices to bound the operator norm of the conditional Hessian, for each coordinate of the moment:
\begin{align*}
    H_{g,t}(S) = \E\{\partial_{\gamma\gamma} \tilde{m}_{\beta, t}(Z; \theta_0^\beta, \psi_0, g_0^\beta(S))\mid S\}
\end{align*}
We analyze each block of the Hessian separately: for $t'\neq t$
\begin{align*}
    H_{q,q,t}(S) =~& 0, & 
    H_{q,p_{1, t'}, t}(S) =~& 0, & 
    H_{q, p_{1, t}, t}(S) =~& 1, & 
    H_{q, p_{2,\beta}, t}(S) =~& 0\\
    & & 
    H_{p_{1, t'},p_{1, t'}, t} =~& 0, &
    H_{p_{1, t'}, p_{1, t}, t}(S) =~& -\theta_{0,t'}^\beta, & 
    H_{p_{1, t'}, p_{2,\beta}, t}(S) =~& 0\\
    & & 
    & &
    H_{p_{1, t}, p_{1, t}, t}(S) =~& -2\theta_{0,t}^\beta, & 
    H_{p_{1, t}, p_{2,\beta}, t}(S) =~& -1\\
    & &
    & &
    & & 
    H_{p_{2,\beta}, p_{2,\beta}, t}(S) =~& 0.
\end{align*}
Since we assumed that $\sup_{\beta}\|\theta_0^\beta\|_2<\infty$ is bounded and the dimension $d$ is a constant, we get that the conditional Hessian is, uniformly over all $\beta>0$, upper bounded by a constant. Thus we get the uniform in $\beta$ smoothness property.
\end{proof}

\subsection{Proof of Lemma~\ref{lem:ratesg}: Beta-Uniform Rates on Nuisance Functions}
\label{sec:betaun}

The proof is an immediate consequence of the following lemma.

\begin{lemma}\label{lem:p2convergence}
Suppose that for some open neighborhood $\mathcal{N}$ of $\psi_0$ in the nuisance space, for any $\tau\in\mathcal{T},$ we have that $\phi(\tau,X)$ and $\sup_{\psi\in\mathcal{N}}|\psi^{\trans}\phi(\tau,X)|$ are uniformly bounded. Let $\hat p_{2,\beta,\psi}(S)$ denote the estimator for $$p_{2,\beta,\psi}^*(S)=\E[\psi^{\trans}\{\phi(T_2, X) - \phi_\beta(X)\}\mid S].$$
Note that then our estimator $\hat p_{2,\beta} = \hat p_{2,\beta,\hat\psi}.$ For any $\psi\in\mathcal{N}$, define
$$\|\hat p_{2,\beta,\psi} - p_{2,\beta,\psi}^*\|_2=\left(\E_S\left[\left\{\hat p_{2,\beta,\psi}(S) - p_{2,\beta,\psi}^*(S)\right\}^2\right]\right)^{1/2}.$$
Suppose that our estimation algorithm for $\hat{p}_{2,\beta, \psi}$ satisfies the following consistency condition: for all $\epsilon>0$, there exists a neighborhood $\mathcal{N}$ of $\psi_0$: as $n\rightarrow\infty,$
\begin{align*}
\sup_{\beta>0}\ pr\left(\sup_{\psi\in\mathcal{N}}\|\hat p_{2,\beta,\psi} - p_{2,\beta,\psi}^*\|_2 \ge \epsilon n^{-1/4}\right)\to 0.
\end{align*}
Then the consistency condition required of our estimator to achieve asymptotic normality, 
\begin{align*}\|\hat{p}_{2,\beta} - p^*_{2,\beta}\|^2_2~=~ \E_X\left\{\|\hat{p}_{2,\beta}(S) - p^*_{2,\beta}(S)\|_2^2\right\} =o_{p,\rm{unif}(\beta)}\left(n^{-1/2}\right),
\end{align*}
will be satisfied.
\end{lemma}
\begin{proof}
Note that by triangle inequality we can decompose
    \begin{align*}
        &\|\hat{p}_{2,\beta} - p^*_{2,\beta}\|_2\\
        =~&\left\|\left(\hat{p}_{2,\beta,\hat\psi} - p^*_{2,\beta,\hat\psi}\right) +\left(p^*_{2,\beta,\hat\psi} - p^*_{2,\beta,\psi_0}\right)\right\|_2\\
        \le~&\left\|\hat{p}_{2,\beta,\hat\psi} - p^*_{2,\beta,\hat\psi}\right\|_2 +\left\|p^*_{2,\beta,\hat\psi} - p^*_{2,\beta,\psi_0}\right\|_2.
    \end{align*}
Now by our assumptions, for all $\epsilon > 0$, there exists a neighborhood $\mathcal{N}$ of $\psi_0$ such that as $n\to\infty,$
\begin{align*}
\sup_{\beta>0}pr\left(\sup_{\psi\in\mathcal{N}}\|\hat p_{2,\beta,\psi} - p_{2,\beta,\psi}^*\|_2 \ge \epsilon n^{-1/4}\right)\to0,
\end{align*}
which then implies that
\begin{align*}
&~\sup_{\beta>0}pr\left(\|\hat{p}_{2,\beta} - p^*_{2,\beta}\|_2 \ge 2\epsilon n^{-1/4}\right)\\
\le&~ \sup_{\beta>0}pr\left(\|\hat{p}_{2,\beta,\hat\psi} - p^*_{2,\beta,\hat\psi}\|_2 \ge \epsilon n^{-1/4}\right)
+
\sup_{\beta>0}pr\left(\|p^*_{2,\beta,\hat\psi} - p^*_{2,\beta,\psi_0}\|_2 \ge \epsilon n^{-1/4}\right)\\
\le&~ \sup_{\beta>0}pr\left(\sup_{\psi\in\mathcal{N}}\left\|\hat{p}_{2,\beta,\psi} - p^*_{2,\beta,\psi}\right\|_2 \ge \epsilon n^{-1/4}\right)
+
\sup_{\beta>0}pr\left(\|p^*_{2,\beta,\hat\psi} - p^*_{2,\beta,\psi_0}\|_2 \ge \epsilon n^{-1/4}\right)\\
&~~~+ pr\left(\hat\psi\notin\mathcal{N}\right)\\
=&~ \sup_{\beta>0}pr\left(\|p^*_{2,\beta,\hat\psi} - p^*_{2,\beta,\psi_0}\|_2 \ge \epsilon n^{-1/4}\right)+ o(1)
\end{align*}
as $n\to\infty.$
Now by Mean Value Theorem, we have that
\begin{align*}
&~\left\|p^*_{2,\beta,\hat\psi} - p^*_{2,\beta,\psi_0}\right\|_2^2~=~\E_S\left[\left\{p^*_{2,\beta,\hat\psi}(S) - p^*_{2,\beta,\psi_0}(S)\right\}^2\right]\\
~=&~\E_S\left(\left[\left\{\frac{\partial}{\partial\psi}p_{2,\beta,\psi}^*(S)\Big|_{\psi = \bar{\psi}}\right\}^{\trans}(\hat\psi - \psi_0)\right]^2\right)
\end{align*}
for some $\bar{\psi}$ between $\hat\psi$ and $\psi_0$.
Here
\begin{align*}
&~\frac{\partial}{\partial\psi}p_{2,\beta,\psi}^*(S)\\
~=&~\E\left[\{\phi(T_2, X) - \phi_{\beta,\psi}(X)\} - \sum_\tau\psi^{\trans}\phi(\tau, X) \frac{\partial}{\partial\psi}W_\tau^{\beta,\psi}\ \middle|\ S\right]\\
~=&~\E\left[\{\phi(T_2, X) - \phi_{\beta,\psi}(X)\} - \beta\sum_\tau\psi^{\trans}\phi(\tau, X) W_\tau^{\beta,\psi}\left\{\phi(\tau, X) - \sum_t W_t^{\beta,\psi}\phi(t, X) \right\}\ \middle|\ S\right],
\end{align*}
where in the first equality, we could interchange differentiation with conditional expectation by uniform boundedness conditions. By Cauchy-Schwarz inequality, we know that 
\begin{align*}
&~\E_S\left(\left[\left\{\frac{\partial}{\partial\psi}p_{2,\beta,\psi}^*(S)\Big|_{\psi = \bar{\psi}}\right\}^{\trans}(\hat\psi - \psi_0)\right]^2\right)\\
\le&~ \E_S\left\{\left\|\frac{\partial}{\partial\psi}p_{2,\beta,\psi}^*(S)\Big|_{\psi = \bar{\psi}}\right\|_2^2\cdot\left\|\hat\psi - \psi_0\right\|_2^2\right\}~=~ \E_S\left\{\left\|\frac{\partial}{\partial\psi}p_{2,\beta,\psi}^*(S)\Big|_{\psi = \bar{\psi}}\right\|_2^2\right\}\cdot\left\|\hat\psi - \psi_0\right\|_2^2.
\end{align*}
By Jensen's inequality and Tower Law, we have that
\begin{align*}
&~\E_S\{\|\frac{\partial}{\partial\psi}p_{2,\beta,\psi}^*(S)\Big|_{\psi = \bar{\psi}}\|_2^2\}\\
\le&~ \E_S(
\E[\|\{\phi(T_2, X) - \phi_{\beta,\psi}(X)\} - \beta\sum_\tau\psi^{\trans}\phi(\tau, X) W_\tau^{\beta,\psi}\{\phi(\tau, X) - \sum_t W_t^{\beta,\psi}\phi(t, X) \}\|_2^2\mid  S]
)\\
=&~\E_S\bigg[
\|\{\phi(T_2, X) - \phi_{\beta,\psi}(X)\} - \beta\sum_\tau\psi^{\trans}\phi(\tau, X) W_\tau^{\beta,\psi}\{\phi(\tau, X) - \sum_t W_t^{\beta,\psi}\phi(t, X)\}\|_2^2\bigg]\\
=&~\E_S\bigg[
\bigg\|\{\phi(T_2, X) - \phi_{\beta,\psi}(X)\} \\
&~~~~~~~~~- \beta\sum_\tau\left\{\psi^{\trans}\phi(\tau, X) -\sum_t\psi^{\trans}\phi(t, X)\right\}\, W_\tau^{\beta,\psi}\,\left\{\phi(\tau, X) - \sum_t W_t^{\beta,\psi}\phi(t, X) \right\}\bigg\|_2^2
\bigg]\\
\le&~2\E_S\Bigg[
\left\{\sum_\tau\|\phi(\tau,X)\|_2\right\}^2\\
&~~~~~~~~~~~ + \left\|\beta\sum_\tau\left\{\psi^{\trans}\phi(\tau, X) -\sum_t\psi^{\trans}\phi(t, X)\right\}\, W_\tau^{\beta,\psi}\, \left\{\phi(\tau, X) - \sum_t W_t^{\beta,\psi}\phi(t, X) \right\}\right\|_2^2
\Bigg]\\ 
\le&~2\E_S[
\{\sum_\tau\|\phi(\tau,X)\|_2\}^2+|\mathcal{T}|\exp(-1)\{\sum_\tau\|\phi(\tau,X)\|_2\}^2
]\\
=&~O_{p,\rm{unif}(\beta)}\left(1\right),
\end{align*}
where the penultimate equality follow from the fact that $\sum_\tau W_\tau^{\beta, \psi} = 1$, and the penultimate inequality follows by triangle inequality and the fact $(a+b)^2\le 2(a^2+b^2)$. The last equality follows from the same line of proof in Equation~\eqref{J}.
  Hence, we obtain that
$$\left\|p_{\beta,\hat\psi}^* - p_{\beta,\psi_0}^*\right\|_2~=~ O_{p,\rm{unif}(\beta)}(n^{-1/2})~=~o_{p,\rm{unif}(\beta)}(n^{-1/4}).$$
Hence, indeed
$$\|\hat p_{2,\beta} - p_{2,\beta}^*\|_2=o_{p,\rm{unif}(\beta)}(n^{-1/4}).$$
\end{proof}

\subsection{Proof of Lemma~\ref{lem:stab-stoch}: Critical Radius Implies Stochastic Equicontinuity}\label{sec:equicont}
In this section, we show that the assumptions in Theorem~\ref{thm:main} suffice to imply stochastic equicontinuity (Assumption~\ref{ass:equicont}). To achieve this, we first present a general lemma claiming that small critical radius would imply equicontinuity in the general setting of Theorem~\ref{thm:general-linearity}, and then we will verify that our special instantiation of dynamic treatment regime in Theorem~\ref{thm:main} falls under the scope of the general lemma, by commenting on how the uniform Lipschitz requirement in the general lemma is met.
\begin{lemma}
Assume that functions $a_\beta(Z;g,h_0)$ and $\nu_\beta(Z;g,h_0)$ are uniformly Lipschitz in $g$ for $g\in\mathcal{G}$ over all $\beta>0$. Assume that the the nuisance estimator is almost surely bounded: $\|\hat g^\beta(X)\|_2 \le B$ almost surely for some $B>0.$
Let $\delta_n$ bound the critical radius of $\mathcal{G}_B=\{g\in\mathcal{G}:\|g\|_2 \le B\}.$
Then if the nuisance estimator $\hat g^\beta$ satisfies the consistency rate
$$\|\hat g^\beta - g_0^\beta\|_2=o_p(n^{-1/4})$$
and if
$$\delta_{n} = o_p(n^{-1/4}),$$
we have that stochastic equicontinuity conditions (Assumption~\ref{ass:equicont}) hold:
\begin{align*}
\begin{aligned}
    n^{-1/2}\left\|\sum_{i=1}^n\left[A(\hat g^\beta,h_0;\beta)- A(g_0^\beta,h_0;\beta)-\{a_\beta(Z_i;\hat g^\beta,h_0) - a_\beta(Z_i;g_0^\beta,h_0)\}\right]\right\|_{op} =~& o_{p}(1)\\
    n^{-1/2}\left\|\sum_{i=1}^n\left[V(\hat g^\beta,h_0;\beta)- V(g_0^\beta,h_0;\beta)-\{\nu_\beta(Z_i;\hat g^\beta,h_0) - \nu_\beta(Z_i;g_0^\beta,h_0)\}\right]\right\|_{2} =~& o_{p}(1).
\end{aligned} 
\end{align*}
\end{lemma}
\begin{proof}
For simplicity, we only prove the first statement here, and similar proofs apply to function $\nu$. For the first statement, it suffices to show that for any $j,k\le p$
    \begin{align*}
        n^{-1/2}\left|\sum_{i=1}^n\left[A_{j,k}(\hat g^\beta,h_0;\beta)- A_{j,k}(g_0^\beta,h_0;\beta)-\{a_{\beta,j,k}(Z_i;\hat g^\beta,h_0) - a_{\beta,j,k}(Z_i;g_0^\beta,h_0)\}\right]\right| =~ o_p(1).
    \end{align*}
In the remainder of the proof we look at a particular $(j,k)$ and hence for simplicity we
 overload notation and we let $a_\beta = a_{\beta,j,k}$ and $A = A_{j,k}.$ 
 
For any $\zeta\in (0,1),$ let $\delta_{n,\zeta}=\delta_n + c_0\{\log(c_1/\zeta)/n\}^{1/2}$, where $\delta_n$ upper bounds the critical radius of the function class $\mathcal{G}_B$, for some appropriately defined universal constants $c_0$ and $c_1$. By Lemma~14 of \cite{foster2019orthogonal}, we know that with probability $1-\zeta$: 
\begin{align*}
    &\left|n^{-1}\sum_{i=1}^n\left[A(\hat g^\beta,h_0;\beta) - A(g_0^\beta,h_0;\beta) - \{a_\beta(Z_i;\hat g^\beta,h_0)- a_\beta(Z_i;g_0^\beta,h_0)\}\right]\right|\\
    \le&~O(\delta_{n,\zeta}\cdot\|\hat g^\beta - g_0^\beta\|_2 + \delta_{n,\zeta}^2).
\end{align*}
Taking $\zeta = n^{-1},$ we obtain that
\begin{align*}
    &~\left|n^{-1}\sum_{i=1}^n\left[A(\hat g^\beta,h_0;\beta) - A(g_0^\beta,h_0;\beta) - \{a_\beta(Z_i;\hat g^\beta,h_0)- a_\beta(Z_i;g_0^\beta,h_0)\}\right]\right|\\
    \le&~ O_p(\delta_{n,*}\cdot\| \hat g^\beta - g_0^\beta\|_2 + \delta_{n,*}^2)
\end{align*}
where we let $\delta_{n,*} = \delta_n + c_0\{\log(c_1n)/n\}^{1/2}$.
Hence, if we have $\hat g^\beta - g_0^\beta = o_p(n^{-1/4})$ and $\delta_{n} = o_p(n^{-1/4})$, and thus $\delta_{n,*} = o_p(n^{-1/4}),$ we can conclude stochastic equicontinuity.
\end{proof}
\begin{remark}
    Our dynamic treatment regime satisfies the uniform Lipschitz condition because of the almost sure boundedness of $\mu(T_1,S)$ and that of the nuisance estimates.
\end{remark}

\subsection{Proof of Lemma~\ref{lem: hes}: Order of Hessian of Softmax}\label{sec: order_hes}
\begin{proof} 
    For simplicity, define
    $$
        K_\tau^{\beta,\psi} = \phi(\tau, X) - \sum_{t\in\mcT} W_{t}^{\beta,\psi}\phi(t, X)
    $$
    and write 
    $$
        U_\tau^{\beta,\psi} =  \beta \psi^{\trans}\phi(\tau, X).
    $$
    We calculate that
    \begin{align*}
        \frac{\partial}{\partial\psi}W_\tau^{\beta,\psi} = \beta W_{\tau}^{\beta,\psi} \{\phi(\tau, X) - \sum_{t\in\mcT} W_{t}^{\beta,\psi}\phi(t, X)\} = \beta W_{\tau}^{\beta,\psi}K_\tau^{\beta,\psi}.
    \end{align*}
    Recall that the first order derivative of the moment function satisfies
    \begin{align*}
        &\partial_{\psi} m_{\beta}(z; \theta, \psi, g) = \{\mu(T_1, S) - p_1(S)\}\, \{A_{\beta,\psi} + Q_{\beta, \psi} - \phi(T_2, X)\}^{\trans}\\
        =~&\{\mu(T_1, S) - p_1(S)\}\,\{A_{\beta,\psi} + \sum_{\tau\in \mcT}  W_{\tau}^{\beta,\psi}U_\tau^{\beta,\psi} K_\tau^{\beta,\psi} - \phi(T_2, X)\}^{\trans}\\
        =~&\{\mu(T_1, S) - p_1(S)\}\,\{A_{\beta,\psi} + \sum_{\tau\in \mcT}  W_{\tau}^{\beta,\psi}(U_\tau^{\beta,\psi} - \sum_t W_t^{\beta,\psi}U_t^{\beta,\psi})K_\tau^{\beta,\psi} - \phi(T_2, X)\}^{\trans},
    \end{align*}
    where we were able to set in the last equality that
    \begin{align*}
        \sum_{\tau\in \mcT}  W_{\tau}^{\beta,\psi}U_\tau^{\beta,\psi} K_\tau^{\beta,\psi}=\sum_{\tau\in \mcT}  W_{\tau}^{\beta,\psi}(U_\tau^{\beta,\psi} - \sum_t W_t^{\beta,\psi}U_t^{\beta,\psi})K_\tau^{\beta,\psi}
    \end{align*}
    since 
    \begin{align*}
        \sum_{\tau\in \mcT}  W_{\tau}^{\beta,\psi}K_\tau^{\beta,\psi} = 0
    \end{align*}
    by exploiting the fact that $\sum_{\tau\in \mcT} W_{\tau}^{\beta,\psi} = 1.$
    Hence, differentiating with respect to $\psi$ again, by chain rule, we obtain that the second order derivative of the moment function
    \begin{align*}
        &\partial_{\psi\psi} m_{\beta,t}(z; \theta, \psi, g)~
        =~\{\mu_t(T_1, S) - p_{1,t}(S)\} \left(R_{1,\beta,\psi} + R_{2,\beta,\psi}+R_{3,\beta,\psi}-R_{4,\beta,\psi}\right) 
    \end{align*}
    where
    \begin{align*}
        R_{1,\beta,\psi}~&=~\beta\sum_{\tau\in \mcT}  W_{\tau}^{\beta,\psi} \phi(\tau,X)(K_\tau^{\beta,\psi})^{\trans},\\
        R_{2,\beta,\psi}~&=~ \beta\sum_{\tau\in \mcT}  W_{\tau}^{\beta,\psi} (U_\tau^{\beta,\psi} - \sum_t W_t^{\beta,\psi}U_t^{\beta,\psi})K_\tau^{\beta,\psi}(K_\tau^{\beta,\psi})^{\trans},
        \\
        R_{3,\beta,\psi}~&=~\beta\sum_{\tau\in \mcT}  W_{\tau}^{\beta,\psi} K_\tau^{\beta,\psi}(K_\tau^{\beta,\psi})^{\trans},\\
        R_{4,\beta,\psi}~&=~\beta\{\sum_{\tau\in \mcT}  W_{\tau}^{\beta,\psi} (U_\tau^{\beta,\psi} - \sum_t W_t^{\beta,\psi}U_t^{\beta,\psi})  \}\{\sum_{\tau\in \mcT}
        W_{\tau}^{\beta,\psi}
        \phi(\tau, X)(K_\tau^{\beta,\psi})^{\trans}\}.
    \end{align*}
We would like to establish that all of 
\begin{align*}
    \sup_{Z\in {\cal Z}, \psi}\|R_{1,\beta,\psi}\|_{op}, \sup_{Z\in {\cal Z}, \psi}\|R_{2,\beta,\psi}\|_{op}, \sup_{Z\in {\cal Z}, \psi}\|R_{3,\beta,\psi}\|_{op},\sup_{Z\in {\cal Z}, \psi}\|R_{4,\beta,\psi}\|_{op}=O_p(\beta).
\end{align*} To achieve this, it suffices to show that for any coordinates $j,k\le n,$ we have that each component
\begin{align*}
    \sup_{Z\in {\cal Z}, \psi}|R_{1,\beta,\psi,j,k}|, \sup_{Z\in {\cal Z}, \psi}|R_{2,\beta,\psi,j,k}|, \sup_{Z\in {\cal Z}, \psi}|R_{3,\beta,\psi,j,k}|, \sup_{Z\in {\cal Z}, \psi}|R_{4,\beta,\psi,j,k}|=O_p(\beta).
\end{align*}
For $R_{1,\beta,\psi}$, we note that
\begin{align*}
|R_{1,\beta,\psi,j,k}|
\le~&\beta\sum_{\tau\in \mcT} W_{\tau}^{\beta,\psi}\cdot|\phi_j(\tau,X)|\cdot|K_{\tau,k}^{\beta,\psi}|
\le\beta\{\sum_{\tau\in \mcT} \|\phi(\tau, X)\|_2\}^2
\end{align*}
where we have used the fact that
\begin{equation}\label{K}
    |K_{\tau,k}^{\beta,\psi}|\le\sum_{\tau\in \mcT}|\phi_k(\tau,X)|\le \sum_{\tau\in \mcT} \|\phi(\tau, X)\|_2.
\end{equation}
Hence, we have that
\begin{align*}
    \sup_{Z\in {\cal Z}, \psi}|R_{1,\beta,\psi,j,k}| \le \beta\left\{\sup_X\sum_{\tau\in \mcT} \|\phi(\tau, X)\|_2\right\}^2 = O_p(\beta)
\end{align*}
by boundedness conditions.
For $R_{2,\beta,\psi}$, we similarly note that
\begin{align*}
|R_{2,\beta,\psi,j,k}|
\le~\beta\sum_{\tau\in \mcT} W_{\tau}^{\beta,\psi}\cdot\left|
U_\tau^{\beta,\psi} - \sum_t W_t^{\beta,\psi}U_t^{\beta,\psi}
\right|\cdot|K_{\tau,j}^{\beta,\psi}|\cdot|K_{\tau,k}^{\beta,\psi}|.
\end{align*}
Now we comment that
\begin{align*}
    W_{\tau}^{\beta,\psi}\cdot\left|
U_\tau^{\beta,\psi} - \sum_t W_t^{\beta,\psi}U_t^{\beta,\psi}
\right| = |J_\tau^{\beta,\psi}|
\end{align*}
where 
\begin{align*}
J_\tau^{\beta,\psi}=W_{\tau}^{\beta,\psi}\left(U_\tau^{\beta,\psi} - \sum_t W_t^{\beta,\psi}U_t^{\beta,\psi}\right),
\end{align*}
and we have shown in Eqn~(\ref{J}) that
\begin{align*}
    |J_\tau^{\beta,\psi}|\le\frac{K}{e}.
\end{align*}
Hence, by Eqn~(\ref{K}) we get that
\begin{align*}
|R_{2,\beta,\psi,j,k}|
\le~\beta\cdot \frac{K}{e}\cdot\sum_{\tau\in \mcT} |K_{\tau,j}^{\beta,\psi}|\cdot|K_{\tau,k}^{\beta,\psi}|\le \beta\cdot \frac{K^2}{e}\left\{\sum_{\tau\in \mcT} \|\phi(\tau, X)\|_2\right\}^2.
\end{align*}
Hence, we have that
\begin{align*}
    \sup_{Z\in {\cal Z}, \psi}|R_{2,\beta,\psi,j,k}| \le \beta\cdot \frac{K^2}{e}\left\{\sup_X\sum_{\tau\in \mcT} \|\phi(\tau, X)\|_2\right\}^2 = O_p(\beta)
\end{align*}
by boundedness conditions.
Similarly for $R_{3,\beta,\psi}$, we have
\begin{align*}
    \sup_{Z\in {\cal Z}, \psi}|R_{3,\beta,\psi,j,k}| \le \beta\left\{\sup_X\sum_{\tau\in \mcT} \|\phi(\tau, X)\|_2\right\}^2 = O_p(\beta).
\end{align*}
For $R_{4,\beta,\psi},$ we note that
\begin{align*}
    |R_{4,\beta,\psi,j,k}| \le \beta\left(\sum_{\tau\in \mcT} J_\tau^{\beta,\psi}\right)\left\{\sum_{\tau\in \mcT} W_\tau^{\beta,\psi}\cdot|\phi_j(\tau,X)|\cdot|K_{\tau,k}^{\beta,\psi}|\right\} = O_p(\beta).
\end{align*}
By invoking Eqn~(\ref{J}) and Eqn~(\ref{K}), we have that
\begin{align*}
    |R_{4,\beta,\psi,j,k}| &\le \beta\left(\sum_{\tau\in \mcT} J_\tau^{\beta,\psi}\right)\left\{\sum_{\tau\in \mcT} W_\tau^{\beta,\psi}\cdot|\phi_j(\tau,X)|\cdot|K_{\tau,k}^{\beta,\psi}|\right\}\\
    &\le\beta\cdot\frac{K^2}{e}\cdot\left\{\sum_{\tau\in \mcT} \|\phi(\tau, X)\|_2\right\}^2 = O_p(\beta).
\end{align*}
Hence, we obtain that
\begin{align*}
    \sup_{Z\in {\cal Z}, \psi}|R_{4,\beta,\psi,j,k}| \le\beta\cdot\frac{K^2}{e}\cdot\left\{\sup_{X}\sum_{\tau\in \mcT} \|\phi(\tau, X)\|_2\right\}^2 = O_p(\beta).
\end{align*}
Hence, since $\sup_{Z\in\mathcal{Z}, g\in\mathcal{G}}\left|\mu_t(T_1,S) - {p}_{1,t}(S)\right|$ is uniformly bounded, for $\beta=o(n^{1/2})$ the lemma follows.
\end{proof}

\section{Omitted Proofs from Section~\ref{sec: aysmlinval}}
\subsection{Proof of Theorem~\ref{thm: value_var}}
\begin{proof}   
Define the following moment function 
    $$m_{V,\beta}(Z;v,\theta,\psi)=Y - \psi^{\trans}\phi(T_2, X) - \theta^{\trans}\mu(T_1, S) + \softmax_{\tau_2} \psi^{\trans}\phi(\tau_2, X) + \softmax_{\tau_1} \theta^{\trans}\mu(\tau_1, S)-v$$
    and its mean
    $$M_V(v,\theta,\psi;\beta)=\E\{m_{V,\beta}(Z;v,\theta,\psi)\}.$$
    By definition we have the following moment condition:
    $$M_V(V_\beta^*,\theta_0,\psi_0;\beta)=0.$$
    Moreover, the estimator $\hat V$ satisfies the empirical moment condition:
    $$M_{V,n}(\hat V,\hat\theta,\hat\psi;\beta)=0$$
    where
     $$M_{V,n}(v,\theta,\psi;\beta)=n^{-1}\sum_{i=1}^nm_{V,\beta}(Z_i;v,\theta,\psi).$$
Since we already established asymptotic linearity of the structural parameter estimators $\hat\theta^\beta$ and $\hat\psi$ in Theorem~\ref{thm:main}, we can thus apply Theorem~\ref{thm:general-linearity}, in particular replacing the original structural parameter $\theta_0^\beta$ with the policy value parameter $V^*_\beta$ and replacing the original nuisance parameter $h_0$ with $\theta_0,\psi_0$, while ignoring the part of the nuisance $g$ that involves orthogonality. To achieve this, we only need to check that the assumptions of Theorem~\ref{thm:general-linearity} are satisfied. For Assumption~\ref{ass:influence}, note that 
$$\partial_\psi m_{V,\beta}(Z;v,\theta,\psi)=\partial_\psi V_{\beta,\psi}-\phi(T_2,X)$$
where recall in \eqref{eq: V_smax} we have defined
$$V_{\beta,\psi}=\softmax_{\tau_2} \psi^{\trans}\phi(\tau_2, X).$$
Hence, the derivative $\partial_\psi m_\beta(Z;v,\theta_0,\psi_0)$ does not depend on $v$ and is thus automatically equicontinuous with respect to $v$. Moreover, as we already established in Section~\ref{sec:main-var}, we have that 
$$\E(\partial_\psi V_{\beta,\psi}) = \E(A_{\beta,\psi}) + \E(Q_{\beta,\psi})\to\E\{\phi(\pi_2^*(X),X)\}$$
as $\beta\to\infty$, where $A_{\beta,\psi}, Q_{\beta,\psi}$ are the terms that we defined in the decomposition \eqref{eq: decom}. Hence, we obtain that
$$\E\{\partial_\psi m_{V,\beta}(Z;V_\beta^*,\theta_0,\psi_0)\}\to\E\{\phi_\infty(X) - \phi(T_2,X)\}$$
as $\beta\to\infty,$ 
where recall that
$$
    \phi_\infty(X)=|\mathcal{M}(X)|^{-1}\sum_{\tau\in\mathcal{M}(X)}\phi(\tau, X)
$$
for $\mathcal{M}(X)=\{\tau: \psi_0^{\trans}\phi(\tau, X) = \max_{t}\psi_0^{\trans}\phi(t, X)\}.$

The third and fourth conditions of Assumption~\ref{ass:influence} (uniform boundedness of the derivative and order of Hessian condition) follow similarly as in Section~\ref{sec:main-var}. Analogously, we can also show that the derivative of the moment function with respect to the first-period structural parameter $\partial_\theta m_{V,\beta}(Z;v,\theta,\psi)$ also satisfies Assumption~\ref{ass:influence} with the limit
$$\lim_{\beta\to\infty}\E\{\partial_\theta m_\beta(Z;V_\beta^*,\theta_0,\psi_0)\}=\E\{\mu_\infty(S)-\mu(T_1,S)\},$$
where recall that
$$\mu_\infty(S)=|\mathcal{L}(S)|^{-1}\sum_{\tau\in\mathcal{L}(S)}\mu(\tau, S),$$
for ${\cal L}(S)=\{\tau: \theta_0^{\trans}\mu(\tau, S) = \max_{t}\theta_0^{\trans}\mu(t, S)\}$.

By dominated convergence theorem, the first condition of Assumption~\ref{ass:mbeta-smooth} holds with 
$$m_{V,\beta}(Z_1;V_\beta^*,\theta_0,\psi_0) - m_{V,*}(Z_1;V^*,\theta_0,\psi_0)$$
where 
$$m_{V,*}(Z;v,\theta,\psi)=Y-\psi^{\trans}\phi(T_2,X)-\theta^{\trans}\mu(T_1,S) + \max_{\tau_2}\psi^{\trans}\phi(\tau_2,X) +  \max_{\tau_1} \theta^{\trans}\mu(\tau_1,S) - v.$$
To check the second condition of Assumption~\ref{ass:mbeta-smooth}, it suffices to note the linear form of the moment function $m_{V,\beta}$:
\begin{align*}
    m_{V,\beta}(Z; v, \theta, \psi) = a_{V,\beta}(Z;\theta, \psi)\, v + \nu_{V,\beta}(Z;\theta, \psi)
\end{align*}
where
$$a_{V,\beta}(Z;\theta, \psi)=-1$$
and
$$\nu_{V,\beta}(Z;\theta, \psi)=Y - \psi^{\trans}\phi(T_2, X) - \theta^{\trans}\mu(T_1, S) + \softmax_{\tau_2} \psi^{\trans}\phi(\tau_2, X) + \softmax_{\tau_1} \theta^{\trans}\mu(\tau_1, S).$$ 
Assumptions~\ref{ass:ortho}, \ref{ass:rates}, \ref{ass:equicont} trivially hold true because there is no nuisance parameters with respect to which the moment $M_V$ is Neyman orthogonal. Assumption~\ref{ass:regularity} is also simple to check with conditions of Theorem~\ref{thm:main} and our new regularity conditions.
\end{proof}

\subsection{Proof of Lemma~\ref{lem: value_bias}}
\begin{proof}
We note that the difference 
$$V_\beta^* - V^* = \E\{\softmax_{\tau_2}\psi_0^{\trans}\phi(\tau_2,X) - \max_{\tau_2}\psi_0^{\trans}\phi(\tau_2,X)\} + \E\{\softmax_{\tau_1}\theta_0^{\trans}\mu(\tau_1,S) - \max_{\tau_1}\theta_0^{\trans}\mu(\tau_1,S)\}.$$
For the first term on the right hand side, we notice that
\begin{align*}
&\left|\E\{\softmax_{\tau_2}\psi_0^{\trans}\phi(\tau_2,X) - \max_{\tau_2}\psi_0^{\trans}\phi(\tau_2,X)\}\right|\\
=&\left|\E\{\sum_{\tau\in\mathcal{T}} W_\tau^\beta (U_\tau^{\psi_0} - U_{\max}^{\psi_0})\}\right|\leq  \E\left(\sum_{\tau\in\mathcal{T}} W_\tau^\beta \left|U_\tau^{\psi_0} - U_{\max}^{\psi_0}\right|\right)
    =\sum_{\tau\in\mathcal{T}} \E\left(W_\tau^\beta \left|U_\tau^{\psi_0} - U_{\max}^{\psi_0}\right|\right). 
\end{align*}
Note that:
 \begin{align*}
    \E\left(W_\tau^\beta \left|U_\tau^{\psi_0} - U_{\max}^{\psi_0}\right|\right) =& \E\left\{\frac{\exp(\beta U_{\tau}^{\psi_0}) \left(U_{\max}^{\psi_0} - U_\tau^{\psi_0}\right)}{\sum_{t\in\mcT} \exp(\beta U_t^{\psi_0})} \right\}\\
    \leq& \E\left\{\frac{\exp(\beta U_{\tau}^{\psi_0})\left(U_{\max}^{\psi_0} - U_\tau^{\psi_0}\right)}{\exp(\beta U_{\max}^{\psi_0})} \right\}\\
    =& \E\left[\exp\{-\beta (U_{\max}^{\psi_0} - U_{\tau}^{\psi_0})\} \left(U_{\max}^{\psi_0} - U_\tau^{\psi_0}\right)\right].
 \end{align*}
 Applying Lemma~\ref{lem:second order}, for any fixed $\epsilon>0$ we have that for $n$ sufficiently large, such that $(1+\epsilon)\log(\beta)/\beta\leq c$ and $(1+\epsilon)\log(\beta) \geq 1$:
 \begin{align*}
    \E\left[\exp\{-\beta (U_{\max} - U_{\tau})\} \left(U_{\max} - U_\tau\right)\right] \leq \frac{H}{\beta^{1+\delta}} + \frac{(1+\epsilon)\log(\beta)}{\beta^{2+\epsilon}}
 \end{align*}
 for appropriately chosen $H$ and $\delta$. Hence, we know that
 $$\left|\E\{\softmax_{\tau_2}\psi_0^{\trans}\phi(\tau_2,X) - \max_{\tau_2}\psi_0^{\trans}\phi(\tau_2,X)\}\right| \le |\mathcal{T}|\cdot\left(\frac{H}{\beta^{1+\delta}} + \frac{(1+\epsilon)\log(\beta)}{\beta^{2+\epsilon}}\right).$$
Similarly, we can also obtain that
$$\left|\E\{\softmax_{\tau_1}\theta_0^{\trans}\mu(\tau_1,S) - \max_{\tau_1}\theta_0^{\trans}\mu(\tau_1,S)\} \right| \leq |\mathcal{T}|\cdot\left(\frac{\tilde H}{\beta^{1+\tilde \delta}} + \frac{(1+\epsilon)\log(\beta)}{\beta^{2+\epsilon}}\right).$$ We therefore can conclude that 
 \begin{align*}
     \sqrt{n}\left|V_\beta^* - V^*\right| \leq |\mcT| n^{1/2} \left(\frac{H}{\beta^{1+\delta}} +\frac{\tilde H}{\beta^{1+\tilde \delta}} + \frac{2(1+\epsilon)\log(\beta)}{\beta^{2+\epsilon}}\right).
 \end{align*}
If $\beta=\omega(n^{1/(2(1+\min\{\delta,\tilde \delta\}))})$, then the latter upper bound converges to $0$ as $n\to\infty$.
\end{proof}

\section{Variant of Main Theorem with Cross-Fitting}\label{app:cfit}

We consider a cross-fitted version of our estimation process. We split the entire dataset into two parts and train the nuisance functions and the second period structural parameter estimator $\hat\psi$ using one data split, while we train the first period structural parameter estimator $\hat\theta^\beta$ using the other split. We formally define the process as follows (without loss of generality we assume that the data size $n$ is even):

\paragraph{Step 1: Cross-fitted estimation procedure to obtain estimate $\hat{\psi}$ and cross-fold variants $\hat{\psi}^{(l)}$, for $l\in \{1, 2\}$, of second period structural parameter $\psi_0$.}

\begin{tabbing}
    \qquad \enspace For each fold $l\in \{1, 2\}$ of the data\\
    \qquad \qquad \enspace Construct estimate $\hat{\psi}^{(l)}$ of second period structural parameter $\psi_0$ as follows:\\
    \qquad \qquad \enspace Randomly split the data fold $S_l$ into two parts: $S_{l,1}$ and  $S_{l,2}$.\\
    \qquad \qquad \enspace For each fold $l'\in \{1, 2\}$\\
    \qquad \qquad \qquad \enspace Construct estimates $\hat{h}_{l'}^{(l)}, \hat{r}_{l'}^{(l)}$ of the nuisance functions $h^*$, $r^*$\\
    \qquad \qquad \qquad \enspace as defined in Equation~\eqref{eqn:period2-nuisance} of the main text using data fold $S_{l,l'}$.\\
    \qquad \qquad \enspace For each data point $i\in S_l$\\
    \qquad \qquad \qquad \enspace Define residuals $\check{Y}_i^{(l)}=Y_i-\hat{h}_{l'(i)}^{(l)}(X_i)$ and $\check{\Phi}_i^{(l)}=\phi(T_{2,i},X_i) - \hat{r}_{l'(i)}^{(l)}(X_i)$,\\
    \qquad \qquad \qquad \enspace where $l'(i)$ is chosen so that $i\notin S_{l,l'(i)}$.\\
    \qquad \qquad Solve $\hat{\psi}^{(l)}$ from the empirical moment equation, with respect to $\psi$:
\end{tabbing}
\begin{align*}
    n^{-1} \sum_{i\in S_l} \{\check{Y}_i^{(l)} - \psi ^{\trans}\check{\Phi}_i^{(l)}\} \check{\Phi}_i^{(l)} = 0.
\end{align*}
\begin{tabbing}
    \qquad \enspace Obtain aggregated estimate $\hat\psi$ by solving the following equation, with respect to $\psi$:
\end{tabbing}
\begin{align*}
    n^{-1}\sum_{i\in S_1} \{\check{Y}_i^{(2)} - \psi ^{\trans}\check{\Phi}_i^{(2)}\}\check{\Phi}_i^{(2)} + n^{-1}\sum_{i\in S_2} \{\check{Y}_i^{(1)} - \psi ^{\trans}\check{\Phi}_i^{(1)}\}\check{\Phi}_i^{(1)} = 0.
\end{align*}

\paragraph{Step 2: Cross-fitted estimation procedure to obtain estimate $\hat{\theta}^\beta$ of first period structural parameter $\theta_0$.}
\begin{tabbing}
   \qquad \enspace For each fold $l\in \{1, 2\}$ of the data\\
   \qquad \qquad \enspace Construct estimates $\hat{q}^{(l)}, \hat{p}_1^{(l)}$ of the nuisance functions
$q^*, p_1^*$\\
    \qquad \qquad \qquad \enspace defined in Equation~\eqref{eqn:period1-nuisance} of the main text using data from fold $S_l$.\\
    \qquad \qquad \enspace Construct estimate $\hat{p}_{2,\beta}^{(l)}$ of $p_{2,\beta}^*$. The estimate $\hat{p}_{2,\beta}^{(l)}$ is constructed by regressing\\
    \qquad \qquad \qquad \enspace $\{\hat{\psi}^{(l)}\}^{\trans}\phi(T_2, X) - \softmax_{\tau} \{\hat{\psi}^{(l)}\}^{\trans}\phi(\tau, X)$ on $S$ using data from fold $S_l$. \\
    \qquad \qquad \enspace For each data point $i\in S$\\
    \qquad \qquad \qquad \enspace Let $l$ be such that $i\notin S_\ell$.\\
    \qquad \qquad \qquad \enspace Define $\hat{Y}_i^{(l)}=Y_i - \hat{q}^{(l)}(S_i)$ and $\hat{M}^{(l)}=\mu(T_{1,i}, S_i) - \hat{p}_1^{(l)}(S_i)$.\\
    \qquad \qquad \qquad \enspace Define $\hat{\Phi}^{(l)}_i = \{\hat{\psi}^{(l)}\}^{\trans}\phi(T_{2,i}, X_i) - \softmax_{\tau} \{\hat{\psi}^{(l)}\}^{\trans}\phi(\tau, X_i) - \hat{p}_{2,\beta}^{(l)}(S_i)$.\\
    \qquad \qquad Construct $\hat{\theta}^\beta$ as the solution, with respect to $\theta$, of the equation:
\end{tabbing}
\begin{align*}
    \frac{1}{n}\sum_{i\in S_1} \{\hat{Y}^{(2)}_i - \hat{\Phi}^{(2)}_i - \theta ^{\trans}\hat{M}^{(2)}\} \hat{M}^{(2)}+\frac{1}{n}\sum_{i\in S_2} \{\hat{Y}^{(1)} - \hat{\Phi}^{(1)} - \theta ^{\trans}\hat{M}^{(1)}\} \hat{M}^{(1)} = 0.
\end{align*}

\begin{theorem}[Main Theorem with Crossfitting]\label{thm:main2}
Assume that the random variables $\{\max_{t} \psi_0 ^{\trans}\phi(t,X) - \psi_0 ^{\trans}\phi(\tau, X)\}_{\tau\in \mcT}$ are almost surely bounded and each admits a density $f_\tau$ on $(0,c)$ for some constant $c>0$, that satisfies $f_\tau(x)\leq H/x^{1-\delta}$, for some $0\leq H<\infty$ and $0<\delta\leq 1$. Suppose that $\beta=\omega(n^{1/\{2(1+\delta)\}})$ and $\beta=o(n^{1/2})$. Moreover, suppose that for each split $l,l' \in \{1,2\}$, the nuisance estimates satisfy the rate conditions:
$$\|h^*-\hat{h}_{l'}^{(l)}\|_{2}, \|r^* - \hat{r}_{l'}^{(l)}\|_2, \|q^* - \hat q^{(l)}\|_2, \|p_1^* - \hat p_1^{(l)}\|_2=o_p(n^{-1/4}),$$ 
and that
$$\|\hat{p}_{2,\beta}^{(l)}-p_{2,\beta}^*\|_2=o_{p,\rm{unif}(\beta)}(n^{-1/4}).$$
Assume the nuisance estimates $\hat q^{(1)}(S), \hat p_1^{(1)}(S), \hat p_{2,\beta}^{(1)}(S),\hat q^{(2)}(S), \hat p_1^{(2)}(S), \hat p_{2,\beta}^{(2)}(S)$ are all almost surely bounded. Moreover, assume the following boundedness conditions:
$$\sup_{\beta>0}\|\theta_0^\beta\|_2<\infty$$
and assume the random variables
$$ \sum_\tau \|\phi(\tau,X)\|_2, \|\tilde{M}\|_2,\tilde{Y}, \sup_{\beta>0}\|p_{2,\beta}^*(S)\|_2, \|\mu(T_1,S)\|_2, \sup_{\tau\in \mathcal{T}, \psi \in \mathcal{N}} |\psi^{\trans} \phi(\tau, X)|$$
are all almost surely bounded, where $\tilde{M}= \mu(T_1,S) - p_1^*(S)$, $\tilde{Y}= Y - q^*(S)$. Assume that the matrix $\E(\tilde{M} \tilde{M} ^{\trans})$ is bounded and strictly positive definite so that its inverse matrix $\E(\tilde{M} \tilde{M} ^{\trans})^{-1}$ exists. Then the estimates $\hat{\psi}, \hat{\theta}^\beta$ are asymptotically linear, i.e. for any $l\in\{1,2\}$
\begin{align*}
    n^{1/2}(\hat{\psi} - \psi_0) = n^{-1/2} \sum_{i=1}^n \rho_{\psi}(Z_i) + o_p(1), \quad
    n^{1/2}(\hat{\theta}^\beta - \theta_0) = n^{-1/2} \sum_{i=1}^n \rho_{\theta}(Z_i) + o_p(1).
\end{align*}
Moreover, asymptotically valid confidence intervals, with target coverage level $\alpha$, can be constructed via any consistent estimates $\hat{\sigma}_\psi^2, \hat{\sigma}_\theta^2$ of the variances $\sigma_\psi^2 = \E\{\rho_\psi(Z)^2\}$ and $\sigma_\theta^2 = \E\{\rho_\theta(Z)^2\}$ as:
\begin{align*}
    CI_{\psi}(\alpha) =~& [\hat{\psi} \pm n^{-1/2}z_{1-\alpha/2} \hat{\sigma}_{\psi}], &
    CI_{\theta}(\alpha) =~& [\hat{\theta}^\beta \pm n^{-1/2}z_{1-\alpha/2} \hat{\sigma}_{\theta}], 
\end{align*}
where $z_{q}$ is the $q$-th quantile of the standard normal distribution.
\end{theorem}
Analogous to Theorem~\ref{thm:main}, Theorem~\ref{thm:main2} will naturally follow from a combination of Lemma~\ref{lem:bias} and a variant of Corollary~\ref{thm:specific-linearity}, which we present below.

\subsection{Asymptotic Linearity Theorem Adapted to Cross-Fitting Approach}\label{sec: mainThmSampleSplitting} 
As was the case without cross fitting, we still
consider the following generalized method of moments framework:
$$
M(\theta, g, h;\beta) = \E_Z\{m_\beta(Z; \theta, g, h)\}, \quad
    M(\theta_0^{\beta_n}, g_0^{\beta_n}, h_0;\beta_n) = 0.
$$
Of note is the distinction that for the cross fitting approach, we split the entire dataset $S=[n]$ into two disjoint sets $S_1$ and $S_2$ where without loss of generality let $|S_1| = |S_2| =n/2,$  and we train $\hat h$ and $\hat g^\beta$ on a different split of the data from the one where we evaluate the empirical moments: for each split $l\in\{1,2\}$, we denote the corresponding estimators $\hat h$ and $\hat g^\beta$ of trained on $S_l$ as $\hat h^{(l)}$ and $\hat g^{\beta, (l)}$. Then consider the cross-fitted estimator $\hat\theta^\beta$ that satisfies
\begin{align*}
    n^{-1} \sum_{i\in S_1}m_\beta(Z_i;\hat\theta^\beta, \hat g^{\beta, (2)}, \hat h^{(2)}) +  n^{-1} \sum_{i\in S_2}m_\beta(Z_i;\hat\theta^\beta, \hat g^{\beta, (1)}, \hat h^{(1)}) = o_p(n^{-1/2}).
\end{align*}
We will keep assuming the original Assumptions~\ref{ass:influence}, \ref{ass:mbeta-smooth}, \ref{ass:ortho}, \ref{ass:regularity} since they do not involve any properties of estimators. We extend Assumption~\ref{ass:rates} and Assumption~\ref{ass:equicont} to the cross-fitting setting as follows:
\begin{assumption}[Rates for $g$, Cross-Fitted Version]\label{ass:rates_new}
Suppose that the nuisance estimates $\hat g^{\beta, (1)},\hat g^{\beta, (2)}\in \mathcal{G}$ satisfy:
\begin{align*}
\begin{aligned}
    \|\hat g^{\beta, (1)} - g_0^\beta\|^2_2~=~ \E_X\{\|\hat g^{\beta, (1)}(X) - g_0^\beta(X)\|_2^2\} &=o_{p,\rm{unif}(\beta)}(n^{-1/2}),\\
    \|\hat g^{\beta, (2)} - g_0^\beta\|^2_2~=~ \E_X\{\|\hat g^{\beta, (2)}(X) - g_0^\beta(X)\|_2^2\}&=o_{p,\rm{unif}(\beta)}(n^{-1/2}).
\end{aligned}
\end{align*}
\end{assumption}

\begin{assumption}[Equicontinuity, Cross-Fitted Version]\label{ass:equicont_new}
Suppose that $\beta$ grows at rate such that the moment $m$ satisfies the stochastic equicontinuity conditions: 
\begin{align*}
    \sqrt{n}\|A(\hat g^{\beta, (1)},h_0;\beta) - A(g_0^\beta, h_0;\beta)  - \{\sum_{i\in S_2} a_\beta(Z_i;\hat g^{\beta, (1)},h_0) 
    - \sum_{i\in S_2} a_\beta(Z_i;g_0^\beta,h_0)
    \}\|_{op} =~& o_{p}(1),\\
    \sqrt{n}\|V(\hat g^{\beta, (1)},h_0;\beta) - V(g_0^\beta, h_0;\beta)  - \{\sum_{i\in S_2} \nu_\beta(Z_i;\hat g^{\beta, (1)},h_0) 
    - \sum_{i\in S_2} \nu_\beta(Z_i;g_0^\beta,h_0)
    \}\|_{op} =~& o_{p}(1),\\
    \sqrt{n}\|A(\hat g^{\beta, (2)},h_0;\beta) - A(g_0^\beta, h_0;\beta)  - \{\sum_{i\in S_1} a_\beta(Z_i;\hat g^{\beta, (2)},h_0) 
    - \sum_{i\in S_1} a_\beta(Z_i;g_0^\beta,h_0)
    \}\|_{op} =~& o_{p}(1),\\
    \sqrt{n}\|V(\hat g^{\beta, (2)},h_0;\beta) - V(g_0^\beta, h_0;\beta)  - \{\sum_{i\in S_1} \nu_\beta(Z_i;\hat g^{\beta, (2)},h_0) 
    - \sum_{i\in S_1} \nu_\beta(Z_i;g_0^\beta,h_0)
    \}\|_{op} =~& o_{p}(1).\\
\end{align*} 
\end{assumption}
Then we present the following theorem:
\begin{theorem}
    Under Assumptions~\ref{ass:influence}, \ref{ass:mbeta-smooth}, \ref{ass:ortho}, \ref{ass:regularity} and  Assumptions~\ref{ass:rates_new}, \ref{ass:equicont_new}, if the nuisance parameter estimates $\hat{h}^{(1)}$ and  $\hat{h}^{(2)}$ are asymptotically linear with some influence function $f_h$ such that $\E\{f_h(Z_i)\} = 0$: for any split $l\in\{1,2\}$
\begin{align*}
    \hat h^{(l)} - h_0 = 2n^{-1}\sum_{i\in S_l} f_h(Z_i) + o_p(n^{-1/2}) 
\end{align*}
then the parameter estimate $\hat{\theta}^\beta$ is asymptotically linear around $\theta_0^\beta$:
\begin{align*}
    \sqrt{n} \left(\hat{\theta}^\beta - \theta_0^\beta\right) = n^{-1/2} \sum_{i=1}^n \rho_{\theta}(Z_i) + o_p(1)
\end{align*}
with influence function:
\begin{align*}
    \rho_{\theta}(Z) = -  A_*(g_0,h_0)^{-1} \{m_*(Z; \theta_0, g_0,h_0) + (J^{*})^{\trans}f_h(Z)\}.
\end{align*}
\end{theorem}

  In Section~\ref{sec:instant}, we verified that Assumptions~\ref{ass:influence}, \ref{ass:mbeta-smooth}, \ref{ass:ortho}, and \ref{ass:regularity} were satisfied if conditions of Theorem~\ref{thm:main2} are satisfied (noting the overlap between conditions of Theorem~\ref{thm:main2} and those of Theorem~\ref{thm:main}). Moreover, the nuisance rate conditions of Theorem~\ref{thm:main2} automatically translate to Assumption~\ref{ass:rates_new} in the general setting. In Section~\ref{subsec:verifystochcross}, we also present Lemma~\ref{lem: split_equi} which also shows that Assumption~\ref{ass:equicont_new} is satisfied when we adopt the cross-fitting approach. Hence, we obtain the following corollary:

\begin{corollary}\label{thm:specific-linearity2}
Suppose there is some influence function $\rho_\psi$ such that for each split $l\in\{1,2\}$, the estimator $\hat\psi^{(l)}$ satisfies
$$\hat\psi^{(l)} - \psi_0 = 2n^{-1}\sum_{i \in S_l}\rho_\psi(Z_i) + o_p(n^{-1/2}).$$
with $\E\{\rho_\psi(Z)\}=0$. Then under the conditions of Theorem~\ref{thm:main2}, the estimate $\hat{\theta}^\beta$ is asymptotically linear around $\theta_0^\beta$:
$$
    n^{1/2} (\hat{\theta}^\beta - \theta_0^\beta) = n^{-1/2} \sum_{i=1}^n \rho_{\theta}(Z_i) + o_p(1)
$$
with influence function:
\begin{align*}
    \rho_{\theta}(Z) =   \E(\tilde{M} \tilde{M}^{\trans})^{-1} \{m_*(Z; \theta_0, \psi_0, g_0) + J_*^{\trans}\rho_{\psi}(Z)\}
\end{align*}
where $m_*$ as defined in Equation~\eqref{eqn:specific-mstar} and $J_*$ as defined in Equation~\eqref{eqn:specific-Jstar}.
\end{corollary}

\subsection{Verifying Assumption~\ref{ass:equicont_new}: Cross-Fitting Implies Stochastic Equicontinuity}\label{subsec:verifystochcross}
In this section, we directly present a general lemma claiming that performing cross fitting would automatically imply a cross-fitted version of equicontinuity. Then it will naturally follow that our cross-fitting approach in Theorem~\ref{thm:main2} guarantees the cross-fitted version Assumption~\ref{ass:equicont_new}.

 \begin{lemma}\label{lem: split_equi}
Assume that the estimators are consistent for the sequence of $\beta=\beta_n$ that we choose:
\begin{align*}
\|\hat g^{\beta, (1)} - g_0^\beta\|_2 = o_p(1)\\
\|\hat g^{\beta, (2)} - g_0^\beta\|_2 = o_p(1).
\end{align*}
and that the nuisance estimates $\|\hat g^{\beta, (1)}(X)\|_2, \|\hat g^{\beta, (2)}(X)\|_2$ are almost surely bounded.
Assume the moment function satisfies mean-squared continuity condition: for all $j,k\le p$, for all $g,g'\in\mathcal{G}$
\begin{align*}
    \E\left[\left\{a_{\beta,j,k}(Z_i;g,h_0) - a_{\beta,j,k}(Z_i;g',h_0)\right\}^2\right]&\le L\cdot\|g-g'\|_2^q\\
    \E\left[\left\{\nu_{\beta,j}(Z_i;g,h_0) - \nu_{\beta,j}(Z_i;g',h_0)\right\}^2\right]&\le L\cdot\|g-g'\|_2^q
\end{align*}
 for some $q<\infty$ and $L>0.$
 Then Assumption~\ref{ass:equicont_new} holds. That is, the following stochastic equicontinuity statements hold:
\begin{align*}
\begin{aligned}
    n^{-1/2}\left\|\sum_{i\in S_2}\left[A(\hat g^{\beta, (1)},h_0;\beta)- A(g_0^\beta,h_0;\beta)-\{a_\beta(Z_i;\hat g^{\beta, (1)},h_0) - a_\beta(Z_i;g_0^\beta,h_0)\}\right]\right\|_{op} =~& o_{p}(1)\\
    n^{-1/2}\left\|\sum_{i\in S_1}\left[A(\hat g^{\beta, (2)},h_0;\beta)- A(g_0^\beta,h_0;\beta)-\{a_\beta(Z_i;\hat g^{\beta, (2)},h_0) - a_\beta(Z_i;g_0^\beta,h_0)\}\right]\right\|_{op} =~& o_{p}(1)\\
    n^{-1/2}\left\|\sum_{i\in S_2}\left[V(\hat g^{\beta, (1)},h_0;\beta)- V(g_0^\beta,h_0;\beta)-\{\nu_\beta(Z_i;\hat g^{\beta, (1)},h_0) - \nu_\beta(Z_i;g_0^\beta,h_0)\}\right]\right\|_{2} =~& o_{p}(1)\\
    n^{-1/2}\left\|\sum_{i\in S_1}\left[V(\hat g^{\beta, (2)},h_0;\beta)- V(g_0^\beta,h_0;\beta)-\{\nu_\beta(Z_i;\hat g^{\beta, (2)},h_0) - \nu_\beta(Z_i;g_0^\beta,h_0)\}\right]\right\|_{2} =~& o_{p}(1).
\end{aligned} 
\end{align*}
\end{lemma}
\begin{proof}
Without loss of generality, we only prove the first statement here, and similar proofs should apply to when we switch data splits and apply to function $\nu$. For the first statement, it suffices to show that for any $j,k\le p$
    \begin{align*}
        n^{-1/2}\left\|\sum_{i\in S_2}\left[A_{j,k}(\hat g^{\beta, (1)},h_0;\beta)- A_{j,k}(g_0^\beta,h_0;\beta)-\{a_{\beta,j,k}(Z_i;\hat g^{\beta, (1)},h_0) - a_{\beta,j,k}(Z_i;g_0^\beta,h_0)\}\right]\right\|_{L_2} =~ o(1)
    \end{align*}
In the remainder of the proof we look at a particular $(j,k)$ and hence for simplicity we
overload notation and we let $a_\beta = a_{\beta,j,k}$ and $A = A_{j,k}.$
Moreover, for simplicity we denote for each $i\in S_2:$
\begin{align*}
K_i= A(\hat g^{\beta, (1)},h_0;\beta)- A(g_0^\beta,h_0;\beta)-\{a_\beta(Z_i;\hat g^{\beta, (1)},h_0) - a_\beta(Z_i;g_0^\beta,h_0)\}
\end{align*}
Then if we take the square, for any $i,j\in S_2,$
\begin{align*}
    \E\left\{\left(n^{-1/2}\sum_{i\in S_2}K_i\right)^2\right\} = \frac{1}{2}\E(K_i^2) + \frac{n-2}{4}\E(K_iK_j)
\end{align*}
the first term can be bounded using mean-squared-continuity:
\begin{align*}
    \frac{1}{2}\E(K_i^2)&=\frac{1}{2}\E\left(\left[A(\hat g^{\beta, (1)},h_0;\beta)- A(g_0^\beta,h_0;\beta)-\{a_\beta(Z_i;\hat g^{\beta, (1)},h_0) - a_\beta(Z_i;g_0^\beta,h_0)\}\right]^2\right)\\
    &\le \E\left[\left\{A(\hat g^{\beta, (1)},h_0;\beta)- A(g_0^\beta,h_0;\beta)\right\}^2\right] + \E\left[\left\{a_\beta(Z_i;\hat g^{\beta, (1)},h_0) - a_\beta(Z_i;g_0^\beta,h_0)\right\}^2\right]\\
    &\le2\E\left[\left\{a_\beta(Z_i;\hat g^{\beta, (1)},h_0) - a_\beta(Z_i;g_0^\beta,h_0)\right\}^2\right]\le 2L\cdot\E[\|\hat g^{\beta, (1)} - g_0^\beta\|_2^q] = o(1),
\end{align*}
where the second to last inequality exploits Jensen's inequality, and the last equality exploits dominated convergence theorem.
Moreover, the second term can also be simplified using tower law:
\begin{align*}
    \E(K_iK_j)&= \E\{\E(K_iK_j \mid \{Z_l\}_{l\in S_1})\} = \E\big\{\E(K_i \mid \{Z_l\}_{l\in S_1})\E(K_j \mid \{Z_l\}_{l\in S_1})\big\} = 0,
\end{align*}
since conditional on the first data split, $K_i$ and $K_j$ are independent with zero means. Hence, we have indeed proved that
\begin{align*}
        n^{-1/2}\left\|\sum_{i\in S_2}\left[A(\hat g^{\beta, (1)},h_0;\beta)- A_(g_0^\beta,h_0;\beta)-\{a_{\beta}(Z_i;\hat g^{\beta, (1)},h_0) - a_{\beta}(Z_i;g_0^\beta,h_0)\}\right]\right\|_{L_2} =~ o(1).
    \end{align*}
\end{proof}

\subsection{Proof of Theorem~\ref{thm:main2}}\label{sec: pf_main2}
\begin{proof}
    The proof of Lemma 16 of \cite{chernozhukov2020adversarial} establishes both asymptotic linearity of $\hat\psi$ and that of $\hat\psi^{(l)}$ for each $l\in\{1,2\}$:
    \begin{equation}\label{eq: an3}
    \sqrt{n}(\hat\psi-\psi_0) = n^{-1/2}\sum_{i=1}^n\rho_\psi(Z_i) + o_p(1)
    \end{equation}
    \begin{equation}\label{eq: an4}
    \hat\psi^{(l)}-\psi_0 = 2n^{-1}\sum_{i\in S_l}\rho_\psi(Z_i) + o_p(n^{-1/2})
    \end{equation}
    where if we define $\tilde{P}=\phi(T_2, X) - \E\{\phi(T_2, X)\mid X\}$:
    $$\rho_\psi(Z)=\E(\tilde{P}\tilde{P}^{\trans})^{-1}\{Y-\E(Y\mid X) - \psi_0^{\trans}\tilde{P}\}\tilde{P}.$$
    Moreover, Lemma~\ref{lem:bias} gives us the softmax bias control that
    $$\sqrt{n}(\theta_0^\beta-\theta_0) = o(1),$$
    and Corollary~\ref{thm:specific-linearity2} gives us asymptotic linearity of $\hat{\theta}$ around $\theta_0^\beta$:
    \begin{align*}
    \sqrt{n} \left(\hat{\theta}^\beta - \theta_0^\beta\right) = n^{-1/2} \sum_{i=1}^n \rho_{\theta}(Z_i) + o_p(1)
\end{align*}
where
\begin{align*}
    \rho_{\theta}(Z) =   \E(\tilde{M} \tilde{M}^{\trans})^{-1} \{m_*(Z; \theta_0, \psi_0, g_0) + J_*^{\trans}\rho_{\psi}(Z)\}.
\end{align*}
Combining these two results, we overall can conclude that
\begin{equation}\label{eq: an5}
    \sqrt{n} \left(\hat{\theta}^\beta - \theta_0\right) = n^{-1/2} \sum_{i=1}^n \rho_{\theta}(Z_i) + o_p(1).
\end{equation}
We can hence use the two asymptotic linearity statements \eqref{eq: an3} and \eqref{eq: an5} to construct confidence intervals for $\psi_0$ and $\theta_0.$
\end{proof}

\section{Supplementary Material for Experimental Section}
\label{app:experiments}

We justify here why the linear blip function specifications used in the experimental section of the main paper are equal or good approximations of the true blip functions. 
The final period blip effect $\gamma_2$ is exactly of the form:
\begin{align*}
    \gamma_2(T_2, X) =~& \alpha_2 (X_1 + 1) T_2,  &
    \phi(T_2, X) =~& (T_2, X_1 T_2), &
    \psi_0 = (\alpha_2, \alpha_2)
\end{align*}
and the optimal policy in the final period when $\alpha_2 \geq 0$ is $\pi_2^*(X) = \mathbbm{1}_{\left\{X_1 + 1 \geq 0\right\}}$. 
The blip effect in the first period, under the optimal continuation policy is of the form:
\begin{align*}
    \gamma_1(T_1, S) =~& \alpha_2 \left(\E_{Z\sim N(\alpha_1 T_1 + S_1 + 1, 1)}\left[Z\, \mathbbm{1}_{\left\{Z\geq 0\right\}}
    \right] - \E_{Z\sim N(S_1 + 1, 1)}\left[Z\, \mathbbm{1}_{\left\{Z\geq 0\right\}}\right]\right) + \alpha_1 T_1.
\end{align*}
Define the function $g(x) = \E_{Z\sim N(x, 1)}[Z \mathbbm{1}_{\left\{Z\geq 0\right\}}]$. Since $T_1$ is binary, we can write:
\begin{align*}
    \gamma_1(T_1, S) =~& \alpha_2 \left(g(\alpha_1 T_1 + S_1 + 1) - g(S_1 + 1)\right) + \alpha_1  T_1\\
    =~& \alpha_2 T_1 \left(g(\alpha_1 + S_1 + 1) - g(S_1 + 1)\right) + \alpha_1  T_1.
\end{align*}
For any $\alpha_1\in [0, 1]$, since $S_1\sim N(0, 1)$, we have that with high probability $\alpha_1 + S_1 + 1\in [-1, 4]$ and $S_1 + 1\in [-1, 3]$. In the regime $[-1, 4]$ the function $g(x)$ can be well approximated by a quadratic function (see Figure~\ref{fig:gx_approx}), i.e. 
\begin{align*}
    g(x) \approx \tilde{g}(x) := \kappa_0 + \kappa_1 x + \kappa_2 x^2.
\end{align*}
\begin{figure}[H]
\centering
    \includegraphics[scale=.4]{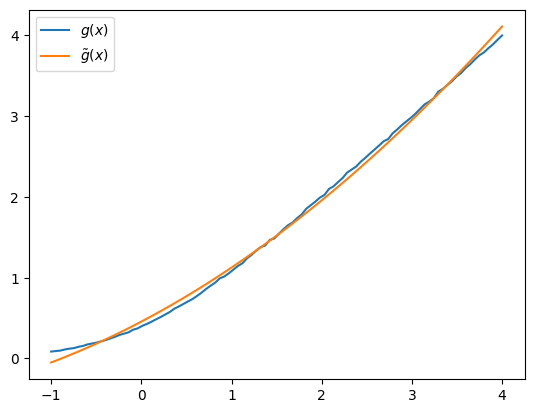}
    \caption{Approximation of the function $g(x)$ by $\tilde{g}(x)$.}\label{fig:gx_approx}
\end{figure}
Thus we then have:
\begin{align*}
\gamma_1(T_1, S) \approx~& \alpha_2 T_1 \left(\tilde{g}(\alpha_1 + S_1 + 1) - \tilde{g}(S_1 + 1)\right) + \alpha_1  T_1.
\end{align*}
Moreover:
\begin{align*}
    \tilde{g}(\alpha_1 + S_1 + 1) - \tilde{g}(S_1 + 1) =~& \kappa_1 \alpha_1 + \kappa_2 ((\alpha_1 + S_1 + 1)^2 - (S_1 + 1)^2)\\
    =~& \kappa_1 \alpha_1 + \kappa_2 \alpha_1 (\alpha_1 + 2 S_1 + 2).
\end{align*}
Overall, we can write:
\begin{align*}
    \gamma_1(T_1, S) \approx \alpha_2 T_1 (\kappa_1 \alpha_1 + \kappa_2 \alpha_1 (\alpha_1 + 2 S_1 + 2)) + \alpha_1 T_1.
\end{align*}
The expression on the right hand side is of the form $\theta_0 T_1 + \theta_1 T_1\, S_1$, for appropriately defined parameters $\theta_0, \theta_1$.
\end{document}